\documentclass{article}

\usepackage[margin = 1in]{geometry}

\usepackage{amssymb,amsmath,amsthm,psfrag,epsfig}
\usepackage{amssymb,amsmath,hyperref,color}
\hypersetup{colorlinks=true,citecolor=blue, linkcolor=blue, urlcolor=blue}

\newtheorem{theorem}{Theorem}
\newtheorem{lemma}{Lemma}
\newtheorem{claim}[lemma]{Claim}

\newtheorem{remark}{Remark}
\newtheorem{condition}{Condition}

\newcommand{\G}{\mathcal{G}}

\newcommand{\E}{\mathbb{E}}
\newcommand{\B}{B}

\newcommand{\TD}{\mathbb{T}_{\Delta}}
\def\eps{\varepsilon}
\def\epsilon{\varepsilon}
\newcommand{\TDary}{\hat{\mathbb{T}}_{\Delta}}

\newcommand{\fptas}{\mathsf{FPTAS}}
\newcommand{\fpras}{\mathsf{FPRAS}}
\newcommand\Tree{\mathbb{T}}

\newcommand{\norm}[1]{\lVert#1\rVert}

\begin{document}

\title{Inapproximability of the Partition Function for the Antiferromagnetic Ising 
and Hard-Core Models}

\author{Andreas Galanis\thanks{
  Department of Computer Science, University of Oxford, Wolfson Building, Parks Road, Oxford, OX1~3QD, UK. Email: agalanis@cs.ox.ac.uk.
  The research leading to these results has received funding from the European Research Council under
  the European Union's Seventh Framework Programme (FP7/2007-2013) ERC grant agreement no.\ 334828. The paper
  reflects only the authors' views and not the views of the ERC or the European Commission.
  The European Union is not liable for any use that may be made of the information contained therein.}
\and Daniel \v{S}tefankovi\v{c}\thanks{Department of Computer Science, University of Rochester,
Rochester, NY 14627.  Email: stefanko@cs.rochester.edu. Research supported in part by NSF grant CCF-0910415.}
\and Eric Vigoda\thanks{School of Computer Science, Georgia
Institute of Technology, Atlanta GA 30332.
Email: vigoda@cc.gatech.edu. Research supported in part by NSF grant CCF-1217458.}
}

\maketitle

\begin{abstract}
Recent inapproximability results of Sly (2010), together with
an approximation algorithm presented by Weitz (2006) establish
a beautiful picture for the computational complexity of approximating
the partition function of the hard-core model.  Let $\lambda_c(\Tree_\Delta)$
denote the critical activity for the hard-model on the infinite $\Delta$-regular
tree.  Weitz presented an $\fptas$ for the partition function when 
$\lambda<\lambda_c(\Tree_\Delta)$ for graphs with constant maximum degree $\Delta$.
In contrast, Sly showed that for all $\Delta\geq 3$, there exists $\eps_\Delta>0$ such
that (unless $RP=NP$) there is no $\fpras$ for approximating the partition function
on graphs of maximum degree $\Delta$ for activities $\lambda$ satisfying
$\lambda_c(\Tree_\Delta)<\lambda<\lambda_c(\Tree_\Delta)+\eps_\Delta$.

We prove that a similar phenomenon holds for the antiferromagnetic Ising model.
Sinclair, Srivastava, and Thurley (2014) extended Weitz's approach to 
the antiferromagnetic Ising model,
yielding an $\fptas$ for the partition function 
for all graphs of constant maximum degree $\Delta$ when the
parameters of the model lie in the uniqueness region of the infinite $\Delta$-regular tree.
We prove the complementary result for the antiferrogmanetic Ising model without
external field, namely, that
unless $RP=NP$, for all $\Delta\geq 3$, 
there is no $\fpras$ for approximating the partition function
on graphs of maximum degree $\Delta$ when the inverse temperature lies
in the non-uniqueness region of the infinite tree $\Tree_\Delta$.
Our proof works by relating certain second moment calculations for
random $\Delta$-regular bipartite graphs to the tree recursions used to establish
the critical points on the infinite tree.  
\end{abstract}

\section{Introduction}

A remarkable computational transition has recently been established 
for the complexity of approximating the partition
function of the hard-core model.  Amazingly, this computational  
transition coincides exactly with the statistical physics phase
transition on infinite $\Delta$-regular trees.  
In this work we prove that a similar phenomenon holds for a much more
general class of 2-spin models.

Our general setup is 2-spin systems on an input graph $G=(V,E)$
with maximum degree $\Delta$. We follow the setup of several related previous works \cite{GJP,SST,LLY}.
Configurations of the system are assignments $\sigma: V\rightarrow\{-1,+1\}$. For a given 
 configuration $\sigma\in \{-1,+1\}^V$, denote by $n_{-}(\sigma)$ the 
 number of vertices assigned $-1$, $m^{-}(\sigma)$ the number of 
 edges with both endpoints assigned $-1$ and $m^{+}(\sigma)$ the number of 
 edges with both endpoints assigned $+1$. 
 
 There are three non-negative parameters $B_1,B_2$ and $\lambda$ where 
 $B_1$ is the edge activity for $(-,-)$ edges, 
 $B_2$ is the edge activity for $(+,+)$ edges, and 
 $\lambda$ is the vertex activity (or external field).
 A configuration $\sigma\in \{-1,+1\}^V $ has 
 weight 
\[w_G(\sigma)=\lambda^{n_{-}(\sigma)} \B_1^{m^{-}(\sigma)}\B_2^{m^{+}(\sigma)}.\] 
The Gibbs distribution $\mu_G(\cdot)$ is over the set $\{-1,+1\}^V$ 
where $\mu_G(\sigma)=w_G(\sigma)/Z$, where
$Z=Z_G(B_1,B_2,\lambda)$ is a normalizing factor, known as the partition function, defined as:
\[ Z  := \sum_{\sigma\in \{-1,+1\}^V} w_G(\sigma)=\sum_{\sigma\in \{-1,+1\}^V}\lambda^{n_{-}
(\sigma)} \B_1^{m^{-}(\sigma)}\B_2^{m^{+}(\sigma)}.\]

Let us point out a few examples of 2-spin models that are of particular interest.
Setting $B_1=0$ and $B_2=1$, the only configurations with 
positive weight in the Gibbs distribution induce independent sets of $G$.  This case is known as the
{\em hard-core model} with fugacity $\lambda$.
Setting $B_1=B_2=B$,
one recovers the Ising model.  The case $\lambda=1$ is the Ising model without external field. 
When $B<1$ it is the antiferromagnetic Ising model, and when $B\geq1$ it is the
ferromagnetic Ising model.  In general, when $B_1B_2<1$ the model is called {\em antiferromagnetic},
and when $B_1B_2\geq1$ it is ferromagnetic.

The Ising and hard-core models (and, more generally, spin models) have been studied thoroughly in various contexts, especially 
in statistical physics. Their interpretation as idealized models of microscopic interaction within a body have nurtured research on understanding their macroscopic behavior on appropriate underlying graph structures, such as the infinite grid or the infinite triangular lattice. The intriguing questions arising have made an impact on other fields as well, most notably probability and computer science. For illuminating accounts of such directions, we refer the reader to \cite{Mezard} and \cite{Talagrand}. 

In this paper, we focus on the computational complexity of approximating the partition function in 2-spin models. A long series of works, which we only touch upon later in this introduction, have identified the connection of the problem with decay of correlation properties in Gibbs measures of the infinite $\Delta$-regular tree. In statistical physics terms, this is known as the phase transition on the Bethe lattice. To briefly sketch this important concept, let $\TD$ denote the infinite $\Delta$-regular tree. When the parameters of the 2-spin model lie in the so-called uniqueness region of $\Tree_\Delta$, 
 for finite  complete trees of height $h$, the influence on the root from fixing a configuration on the leaves dies off in the limit as the height goes to infinity. In sharp contrast, in the non-uniqueness region there exist configurations of the leaves for which their influence on the root persists in the limit. For convenience, we often refer to the uniqueness region omitting the reference to 
$\Tree_\Delta$.

 The uniqueness/non-uniqueness regions can be determined rather easily on $\Tree_\Delta$ (for other infinite graphs this is far from easy, see \cite{Beffara} and \cite{RSTVY} for recent advances in the two dimensional grid). For the hard-core model, Kelly \cite{Kelly} showed non-uniqueness holds iff $\lambda>\lambda_c(\Tree_\Delta):= \frac{(\Delta-1)^{\Delta-1}}{(\Delta-2)^\Delta}$. For the antiferromagnetic Ising model without external field, non-uniqueness holds iff $B<B_c(\Tree_\Delta):=\frac{\Delta-2}{\Delta}$ (see, e.g., \cite{SST}). For general antiferromagnetic 2-spin models with soft constraints
(i.e., $B_1B_2>0$), non-uniqueness holds iff $\sqrt{B_1B_2}<(\Delta-2)/\Delta$ and $\lambda\in (\lambda_1,\lambda_2)$ for some critical values $\lambda_1(\Delta,B_1,B_2),$ $\lambda_2(\Delta,B_1,B_2)$ (see, e.g., \cite{LLY}). 

Before stating our results, it will be instructive to review some known results from the literature. When $\lambda=1$, the partition function is the number of independent  sets in $G$, and this quantity is \#P-complete to compute exactly \cite{Valiant},  even when $\Delta=3$ \cite{Greenhill}. For the Ising model, once again, exact computation of the partition function 
for graphs of maximum degree $\Delta$ is \#P-complete, even for the ferromagnetic model \cite{JS}. Hence, the focus is on the computational complexity of approximating the partition function.

This is where the phase transition on $\Tree_\Delta$ comes into play. Specifically, for the hard-core model, recently, a deep connection between uniqueness/non-uniqueness and the complexity of approximating the partition function was established. Namely, for graphs of constant maximum degree $\Delta$, Weitz \cite{Weitz} presented a beautiful $\fptas$ for approximating the partition function when $\lambda<\lambda_c(\Tree_\Delta)$. On the other hand, Sly \cite{Sly10} presented a clever reduction which proves that
 for every $\Delta\geq 3$, there exists $\eps_\Delta>0$, such that  it is NP-hard (unless RP=NP) to approximate the partition function for graphs of maximum degree $\Delta$ for any $\lambda$ where $\lambda_c(\Tree_\Delta)<\lambda<\lambda_c(\Tree_\Delta)+\eps_\Delta$.
 
It was believed that Sly's inapproximability result should hold for all $\lambda>\lambda_c(\Tree_\Delta)$.
However, Sly's work utilized results of Mossel, Weitz and Wormald \cite{MWW}
which showed that for a random $\Delta$-regular bipartite graph, an independent set
chosen from the Gibbs distribution is ``unbalanced'' with high probability
for $\lambda$ under the same condition as in Sly's result.  

\cite{MWW} used a technically complicated second moment argument that is characterized in \cite{Sly10} as a ``technical tour-de-force''.  
In \cite{Galanis} we extended the results of \cite{MWW} to 
all $\lambda>\lambda_c(\Tree_\Delta)$
for $\Delta=3$ and $\Delta\geq 6$.  However, that work built upon \cite{MWW}
and had an even more complicated analysis.  
In this work, as a byproduct of the approach we devise for
coping with the Ising model, we can prove inapproximability for the
hard-core model for all $\lambda>\lambda_c(\Tree_\Delta)$
for $\Delta=3,4,5$, and our proof is simpler than in \cite{Galanis}.  
This settles the picture for the hard-core model,
and also proves Conjecture 1.2 of \cite{MWW}.

\begin{theorem}
\label{thm:hardcore}
Unless $NP=RP$, for the hard-core model, for all $\Delta\geq 3$, for all $\lambda$ in the 
non-uniqueness region of the infinite tree $\Tree_\Delta$, 
there does not exist an $\fpras$ for the partition function at activity $\lambda$
for graphs of maximum degree at most $\Delta$.
\end{theorem}

Our main focus in this paper is to address whether the above phenomenon for the
hard-core model occurs for other models.  In particular, our goal is to address for
2-spin antiferromagnetic models whether
the computational complexity of
approximating the partition function for  graphs with maximum degree $\Delta$
undergoes a transition that coincides exactly with the uniqueness/non-uniqueness
phase transition on the infinite tree $\Tree_\Delta$.

Let us first review existing results in the literature.  
For the ferromagnetic Ising model, Jerrum and Sinclair \cite{JS}
presented an $\fpras$ for all graphs, for all $B\geq 1$ where $B=B_1=B_2$.  
For the antiferromagnetic Ising model, for general graphs there does not exist an $\fpras$ for $B<1$ \cite{JS}, unless $NP=RP$.  
For constant $\Delta$, Sinclair, Srivastava, and Thurley \cite{SST} extend Weitz's approach 
to the antiferromagnetic Ising model (see also, \cite{ZLB,RSTVY,LLY} for other
results on the Ising model), yielding in the uniqueness region an $\fptas$ for the partition function
of graphs with maximum degree $\Delta$.
 
 We prove for the antiferromagnetic Ising model without external field
 that for all $B$ in the non-uniqueness region of the infinite tree
  there does not exist an $\fpras$ for the antiferromagnetic Ising model.
 
\begin{theorem}\label{thm:Ising}
Unless $NP=RP$, for the antiferromagnetic Ising model without external field,
for all $\Delta\geq 3$, 
for all  $B$ in the non-uniqueness region of the infinite tree $\Tree_\Delta$,
there does not exist an $\fpras$ for the partition function 
at inverse temperature $B$
for graphs of maximum degree at most~$\Delta$.
\end{theorem}

The hard-core model and Ising model are the two most well-studied examples
 of 2-spin systems. For ferromagnetic 2-spin systems with no external field, i.e., $B_1B_2>1$ and $\lambda=1$, 
  Goldberg, Jerrum, and Paterson \cite{GJP} presented an $\fpras$ for the partition function for any graph.  
   For any antiferromagnetic 2-spin model, for constant $\Delta$, an $\fptas$ for the partition
 function was obtained by   Sinclair, Srivastava, and Thurley \cite{SST} for $\Delta$-regular
 graphs in the uniqueness region of the infinite tree $\Tree_\Delta$,
 and by Li, Lu, and Yin \cite{LLY} for graphs of maximum degree $\Delta$ in the
 intersection of the uniqueness regions for infinite trees $\Tree_d$ for $d\leq\Delta$.
  For antiferromagnetic  2-spin models we obtain  a complementary inapproximability result in the non-uniqueness region.
  Our results do not reach the uniqueness/non-uniqueness threshold in general.

\begin{theorem}\label{thm:2spin}
Unless $NP=RP$, for all $\Delta\geq 3$, for all $B_1,B_2,\lambda$ that 
lie in the non-uniqueness region of the infinite tree $\Tree_{\Delta}$ and
$\sqrt{B_1B_2} \geq \frac{\sqrt{\Delta-1}-1}{\sqrt{\Delta-1}+1}$,
there does not exist an $\fpras$ for the partition function  
at parameters $B_1,B_2,\lambda$
for graphs of maximum degree at most $\Delta$.
\end{theorem}

 \subsection{Independent Results of Sly and Sun}
 In an independent and simultaneous work, Sly and Sun \cite{SlySun} obtained
 closely related results.  They prove inapproximability of the partition
 function in the case of the hard-core model  and the antiferromagnetic Ising model.  Their
 result for the Ising model also covers the case of an external field, which then extends to general 2-spin antiferromagnetic models \cite{SST,SlySun}.

The main technical result in our proof is a second moment argument to 
analyze the partition function for the Ising model on random $\Delta$-regular bipartite graphs,
as outlined in the following Section \ref{sec:proof-approach}.
As a consequence we can estimate the partition function within any arbitrarily small
polynomial factor, which allows us to use the reduction in \cite{Sly10}.
In contrast, Sly and Sun's approach builds upon the recent interpolation scheme
 of Dembo, Montanari, and Sun \cite{DMS} to 
 analyze the logarithm of the partition function, which yields estimates
of the partition function within an arbitrarily small exponential factor. 
 To get their NP-hardness results, they use a modification of the approach
 appearing in \cite{Sly10} which allows to use a constant-sized gadget. 
 
 \subsection{Further work}
 
 Since the original publication of this work on arXiv \cite{GSV:arxiv} there has
been further progress on this topic.
In~\cite{GSV-STOC} we present a general approach for
analyzing the second moment of spin systems on random regular bipartite graphs
using induced matrix norms.  As a byproduct, we obtain hardness
results for all 2-spin antiferromagnetic systems in the tree non-uniqueness
region as in \cite{SlySun}, and in addition hardness results for approximately counting
$k$-colorings for even $k<\Delta$ and for the antiferromagnetic Potts model in
the conjectured tree non-uniqueness region. 
 
We point out that the present high-level approach to analyze the second moment is different than the one in \cite{GSV-STOC}. Here, we give an analysis of the critical points of the second moment in the case of the antiferromagnetic Ising model and the hard-core model; this analysis is partly possible because the spins take binary values. As a consequence, in certain cases we obtain stronger results than those needed for the desired hardness results (see for example the discussion in Section~\ref{sec:remarks}).   In contrast, in \cite{GSV-STOC} the approach is targeted to the maximizers of the second moment (which is sufficient for the desired hardness results) and connecting them to the maximizers of the first moment.

We should also remark that, in \cite{GSV-STOC}, the  gadget, which is a random regular bipartite graph, has fixed size as in \cite{SlySun}.
In contrast, the gadget studied here, which is Sly's original gadget in \cite{Sly10},
allows to quantify certain properties in terms of its size (see Lemma~\ref{lem:gadget2}) which is desirable in approximation-preserving reductions. For example, the results in this
paper combined with those in \cite{GSV-STOC} were used in \cite{CGGGJSV}
to prove \#BIS-hardness on bipartite graphs
for 2-spin antiferromagnetic systems in the tree non-uniqueness region.

\section{Proof Outline}
\label{sec:proof-approach}

The main element of the proof of 
Theorem \ref{thm:2spin} (and similarly, the proofs of
Theorems \ref{thm:hardcore} and \ref{thm:Ising})
is to analyze random $\Delta$-regular bipartite graphs to show
a certain bimodality as in \cite{MWW} for the
hard-core model.
From there the same reduction of \cite{Sly10} applies
with some non-trivial modifications in the analysis.  
Hence, in some sense our main technical result is to show that 
in random $\Delta$-regular bipartite graphs, under the
conditions of Theorem \ref{thm:2spin},
a configuration selected from the Gibbs distribution is ``unbalanced'' with high probability.
 
 To formally state our results, let us introduce some notation. Let $\G(n,\Delta)$
denote the probability distribution over bipartite graphs with $n+n$
vertices formed by taking the union of $\Delta$ random perfect matchings. 
We will denote the two sides of the bipartition of the graphs as $V_1,V_2$. 
Strictly speaking, this distribution is over bipartite multi-graphs.  
However, since our results hold asymptotically almost surely (a.a.s.{}) over
$\G(n,\Delta)$, as noted in \cite{MWW}, 
by contiguity arguments they also hold a.a.s.\  for the uniform distribution over bipartite
$\Delta$-regular graphs. For a complete account of contiguity, we refer the reader to \cite[Chapter 9]{JLR}.

For a graph $G\sim\G(n,\Delta)$ and a given configuration $\sigma\in \{-1,+1\}^V$, denote by $N_{-}(\sigma)$ the set of vertices assigned $-1$ and by $N_{+}(\sigma)$ 
the set of vertices assigned $+1$.  For $\alpha,\beta\geq0$, let
\[ \Sigma^{\alpha,\beta} = \{\sigma\in\{-1,+1\}^V \mid  |N_{-}(\sigma) \cap V_1| = \alpha n,\ |N_{-}(\sigma) \cap V_2| = \beta n\},
\]
that is, $\alpha$ (resp. $\beta$) is the fraction of the vertices 
in $V_1$ (resp. $V_2$) whose spin is $-1$. Denote by $\Sigma^{Bal}$ the set of all balanced configurations, and for $\rho>0$, denote by $\Sigma^\rho$ the set of $\rho n$-unbalanced configurations. Formally,
\[\Sigma^{Bal}=\mathop{\bigcup}_\alpha\Sigma^{\alpha,\alpha}, \quad \Sigma^\rho=\mathop{\bigcup}_{|\alpha-\beta|\geq \rho}\Sigma^{\alpha,\beta}.\]
The following theorem, which is proved later in this section, establishes  that in the Gibbs distribution of a random $\Delta$-regular bipartite graph, balanced configurations have exponentially small measure. 
\begin{theorem}\label{thm:bimodality}
Under the hypotheses of Theorem \ref{thm:2spin},
there exist constants $a>1$ and $\rho>0$,
 such that asymptotically almost surely, for a graph $G$ sampled from $\G(n,\Delta)$, the Gibbs distribution $\mu_G$ satisfies:
\[\mu_G\left(\Sigma^{Bal}\right)\leq a^{-n}\cdot \mu_G(\Sigma^\rho).\]
Therefore, the Glauber dynamics is torpidly mixing.
\end{theorem}
While Theorem~\ref{thm:bimodality} is interesting on its own right, our inapproximability results require a far more precise quantification of the bimodality than the one given in Theorem~\ref{thm:bimodality}. Nevertheless, it is instructive to outline the proof of  Theorem~\ref{thm:bimodality}, since it will allow us  to introduce the key ingredients that are needed for the hardness results. 

Namely, to establish Theorem~\ref{thm:bimodality}, we use a second moment approach as in \cite{MWW}. For $G\sim \G(n,\Delta)$ we analyze the random variable $Z^{\alpha,\beta}_G$, where 
\[Z^{\alpha,\beta}_G = \sum_{\sigma \in \Sigma^{\alpha,\beta}} w_G(\sigma),\] for  some well-suited values of $\alpha,\beta$. In particular, the $\alpha,\beta$ will be selected to maximize the first moment of $Z^{\alpha,\beta}_G$. One of the key ingredients in obtaining Theorem~\ref{thm:bimodality} is that the maximum of the first moment occurs for $\alpha\neq\beta$. More generally, this is  true for general antiferromagnetic 2-spin systems in the non-uniqueness region of the infinite tree $\TD$ (see the upcoming Lemma~\ref{lem:firstmomentmax}). The proof of this fact is based on an explicit connection with three specific Gibbs measures in the infinite tree $\TD$ which are invariant under parity-preserving transformations. We next introduce these relevant Gibbs measures, deferring a more technical discussion to Section \ref{sec:tree-recursions}.

Let $B_1,B_2,\lambda$ specify an antiferromagnetic 2-spin system. There exists a unique translation invariant Gibbs measure $\mu^*$ on $\TD$, 
known as the free measure, whose marginal probability for the root $\rho$ being assigned spin $-1$ will be denoted by
$p^* := \mu^*(\sigma_\rho=-1)$ (the exponent denotes the boundary condition; $*$ refers to free boundary).
 In the non-uniqueness region,
there also exist two semi-translation invariant measures 
$\mu^+$ and $\mu^-$, corresponding to the ``even" and ``odd" boundary conditions, respectively. These measures 
can be obtained by conditioning the leaves on level $2\ell$ (resp. $2\ell + 1$) of the 
tree to have spin $-1$ and taking the weak limit as $\ell\rightarrow\infty$. We will denote the marginals for the root being assigned the spin $-1$ by $p^+ := \mu^+(\sigma_\rho)$, $p^-:=\mu^-(\sigma_\rho)$. In the non-uniqueness region of $\TD$, it holds that $\mu^{*}\neq \mu^{\pm}$; in fact, it holds that $p^-<p^*<p^+$. 

The following lemma illustrates the relevance of $p^-,p^*,p^+$ in our arguments.
\begin{lemma}\label{lem:firstmomentmax}
Let $\Delta\geq 3$. For the distribution $\G(n,\Delta)$, $\displaystyle\lim_{n\rightarrow\infty}\frac{1}{n}\log\E_\G[Z^{\alpha,\beta}_G]$ is maximized for:
\begin{enumerate}
\item $(\alpha,\beta)=(p^*,p^*)$, whenever $B_1,B_2,\lambda$ are in the uniqueness region of the infinite tree $\TD$.
\item $(\alpha,\beta)=(p^{\pm},p^{\mp})$, whenever $B_1,B_2,\lambda$ are in the non-uniqueness region of the infinite tree $\TD$.
\end{enumerate}
\end{lemma} 

Note that Lemma~\ref{lem:firstmomentmax} already yields a weak form of a bimodality in the Gibbs distribution on random $\Delta$-regular graphs: in the non-uniqueness region, the configurations with the largest contribution are unbalanced in expectation. Since $Z^{\alpha,\beta}_G$ is typically exponential in $n$, Markov's inequality implies an upper bound on the number of balanced configurations, which holds with high probability over the choice $G\sim \G(n,\Delta)$.

However, to get tight results, we need a strong form of this bimodality. To do this, we look more carefully at the random variable $Z^{\alpha,\beta}_G$ for $(\alpha,\beta)=(p^\pm,p^\mp)$. In particular, suppose that we are able to show that  $Z^{p^\pm,p^\mp}_G$ is within a polynomial multiplicative factor from its expectation. Lemma~\ref{lem:firstmomentmax} would then yield that, for $\epsilon>0$ and all $\alpha',\beta'$ satisfying $\norm{(\alpha',\beta')-(p^\pm,p^\mp)}\geq \epsilon$,  $Z^{\alpha',\beta'}_G$ is exponentially smaller than $Z^{p^\pm,p^\mp}_G$. Using arguments in \cite{DFJ} and \cite{MWW}, one can then easily deduce Theorem~\ref{thm:bimodality}. 

We will formalize the above argument shortly. Prior to that, let us remark  that the same idea underlies the reduction used by Sly \cite{Sly10} (which we will use to obtain Theorems~\ref{thm:hardcore},~\ref{thm:Ising} and~\ref{thm:2spin}). Namely, define the phase of a configuration of $G$ to be the bipartition which has the largest number of vertices assigned $-1$. With high probability, a configuration sampled from $\mu_G$ will be from the set $\Sigma^{p^{\pm},p^{\mp}}$ (roughly) and hence the phase will take binary values. Informally, this allows to view $G$ as a Boolean gadget (the precise construction of the gadget used by Sly \cite{Sly10} is given in Section~\ref{sec:NP-outline}). Consider now an arbitrary graph $H$ and replace each vertex of $H$ with a (distinct) copy of the graph $G$. Sly showed how to encode the edges of $H$ (i.e., make connections between the copies of $G$), so that the final graph has maximum degree $\Delta$ and its partition function is dominated by configurations where the phases of the copies of $G$ correspond to a maximum cut of $H$.  The quantification of this scheme involves hard work, but this has already been done in \cite{Sly10} (certain calculations do require rather lengthy modifications to account for general antiferromagnetic 2-spin systems; see Section~\ref{sec:NP-outline} for details).

Let us now return to the technical core of the argument, which requires analyzing the second moment of $Z^{\alpha,\beta}_G$ for $(\alpha,\beta)=(p^\pm,p^\mp)$. By symmetry, we may clearly focus on $(\alpha,\beta)=(p^+,p^-)$. Pick two configurations $\sigma_1,\sigma_2$ from the set $\Sigma^{p^+,p^-}$ (not necessarily distinct). We need variables $\gamma,\delta$ that capture the overlap of the configurations $\sigma_1,\sigma_2$. Formally, for $\gamma,\delta\geq 0$, let
\begin{multline*}
\Sigma^{\alpha,\beta}_{\gamma,\delta}=
\big\{(\sigma_1,\sigma_2)\mid \sigma_1,\sigma_2\in \Sigma^{\alpha,\beta}, |N_{-}(\sigma_1)\cap N_{-}(\sigma_2) \cap V_1|=\gamma n,\big.\\ 
\big.|N_{-}(\sigma_1)\cap N_{-}(\sigma_2) \cap V_2|=\delta n\big\},
\end{multline*}
and $Y^{\gamma,\delta}_G=\sum_{(\sigma_1,\sigma_2)\in\Sigma^{\alpha,\beta}_{\gamma,\delta}}w_G(\sigma_1)w_G(\sigma_2).$ We will study $Y^{\gamma,\delta}_G$ for $(\alpha,\beta)=(p^+,p^-)$, so this is why we dropped the dependence on $\alpha,\beta$. Observe that $\E_\G[Y^{\gamma,\delta}_G]$ is the contribution to the second moment $\E_\G[(Z^{\alpha,\beta}_G)^2]$ coming from pairs of configurations in $\Sigma^{\alpha,\beta}$ with overlap $\gamma n,\delta n$.

In our case, for the second moment approach to succeed we need that 
the largest contribution to the second moment comes from ``uncorrelated" pairs of configurations. 
This is captured by the following condition.
\begin{condition}\label{cond:maxima}
For $(\alpha,\beta)=(p^+,p^-)$, $\displaystyle\lim_{n\rightarrow\infty}\frac{1}{n}\log\E_\G[Y^{\gamma,\delta}_G]$ is maximized at $\gamma=\alpha^2, \delta=\beta^2$.
\end{condition}
\begin{lemma}\label{lem:secondmax2spin}
Under the hypotheses of Theorem~\ref{thm:2spin}, Condition~\ref{cond:maxima} holds.
\end{lemma} 
\begin{lemma}\label{lem:secondmaxising}
For the antiferromagnetic Ising model with $\lambda=1$, $\Delta=3$ and $B$ in the non-uniqueness region of $\TD$, Condition~\ref{cond:maxima} holds.
\end{lemma} 
\begin{lemma}\label{lem:secondmaxhardcore}
For the hard-core model, $\Delta=3,4,5$ and $\lambda$ in the non-uniqueness region of $\TD$, Condition~\ref{cond:maxima} holds.
\end{lemma}

The innocuous statements of Lemmas~\ref{lem:secondmax2spin},~\ref{lem:secondmaxising},
and \ref{lem:secondmaxhardcore} are rather misleading. Namely proving that uncorrelated pairs of random variables ``dominate'' the second moment is a standard complication of the second moment approach in a variety of settings, see for example \cite{AN,AP,MWW}. To highlight the technical difficulties, the second moment typically has a number of extra other variables that need to be considered and those are usually treated by sophisticated analysis arguments. We escape this paradigm in the following sense: our arguments try to explain the success of the second moment by relating it to the intrinsic nature of the problem at hand (in our context, uniqueness on the infinite tree). To stress the gains of such a scheme, it underlies perhaps a more generic method to translate uniqueness proofs to similar second moment arguments. 

 In the hard-core model, due to the hard constraints, only one extra variable (besides $\gamma,\delta$) was needed and it could be expressed explicitly in terms of $\gamma,\delta$, so that, in fact, 
one had a two-variable optimization. Even with this seemingly simple setup and despite the variety and complexity of approaches appearing in \cite{MWW}, \cite{Sly10} and \cite{Galanis}, certain regions for $\Delta=4,5$ remained open. 
In the general 2-spin model the second moment is more complex than in the hard-core model in the sense that the soft constraints cause the number of extra variables to increase to nine. To make matters even worse, 
one cannot get rid of those extra variables in some easy algebraic way, since such a scheme reduces to a sixth order polynomial equation 
in terms of $\gamma,\delta$. These obstacles are surmounted
 in the proof of Lemma~\ref{lem:secondmax2spin} by a clean argument: the analysis of the second moment is reduced to certain tree recursions on $\TD$, allowing for a short inequality argument.  
 
 After establishing Condition~\ref{cond:maxima}, one obtains that for $(\alpha,\beta)=(p^\pm,p^\mp)$ it holds that $\E_\G[(Z^{\alpha,\beta}_G)^2]\approx \E_\G[Y^{\alpha^2,\beta^2}_G]$. 
At this point, we are in position to study the ratio $\E_\G[(Z^{\alpha,\beta}_G)^2]/(\E_\G[Z^{\alpha,\beta}_G])^2$. This is a technical computation and requires an asymptotic tight approximation of sums by appropriate Gaussian integrals. While the asymptotic value of the ratio is a constant, as in \cite{MWW}, the constant is greater than 1. 
\begin{lemma}\label{lem:ratiofirstsecond}
When Condition~\ref{cond:maxima} is true, for $(\alpha,\beta)=(p^\pm,p^\mp)$, it holds that
\[\lim_{n\rightarrow\infty}\frac{\E_{\G}\left[(Z^{\alpha,\beta}_G)^2\right]}{\left(\E_{\G}[Z^{\alpha,\beta}_G]\right)^2}=\big(1-\omega^2\big)^{-(\Delta-1)/2}\big(1-(\Delta-1)^2\omega^2\big)^{-1/2},\]
where $\omega$ is given by \eqref{eq:definitionofomega}.
\end{lemma}
\begin{remark}
The formula given in the lemma is valid for the hard-core model as well ($B_1=0, B_2=1$), agreeing with the statement of \cite[Theorem 3.3]{MWW}.
\end{remark}
The fact that the ratio of the second moment to the first moment squared converges to a constant greater than 1 does not allow to get a.a.s.\ results using the standard second moment method. This can be dealt with by the so-called small subgraph conditioning method \cite{Wormald,JLR}. The consequence of applying the small subgraph conditioning method is the following result (proved in Section~\ref{sec:oversmall}).
\begin{lemma}\label{lem:smallgraph}
When Condition~\ref{cond:maxima} is true, for $(\alpha,\beta)=(p^\pm,p^\mp)$, it holds asymptotically almost surely over the graph $G\sim\G(n,\Delta)$ that
\[ Z^{\alpha,\beta}_G\geq \frac{1}{n}\E_\G[Z^{\alpha,\beta}_G].\]
\end{lemma}
We next conclude Theorem~\ref{thm:bimodality}.
\begin{proof}[Proof of  Theorem~\ref{thm:bimodality}.]
Lemmas~\ref{lem:firstmomentmax} and~\ref{lem:smallgraph} establish the necessary ingredients to carry out the argument in \cite[Proof of Theorem 2.2]{MWW} (see also \cite[Proofs of Theorems 2 and 3]{Galanis}). The torpid mixing of Glauber dynamics  follows immediately by a conductance argument, using  \cite[Claim 2.3]{DFJ}.
\end{proof}

The outline of the rest of the paper is as follows. 
In Section~\ref{sec:tree-recursions} we derive certain 
properties for the infinite tree which are used throughout the remainder of the paper.
We then look at the first and second moments in Section \ref{sec:moments}.
In Section \ref{sec:maxima}, we look at the optimization problems
for the first and second moments and prove Lemmas~\ref{lem:firstmomentmax} and \ref{lem:secondmax2spin}. 
In Section \ref{sec:NP-outline}, we give the proof of our inapproximability
results (Theorems~\ref{thm:hardcore}, \ref{thm:Ising} and \ref{thm:2spin}) based on Sly's reduction for the hard-core model. 

In Section \ref{sec:calculating-moments} we analyze the ratio of the 
second moment to the first moment squared for random $\Delta$-regular graphs (Lemma~\ref{lem:ratiofirstsecond})
and for the modified graph used in Sly's reduction. In Section \ref{sec:smallgraph}, we apply the small subgraph conditioning method  to get a.a.s.\ results (over the choice of the random graph), where we also prove Lemma~\ref{lem:smallgraph}. Finally, in Section~\ref{sec:remainingproofs} we give the proofs of some remaining lemmas, including Lemmas~\ref{lem:secondmaxising} and \ref{lem:secondmaxhardcore}.

\section{Extremal Measures on the Infinite Tree for 2-Spin Models}
\label{sec:tree-recursions}

In Section~\ref{sec:proof-approach}, we defined the measures $\mu^{*},\mu^{\pm}$ on $\TD$ which will be of interest to us.  In our context, the phase transition on the infinite $(\Delta-1)$-ary tree $\TDary$ rooted at $\rho$ will also be of interest (note that $\TDary$ differs from $\TD$ only in that the root has degree $\Delta-1$). The uniqueness/non-uniqueness regions for $\TDary$ coincide with those of $\TD$, albeit with different occupation probabilities for the root in the respective free and fixed boundary measures. 

We denote by $\hat{\mu}^{\pm}, \hat{\mu}^*$ the measures on $\TDary$ which are the analogues of the measures $\mu^{\pm}, \mu^*$ on $\TD$, respectively. Recall that we use $+,-$ to denote the even and odd boundary conditions, and  $*$ for the free boundary condition. We will denote by $q^{\pm},q^*$ the marginal probabilities of  the root being assigned $-1$ in each of these measures on $\TDary$, i.e., $q^{\pm}=\hat{\mu}^{\pm}(\sigma_\rho=-1)$ and $q^{*}=\hat{\mu}^{*}(\sigma_\rho=-1)$. Once again, it can easily be proved that $q^-<q^*<q^+$ using (anti)monotonicity arguments. 

A well known and interesting property that the measures $\mu^{\pm}$ (and $\hat{\mu}^{\pm}$) possess is that of extremality. Namely, in the  non-uniqueness region they cannot be written as nontrivial convex combinations of other Gibbs measures. In this particular setting, this follows from (and can be interpreted as) the fact that the even and odd boundaries force the maximal bias on the probability that the root $\rho$ is assigned the spin $-1$ as the height of the tree goes to infinity; in other words, for every Gibbs measure $\mu\neq \mu^{\pm}$ on $\TD$, it holds that   $p^{-}<\mu(\sigma_\rho=-1)<p^{+}$. For a more complete account of extremality, we refer the reader to \cite{Georgii}.

In this section,  we state some further properties for the densities $q^\pm, q^*$ and $p^\pm,p^*$, which we are going to utilize throughout the text. Standard tree recursions for Gibbs measures (e.g., see \cite{SST}) establish that
\begin{equation}
\frac{q^\pm}{1-q^\pm}=\lambda\left(\frac{B_1\cdot \frac{q^\mp}{1-q^\mp}+1}{\frac{q^\mp}{1-q^\mp}+B_2}\right)^{\Delta-1}, \quad \frac{q^*}{1-q^*}=\lambda\left(\frac{B_1\cdot \frac{q^*}{1-q^*}+1}{\frac{q^*}{1-q^*}+B_2}\right)^{\Delta-1}. \label{eq:recurone}
\end{equation}
We define separate expressions for $\frac{q^\pm}{1-q^\pm}, \frac{q^*}{1-q^*}$, i.e.,
\begin{equation}
Q^{+}=\frac{q^+}{1-q^+},\qquad Q^{-}=\frac{q^-}{1-q^-}, \qquad Q^{*}=\frac{q^*}{1-q^*}.\label{eq:qpqm}
\end{equation}
The following lemma summarizes the properties 
of the occupation densities $q^+, q^-, q^*$ which we are going to utilize. Its rather folklore proof is a straightforward analysis of appropriate tree recursions and is omitted.
\begin{lemma}\label{lem:densitiesone}
For $\Delta\geq 3$ and $\B_1,\B_2,\lambda>0$, consider the system of equations
\begin{equation}
x=\lambda\left(\frac{B_1 y+1}{y+B_2}\right)^{\Delta-1},\ y=\lambda\left(\frac{B_1 x+1}{x+B_2}\right)^{\Delta-1}, \label{eq:densitiesone}
\end{equation} 
with $x,y\geq 0$.  
\begin{enumerate}
\item In the uniqueness region of $\TDary$, the only solution 
to \eqref{eq:densitiesone} is $(Q^*,Q^*)$.
\item In the non-uniqueness region of $\TDary$, \eqref{eq:densitiesone} has 
exactly three solutions which are given by $(Q^\pm,Q^\mp)$ and $(Q^*,Q^*)$. It holds that $Q^-<Q^*<Q^+$.
\end{enumerate}
\end{lemma}
\begin{proof}
See, for example, \cite[Section 6.2]{MST}. 
\end{proof}

We also define the following quantities $\omega,\ \omega^{*}$ which will emerge throughout the paper
(and will be motivated shortly):
\begin{equation}
\omega:=\frac{(1-B_1B_2)^2Q^+Q^-}{(B_2+Q^-)(B_2+Q^+)(1+B_1Q^-)(1+B_2Q^+)},\quad \omega^*:=\frac{(1-B_1B_2)^2(Q^*)^2}{(B_2+Q^*)^2(1+B_1Q^*)^2}.\label{eq:definitionofomega}
\end{equation}
Note that $\omega^*$ is obtained from $\omega$ after identifying $Q^\pm$ with $Q^*$. The following lemma for $\omega,\omega^*$ will be crucial to prove some optimality conditions required to establish Lemma~\ref{lem:firstmomentmax} and Condition~\ref{cond:maxima}. Its proof is given in Section~\ref{sec:technicalinequality}.
\begin{lemma}\label{lem:technicalinequality}
Let $\Delta\geq 3$. For $(B_1,B_2,\lambda)$ in the non-uniqueness region of $\TDary$, it holds that
\[(\Delta-1)^2\omega<1<(\Delta-1)^2\omega^*.\]
\end{lemma}
Let us explain the mysterious looking inequality in Lemma~\ref{lem:technicalinequality}. While not obvious, it is related with the two-step recurrence on the infinite tree $\TDary$. The two step-recurrence is obtained by the composition of two successive one step recurrences. In the non-uniqueness region, the fixed points of the two step recurrence are exactly $q^\pm,q^*$ (this is basically Lemma~\ref{lem:densitiesone}). The inequality in Lemma~\ref{lem:technicalinequality} establishes, after some calculations, that the fixed point $q^+$ (resp. $q^-$) is an attractor of the two step recurrence in $\hat{\mu}^+$ (resp. $\hat{\mu}^-$), while  the fixed point $q^*$ is repulsive. More specifically, the derivative of the two step function at $q^\pm$ is equal to $(\Delta-1)^2\omega$, and similarly at $q^*$ is equal to $(\Delta-1)^2\omega^*$. This property will be used to show that the measures $\hat{\mu}^{\pm}$ satisfy a strong form of non-reconstruction  with exponential decay of correlation. We should also note that the analysis in \cite{SST} (see also \cite{LLY}) establishes the following criterion for uniqueness, which we reformulate for our purposes: uniqueness holds on $\TDary$ iff $(\Delta-1)^2\omega^*\leq 1$.  

Finally, let us translate the above to obtain a handle on $p^\pm, p^*$. Using a slight modification of the tree recursions in $\TDary$, 
one obtains that
\begin{equation}
\frac{p^\pm}{1-p^\pm}=\frac{q^\pm}{1-q^\pm}\left(\frac{B_1\cdot \frac{q^\mp}{1-q^\mp}+1}{\frac{q^\mp}{1-q^\mp}+B_2}\right), \qquad \frac{p^*}{1-p^*}=\frac{q^*}{1-q^*}\left(\frac{B_1\cdot \frac{q^*}{1-q^*}+1}{\frac{q^*}{1-q^*}+B_2}\right). \label{eq:pppmtemp}
\end{equation}
Using \eqref{eq:qpqm} and \eqref{eq:pppmtemp}, an easy calculation gives
\begin{equation}
p^{\pm}=\frac{Q^{\pm}(1+B_1 Q^{\mp})}{B_2+Q^{+}+Q^{-}+B_1 Q^{+} Q^{-}}.\label{eq:pppm}
\end{equation}
Equation \eqref{eq:pppm} will prove to be very useful to simplify seemingly intricate calculations. 

For the case of the Ising model with zero external field ($B_1=B_2=B,\lambda=1$), it holds that $p^++p^-=q^++q^-=1$, $p^*=q^*=1/2$. These properties may be proved easily using the symmetries of the model. Namely, observe that if $(x,y)$ satisfy \eqref{eq:densitiesone} then so do $(1/y,1/x)$. By items 1 and 2 of Lemma~\ref{lem:densitiesone}, this implies that $q^++q^-=1$ and $q^*=1/2$. By \eqref{eq:pppmtemp}, it also follows that $p^++p^-=1$, $p^*=1/2$.

\section{Moment Analysis}
\label{sec:moments}
In this section, we give expressions for $\E_{\G}\big[Z^{\alpha,\beta}_G\big]$, $\E_\G\big[Y^{\gamma,\delta}_G\big]$ and then study the asymptotics of $\frac{1}{n}\log\E_{\G}\big[Z^{\alpha,\beta}_G\big]$, $\frac{1}{n}\log\E_{\G}\big[Y^{\gamma,\delta}_G\big]$.

\subsection{First Moment}\label{sec:firstmomentform}
The first moment of $Z^{\alpha,\beta}_G$ is given by
\begin{equation}
\E_{\G}\big[Z^{\alpha,\beta}_G\big] = \lambda^{n(\alpha+\beta)}\binom{n}{\alpha n}\binom{n}{\beta n}\left(\sum_{\delta\leq \alpha, \delta\leq \beta}
\frac{\binom{\alpha n}{\delta n}\binom{(1-\alpha)n}{(\beta-\delta) n}}{\binom{n}{\beta n}}B_1^{\delta n}B_2^{(1-\alpha-\beta+\delta)n}\right)^{\Delta}.\label{eq:firstmoment1}
\end{equation}
Ignoring the contribution of the external field $\lambda$, the terms outside the sum count $|\Sigma^{\alpha,\beta}_G|$ . 
The sum gives the expected contribution to
$w_G(\sigma)$ (for $\sigma\in \Sigma^{\alpha,\beta}_G$) from a \textit{single} matching, 
where $\delta n$ is the number of edges between 
$N_-(\sigma)\cap V_1$ and $N_-(\sigma)\cap V_2$ in a single matching. 
The weight of $\sigma$ for a single matching depends only 
on $\delta n$ and is equal to $B_1^{\delta n}B_2^{(1-\alpha-\beta+\delta)n}$. 
The fraction in the sum is the probability that there are 
exactly $\delta n$ edges from $N_-(\sigma)\cap V_1$ to 
$N_-(\sigma)\cap V_2$. 
Since the matchings are chosen independently, the exponentiation to the power $\Delta$ gives the 
expected contribution of $\sigma\in \Sigma^{\alpha,\beta}_G$ 
to the partition function for $\Delta$ matchings. Finally, the $\lambda$ term in the beginning accounts for the contribution of the external field for $\sigma\in \Sigma^{\alpha,\beta}_G$. 

A different way to arrive at a slightly different formulation for the first moment follows. Here, we introduce some redundant variables at the cost of introducing some complexity of the presentation, but underlying a general approach to treat such expressions. 

Observe that the sets $N_{-}(\sigma)\cap V_1,\ N_{+}(\sigma)\cap V_1$ form a partition of $V_1$ and similarly for $V_2$. Thus for a perfect matching let $x_{ij}$, ($1\leq i,j\leq 2$) denote the number of edges having their endpoints in partition $i$ of $V_1$ and in partition $j$ of $V_2$. Note that the $x_{ij}$ must satisfy 
\begin{equation}
\begin{array}{ll}
x_{11}+x_{12}=\alpha,& x_{11}+x_{21}=\beta,\\
x_{21}+x_{22}=1-\alpha,& x_{12}+x_{22}=1-\beta,\\
x_{ij}\geq 0.&\\
\end{array}\label{eq:xregion}
\end{equation}
The first moment may be written as
\begin{equation}
\E_{\G}\big[Z^{\alpha,\beta}_G\big] = \lambda^{(\alpha+\beta)n}\binom{n}{\alpha n}\binom{n}{\beta n}\left(\sum_{\mathbf{X}}
\frac{\binom{\alpha n}{x_{11}n}\binom{(1-\alpha)n}{x_{21}n}\binom{\beta n}{x_{11}n}\binom{(1-\beta)n}{x_{21}n}}{\binom{n}{x_{11}n,x_{12}n,x_{21}n,x_{22}n}}B_1^{x_{11}n}B_2^{x_{22}n}\right)^{\Delta}. \label{eq:firstmoment2}
\end{equation}
In the above sum, we denote by $\mathbf{X}$ the $x_{ij}$ which satisfy \eqref{eq:xregion}. Note that the expressions \eqref{eq:firstmoment1} and \eqref{eq:firstmoment2} are completely equivalent after identifying $\delta$ with $x_{11}$ and using \eqref{eq:xregion} to express the variables $x_{12},x_{21},x_{22}$ in terms of $x_{11}$.

\subsection{Second Moment}~\label{sec:secondmomentform}
We continue with the ideas presented in Section~\ref{sec:firstmomentform}, to get an expression for $\E_\G\big[Y^{\gamma,\delta}_G\big]$.
Recall that for the formulation of the second moment $\alpha$ and $\beta$ are fixed, and hence suppressed from the notation. 
Now there are four induced parts by $\sigma_1,\sigma_2$ on each $V_1$ and $V_2$. Let $\gamma n,\delta n$ stand for the number of  vertices in $V_1$ and $V_2$ whose spin is $-1$ in both $\sigma_1,\sigma_2$ and $y_{ij}n$ the number of edges incident with endpoints in partition $i$ of $V_1$ and partition $j$ of $V_2$ in a single matching. For ease of exposition, define 
\begin{equation*}
\begin{array}{lll}
L_1=\gamma,& L_2=L_3=\alpha-\gamma,& L_4=1-2\alpha+\gamma,\\
R_1=\delta,& R_2=R_3=\beta-\delta,& R_4=1-2\beta+\delta.
\end{array}
\end{equation*}
In the same spirit with expression \eqref{eq:firstmoment2} for the first moment, we have that
\begin{align}
\E_\G\big[Y^{\gamma,\delta}_G\big]&=\lambda^{2(\alpha+\beta)n}\binom{n}{\alpha n}\binom{n}{\beta n}\binom{\alpha n}{\gamma n}\binom{(1-\alpha)n}{(\alpha-\gamma)n}\binom{\beta n}{\delta n}\binom{(1-\beta)n}{(\beta-\delta)n}\label{eq:secondmoment}\\
&\quad \times \bigg(\sum_{\mathbf{Y}}\frac{\prod^4_{i=1} \binom{L_i n}{y_{i1}n,y_{i2}n,y_{i3}n,y_{i4}n}\prod^4_{j=1} \binom{R_jn}{y_{1j}n,y_{2j}n,y_{3j}n,y_{4j}n}}{\binom{n}{y_{11}n, y_{12}n, \dots ,y_{44}n}}\notag\\ 
&\qquad \times B_1^{(2y_{11} + y_{12} + y_{13}+ y_{21} + y_{22}+y_{31}+y_{33})n}B_2^{(y_{22}+y_{24}+y_{33}+y_{34}+y_{42}+y_{43}+2y_{44})n}\bigg)^{\Delta}.\notag
\end{align}
In the above sum, we denote by $\mathbf{Y}$ the $y_{ij}$ which satisfy
\begin{equation}
\begin{array}{ll}
\sum_{j}y_{ij}=L_i,&\mbox{ for } i=1,\hdots,4\\
\sum_{i}y_{ij}=R_j,&\mbox{ for } j=1,\hdots,4\\
y_{ij}\geq 0.&\\
\end{array}\label{eq:yregion} 
\end{equation}
The interpretation of the formula \eqref{eq:secondmoment} should be clear: ignoring the external field, the first line accounts for $|\Sigma^{\alpha,\beta}_{\gamma,\delta}|$, a term in the sum accounts for the weight of a matching specified by the cardinalities $Y$ (last line) as well as the probability of such a matching occurring (middle line).

\subsection{Asymptotics of the Logarithms of the Moments}\label{sec:generallog}
We next turn to finding the asymptotics of $\frac{1}{n}\log\E_{\G}\big[Z^{\alpha,\beta}_G\big]$, $\frac{1}{n}\log\E_{\G}\big[Y^{\gamma,\delta}_G\big]$. The asymptotic order of the sums in \eqref{eq:firstmoment2} and  \eqref{eq:secondmoment}
are dominated by their maximum terms. For the sole purpose of having a unifying treatment for soft and hard constrained 2-spin models, here and in the rest of the paper, we adopt the usual conventions that $\ln 0\equiv -\infty$ and $0\ln 0\equiv 0$. Standard application of Stirling's approximation yields for the first moment:
\begin{align}
\lim_{n\rightarrow \infty}&\ \frac{1}{n}\log\E_{\G}\big[Z^{\alpha,\beta}_G\big] = \max_{\mathbf{X}}\Phi_1(\alpha,\beta,\mathbf{X}),\label{eq:limitfirst}\\
\mbox{ where } \ \ \ 
\Phi_1(\alpha,\beta,\mathbf{X})&:=(\alpha+\beta)\ln \lambda+(\Delta-1)f_1(\alpha,\beta)+ \Delta g_1(\mathbf{X})\notag\\
f_1(\alpha,\beta)&:=\alpha\ln\alpha+(1-\alpha)\ln(1-\alpha)+\beta\ln\beta+(1-\beta)\ln(1-\beta)\notag\\
g_1(\mathbf{X})&:=(x_{11}\ln B_1+x_{22}\ln B_2)-\mbox{$\sum_{i,j}$}x_{ij}\ln x_{ij}.\notag
\end{align}
And for the second moment:
\begin{align}
&\lim_{n\rightarrow \infty} \frac{1}{n}\log\E_{\G}\big[Y^{\gamma,\delta}_G\big]= \max_\mathbf{Y}\Phi_2(\gamma,\delta,\mathbf{Y}), \label{eq:limitsecond}\\
\mbox{ where } \ \ \ 
\Phi_2(\gamma,\delta,\mathbf{Y})&:=2(\alpha+\beta)\ln \lambda+ (\Delta-1)f_2(\gamma,\delta)+\Delta g_2(\mathbf{Y})\notag\\[0.15cm]
f_2(\gamma,\delta)&:=2(\alpha-\gamma)\ln(\alpha-\gamma)+\gamma \ln\gamma+(1-2\alpha+\gamma)\ln(1-2\alpha+\gamma)\notag\\
&\quad \quad +2(\beta-\delta)\ln(\beta-\delta)+\delta \ln\delta+(1-2\beta+\delta)\ln(1-2\beta+\delta)\notag\\[0.15cm]
g_2(\mathbf{Y})&:=(2y_{11} + y_{12} + y_{13}+ y_{21} + y_{22}+y_{31}+y_{33})\ln B_1\notag\\
&\qquad +(y_{22}+y_{24}+y_{33}+y_{34}+y_{42}+y_{43}+2y_{44})\ln B_2\notag\\
&\hskip 6.5cm -\mbox{$\sum_{i,j}$}y_{ij}\ln y_{ij}.\notag
\end{align}
Recall that we are going to study the second moment for specific values of $\alpha,\beta$, namely $(\alpha,\beta)=(p^\pm,p^\mp)$, so this is why we dropped the dependence of the functions $\Phi_2,f_2$ on $\alpha,\beta$. The functions $g_1, g_2$ are defined on the regions \eqref{eq:xregion}, \eqref{eq:yregion} respectively. For fixed $\alpha,\beta,\gamma,\delta$, they are strictly concave over the convex regions on which they are defined and hence they have a unique maximum, so this also the case for the functions $\Phi_1, \Phi_2$. The limits \eqref{eq:limitfirst} and \eqref{eq:limitsecond} can thus be justified using the standard Laplace method (see for example \cite[Chapter 4]{deBru}).

\section{Finding the Maxima}
In this section, we use the information obtained in Section~\ref{sec:generallog} to establish Lemmas~\ref{lem:firstmomentmax} and \ref{lem:secondmax2spin}. The proofs for Lemmas \ref{lem:secondmaxising} and~\ref{lem:secondmaxhardcore} are given in Sections~\ref{sec:secondmaxisingproof} and~\ref{sec:hardhard} respectively, see also Section~\ref{sec:remarks} for an overview. These proofs are based on the approach given in this section.
\label{sec:maxima}
\subsection{Optimizing Entropy Distributions}\label{sec:maxentropy}
By the limits \eqref{eq:limitfirst} and \eqref{eq:limitsecond}, to establish Lemmas~\ref{lem:firstmomentmax} and \ref{lem:secondmax2spin} (and~\ref{lem:secondmaxising},~\ref{lem:secondmaxhardcore} as well), it suffices to study the maxima of the functions
\begin{equation}\label{eq:phi1phi2}
\phi_1(\alpha,\beta):=\max_{\mathbf{X}} \Phi_1(\alpha,\beta, \mathbf{X})\mbox{ and }\phi_2(\gamma,\delta):=\max_{\mathbf{Y}} \Phi_2(\gamma,\delta, \mathbf{Y}).
\end{equation}
A small abstraction allows us to treat the two functions in a unifying way. 

Namely, observe that $\phi_1,\phi_2$  have the form  $\phi(u,v)=\max_{\mathbf{Z}}\Phi(u,v,\mathbf{Z})$, where $\Phi$ can be decomposed as $\Phi(u,v,\mathbf{Z})=f(u,v)+g(\mathbf{Z})$, for some functions $f,g$ and $\mathbf{Z}$ stands for $mn$ non-negative variables $Z_{ij}$, $i=1,\hdots,m$, $j=1,\hdots,n$.\footnote{In all the applications of the arguments in this section, we will have $m=n$. The slightly more general setting considered here (which allows $m\neq n$) is not crucial but the arguments extend to this setting without any additional effort.} Thus in the context of $\phi_1$ (resp. $\phi_2$), $u,v,\mathbf{Z}$ correspond 
to $\alpha,\beta,\mathbf{X}$ (resp. $\gamma,\delta,\mathbf{Y}$). Moreover, $\mathbf{Z}$ satisfies some prescribed marginals, i.e.,
\begin{equation}\label{e1}
\mbox{$\sum_{j}$} Z_{ij} = \alpha_i(u),\quad \mbox{$\sum_{i}$} Z_{ij} = \beta_j(v), \quad Z_{ij}\geq 0,
\end{equation}
for some functions $\alpha_i(u)$ and $\beta_j(v)$, which satisfy $\sum_{i} \alpha_i(u) =\sum_{j} \beta_j(v) = 1$. In this section, we will only consider $u,v$ such that $\alpha_i(u)\neq 0\neq \beta_j(v)$ for all $i,j$, since the remaining cases correspond to boundary conditions in our setup and these are treated more easily in the specific instance of the maximization (see the relevant  Lemmas~\ref{lem:firboundary} and \ref{lem:secboundary}). Finally, the function $g(\mathbf{Z})$ is an entropy-like function and has the form
\begin{equation}
g(\mathbf{Z})=\mbox{$\sum_{i}\sum_j$} \left(Z_{ij} \ln M_{ij} - Z_{ij}\ln Z_{ij}\right),
\end{equation}
where $M\in\mathbb{R}^{m\times n}$ is a matrix with non-negative entries.

We will use another interpretation of the function $g(\mathbf{Z})$, which will be handy when we perform integrations: let $\mathbf{Z^f}$ denote the variables $Z_{ij}$ for $i=1,\hdots, m-1,$ $j=1,\hdots,n-1$. Using the equations in \eqref{e1}, we may write $g(\mathbf{Z})$ as a function of the variables $\mathbf{Z^f}$. Note that the region \eqref{e1} must be adjusted to account for the non-negativity of the substituted variables. Other than this, the ``new'' function  has exactly the same behaviour as $g(\mathbf{Z})$ but in a full-dimensional space. Thus we will also think of $g(\mathbf{Z})$ as a function of the variables $\mathbf{Z^f}$ where the $Z_{ij}$ with $i=m$ or $j=n$ are just shorthands for the expressions of these variables when substituted by the equations in \eqref{e1}. We will refer to this reformulation as the \textit{full-dimensional representation} of $g(\mathbf{Z})$. We define similarly the full-dimensional representation of $\Phi(u,v,\mathbf{Z})$.

It is immediate to see that $g(\mathbf{Z})$ is strictly concave in the convex region \eqref{e1}, hence it attains its maximum over the region at a unique point. A tedious but straightforward optimization of $g(\mathbf{Z})$ using Lagrange multipliers, yields:
\begin{lemma}\label{lem:gmax}
Fix $u,v$ such that $\alpha_i(u)\neq 0 \neq \beta_j(v)$ for all $i,j$. Let $\mathbf{Z}^*=\arg\max_{\mathbf{Z}}g(\mathbf{Z})$ (the maximum taken over the region \eqref{e1}). Then
\begin{equation}\label{e2}
Z^*_{ij} = M_{ij} R_i C_j,
\end{equation}
where $R_i$ and $C_j$ are positive real numbers, which satisfy
\begin{equation}\label{e3}
R_i \sum_{j} M_{ij} C_j = \alpha_i\quad\mbox{and}\quad
C_j \sum_{i} M_{ij} R_i = \beta_j.
\end{equation}
Moreover, the full-dimensional representation of $g(\mathbf{Z})$ decays quadratically in a sufficiently small ball around $\mathbf{Z}^*$.
\end{lemma}
The proof of Lemma~\ref{lem:gmax} is given in Section~\ref{sec:maxentropyproof}. It is relevant to point out that while the values of $Z^*_{ij}$ are unique in the region \eqref{e1}, the values of $R_i,C_j$ which satisfy \eqref{e3} are not unique since a positive solution $\{\{R_i\}_i,\{C_j\}_j\}$ of \eqref{e3} gives rise to a positive solution $\{\{tR_i\}_i,\{C_j/t\}_j\}$ for arbitrary $t>0$. Modulo this mapping, the uniqueness of $Z^*_{ij}$ easily implies that the positive $R_i,C_j$ which satisfy \eqref{e3} are otherwise unique.

Plugging \eqref{e2} in $g(\mathbf{Z})$, we obtain 
\begin{equation}\label{eq:QQ}
\max_{\mathbf{Z}}g(\mathbf{Z}) = -\mbox{$\sum_{i} \sum_{j}$} M_{ij} R_i C_j \ln (R_i C_j):=g^{*}(u,v).
\end{equation}

In our setup, it will be hard in general to solve for the $R_i,C_j$, so that we need to study $\frac{\partial g^{*}}{\partial u},\frac{\partial g^{*}}{\partial v}$ using implicit differentiation. The following lemma will be very useful.
\begin{lemma}\label{lem:maxentropy}
For $g^{*}=\max_{\mathbf{Z}} g(\mathbf{Z})$, we have
\begin{equation*}
\frac{\partial g^{*}}{\partial u}= -\sum_{i} \ln(R_i) \frac{\partial \alpha_i}{\partial u} \mbox{\  and \ } \frac{\partial g^{*}}{\partial v}= - \sum_{j} \ln(C_j) \frac{\partial \beta_j}{\partial v}.
\end{equation*}
In particular,
\begin{equation*}
\frac{\partial (f+g^{*})}{\partial u}= \frac{\partial f}{\partial u}-\sum_{i} \ln(R_i) \frac{\partial \alpha_i}{\partial u} \mbox{\  and \ } \frac{\partial (f+g^{*})}{\partial v}= \frac{\partial f}{\partial v}- \sum_{j} \ln(C_j) \frac{\partial \beta_j}{\partial v}.
\end{equation*}
where the $R_i,C_j$ are positive solutions of~\eqref{e3}.
\end{lemma}
The proof of Lemma~\ref{lem:maxentropy} is given in Section~\ref{sec:maxentropyproof}.

\subsection{The Logarithm of the First Moment - Maximum}\label{sec:logfirstmax}
In this section, we prove Lemma~\ref{lem:firstmomentmax}. As observed in Sections~\ref{sec:generallog} and~\ref{sec:maxentropy}, it suffices to study the maxima of $\phi_1(\alpha,\beta)$. Lemma~\ref{lem:firstmomentmax} is implied by the following lemmas. 
\begin{lemma}\label{lem:fircritical}
The only critical points of $\phi_1(\alpha,\beta)$ are:
\begin{enumerate}
\item $(\alpha,\beta)=(p^\pm,p^\mp)$ and $(\alpha,\beta)=(p^*,p^*)$, when $B_1,B_2,\lambda$ lie in the non-uniqueness region of $\TD$.
\item $(\alpha,\beta)=(p^*,p^*)$, when $B_1,B_2,\lambda$ lie in the uniqueness region of $\TD$.
\end{enumerate}
\end{lemma}
\begin{lemma}\label{lem:firboundary}
The function $\phi_1(\alpha,\beta)$ does not have any local maximum on the boundary of the region \[\big\{(\alpha,\beta)\mid 0\leq \alpha\leq 1,\ 0\leq\beta\leq 1\big\}.\]
\end{lemma}
\begin{lemma}\label{lem:firhessian}
Let $\mathbf{X}^{*}=\arg\max_{\mathbf{X}}\Phi_1(p^+,p^-,\mathbf{X})$ and $\mathbf{X}^{o}=\arg\max_{\mathbf{X}}\Phi_1(p^*,p^*,\mathbf{X})$. 
\begin{enumerate}
\item When $B_1,B_2,\lambda$ lie in the non-uniqueness region of $\TD$, the function $\Phi_1(\alpha,\beta,\mathbf{X})$ has a local maximum at the point $(p^+,p^-, \mathbf{X}^*)$ and a saddle point at $(p^*,p^*, \mathbf{X}^o)$.
\item When $B_1,B_2,\lambda$ lie in the uniqueness region of $\TD$, the function $\Phi_1(\alpha,\beta,\mathbf{X})$ has a local maximum at $(p^*,p^*, \mathbf{X}^o)$.
\end{enumerate}
\end{lemma}
\begin{proof}[Proof of Lemma~\ref{lem:firstmomentmax}.]
Recall from \eqref{eq:limitfirst} and \eqref{eq:phi1phi2} that  
\[\lim_{n\rightarrow \infty} \frac{1}{n}\log\E_{\G}\big[Z^{\alpha,\beta}_G\big] = \phi_1(\alpha,\beta)=\max_{\mathbf{X}}\Phi_1(\alpha,\beta,\mathbf{X}).\] 
By Lemma~\ref{lem:firboundary}, the function $\phi_1(\alpha,\beta)$ does not attain its maximum on the boundary and hence its (global) maximum is achieved at a critical point. Lemma~\ref{lem:fircritical} gives the critical points in the uniqueness/nonuniqueness region and Lemma~\ref{lem:firhessian} classifies which are local maxima. The result follows. 
\end{proof}
We defer the proofs of Lemmas~\ref{lem:firboundary} and~\ref{lem:firhessian} to Section~\ref{sec:logfirstmaxproof}. We give the more interesting proof of Lemma~\ref{lem:fircritical}.

\begin{proof}[Proof of Lemma~\ref{lem:fircritical}.]
We apply Lemma~\ref{lem:maxentropy} with 
\begin{gather*}
\mathbf{Z}\leftarrow \mathbf{X},\ u\leftarrow\alpha,\  v\leftarrow\beta,\ \alpha_1\leftarrow\alpha,\ \alpha_2\leftarrow1-\alpha \ \beta_1\leftarrow\beta,\ \beta_2\leftarrow1-\beta,\\
f\leftarrow(\alpha+\beta)\ln\lambda+(\Delta-1)f_1,\ g\leftarrow \Delta g_1.
\end{gather*}
The matrix $M$ is given by
\[M= \left(\begin{array}{cc}
   B_1 & 1 \\
   1 & B_2 \\
 \end{array}\right).\]
Lemma~\ref{lem:maxentropy} gives that the critical points of  $\max_\mathbf{X}\Phi_1 (\alpha,\beta,\mathbf{X})$ must satisfy
\begin{equation}\label{eaeaea2}
\begin{aligned}
\frac{\partial \phi_1 }{\partial\alpha}&=\ln\left(\lambda \left(\frac{\alpha}{1-\alpha}\right)^{\Delta-1} \left(\frac{R_2}{R_1}\right)^\Delta \right)=0,\\
\frac{\partial \phi_1}{\partial\beta}&=\ln\left(\lambda \left(\frac{\beta}{1-\beta}\right)^{\Delta-1} \left(\frac{C_2}{C_1}\right)^\Delta \right)=0,
\end{aligned}
\end{equation}
where
\begin{equation}\label{eb}
\begin{array}{lll}
R_1 (B_1 C_1 + C_2)= \alpha,& & C_1 (B_1 R_1 + R_2)= \beta,\\
R_2 (C_1 + B_2 C_2)= 1-\alpha,& & C_2 (R_1 + B_2 R_2)=1-\beta.
\end{array}
\end{equation}
Let $r=R_1/R_2$ and $c=C_1/C_2$. Using \eqref{eaeaea2}  and \eqref{eb}, we obtain
\[r \frac{B_1 c + 1}{c + B_2} = \frac{\alpha}{1-\alpha} = r^{1+1/(\Delta-1)} (1/\lambda)^{1/(\Delta-1)}\Longrightarrow r=\lambda \left(\frac{B_1 c + 1}{c + B_2}\right)^{\Delta-1},\]
\[c \frac{B_1 r + 1}{r + B_2} = \frac{\beta}{1-\beta} = c^{1+1/(\Delta-1)} (1/\lambda)^{1/(\Delta-1)}\Longrightarrow c=\lambda \left( \frac{B_1 r + 1}{r + B_2}\right)^{\Delta-1}.\]
The right hand side equations are exactly the equations \eqref{eq:densitiesone} and in light of Lemma~\ref{lem:densitiesone} it must be the case that in the non-uniqueness region either that $r=\frac{q^\pm}{1-q^{\pm}}, c=\frac{q^\mp}{1-q^{\mp}}$ or $r=c=\frac{q^{*}}{1-q^{*}}$. By \eqref{eq:pppm}, we hence obtain that the critical points for the first moment in the non-uniqueness region are given by $(\alpha,\beta)=(p^\pm,p^\mp)$ or $(\alpha,\beta)=(p^{*},p^{*})$. Similarly for the uniqueness region.
\end{proof}

\subsection{The Logarithm of the Second Moment - Maximum}\label{sec:logsecmax}
In this section, we prove Lemma~\ref{lem:secondmax2spin}. As established in Sections~\ref{sec:generallog} and~\ref{sec:maxentropy}, it suffices to study the maxima of $\phi_2(\gamma,\delta)$. Lemma~\ref{lem:secondmax2spin} is immediately implied by the following lemmas. 

\begin{lemma}\label{lem:seccritical}
Under the hypotheses of Theorem~\ref{thm:2spin}, for $(\alpha,\beta)=(p^\pm,p^\mp)$, the only critical points of $\phi_2(\gamma,\delta)$ satisfy $\gamma=\alpha^2,\ \delta=\beta^2$.
\end{lemma}
\begin{lemma}\label{lem:secboundary}
For $(\alpha,\beta)=(p^\pm,p^\mp)$, the function $\phi_2(\gamma,\delta)$ does not have any local maximum on the boundary of the region
\[\big\{(\gamma,\delta)\mid \gamma\geq0,\ \alpha-\gamma\geq0,\ 1-2\alpha+\gamma\geq 0,\ \delta\geq 0,\ \beta-\delta\geq0,\ 1-2\beta+\delta\geq0\big\}.\]
\end{lemma}
\begin{lemma}\label{lem:sechessian}
For $(\alpha,\beta)=(p^+,p^-)$, let $\mathbf{Y}^{*}=\arg\max_{\mathbf{Y}}\Phi_2(\alpha^2,\beta^2,\mathbf{Y})$.  The function $\Phi_2(\gamma,\delta,\mathbf{Y})$ has a local maximum at $(\alpha^2,\beta^2,\mathbf{Y}^{*})$.
\end{lemma}
\begin{proof}[Proof of Lemma~\ref{lem:secondmax2spin}.]
Just combine Lemmas~\ref{lem:seccritical},~\ref{lem:secboundary} and~\ref{lem:sechessian} as in the proof of Lemma~\ref{lem:firstmomentmax} in Section~\ref{sec:logfirstmax}.
\end{proof}
The proofs of Lemmas~\ref{lem:secboundary} and~\ref{lem:sechessian} are valid for all $B_1,B_2, \lambda$, 
so that to prove Lemmas~\ref{lem:secondmaxising} and~\ref{lem:secondmaxhardcore}, one only needs to obtain the analogues of Lemma~\ref{lem:seccritical} in the respective settings of $B_1,B_2$. The proofs of these analogues are given in Sections~\ref{sec:secondmaxisingproof} and \ref{sec:hardhard}, see also Section~\ref{sec:remarks} for a proof overview. We defer the proofs of Lemmas~\ref{lem:secboundary} and~\ref{lem:sechessian} to Section~\ref{sec:logsecmaxproof} and give here the more interesting proof of Lemma~\ref{lem:seccritical}.
\begin{proof}[Proof of Lemma~\ref{lem:seccritical}.]
We apply once more Lemma~\ref{lem:maxentropy} with
\begin{gather*}
\mathbf{Z}\leftarrow \mathbf{Y},\ u\leftarrow \gamma,\ v\leftarrow \delta,\ \alpha_1\leftarrow\gamma,\ \alpha_2=\alpha_3\leftarrow\alpha-\gamma,\ \alpha_4\leftarrow1-2\alpha+\gamma\\
\beta_1\leftarrow\gamma,\ \beta_2=\beta_3\leftarrow\beta-\delta,\ \beta_4\leftarrow1-2\beta+\delta,\ f\leftarrow (\Delta-1)f_2,\ g\leftarrow \Delta g_2.
\end{gather*}
The matrix $M$ is given by
\[M=\left(\begin{array}{cccc}
   B_1^2 & B_1 & B_1 & 1 \\
   B_1 & B_1 B_2 & 1 & B_2 \\
   B_1 & 1 & B_1 B_2 & B_2 \\
   1 & B_2 & B_2 & B_2^2 \\
 \end{array}\right).\]
We obtain that the critical points of $\max_{\mathbf{Y}}\Phi_2(\gamma,\delta,\mathbf{Y})$ must satisfy
\begin{align}
\frac{\partial \phi_2}{\partial \gamma}&  =
\ln\left( \left(\frac{\gamma(1-2\alpha+\gamma)}{(\alpha-\gamma)^2}\right)^{\Delta-1}\left(\frac{R_2 R_3}{R_1 R_4}\right)^\Delta \right)=0,\label{es1}\\
\frac{\partial \phi_2}{\partial \delta}& =
\ln\left( \left(\frac{\delta(1-2\beta+\delta)}{(\beta-\delta)^2}\right)^{\Delta-1}\left(\frac{C_2 C_3}{C_1 C_4}\right)^\Delta \right)=0,\label{es2}
\end{align}
where the $R_i,C_j$ satisfy the following equations (symbol $\circ$ denotes the Hadamard product).
\begin{equation}\label{et}
\left(\begin{array}{cc}R_1 & C_1\\R_2 & C_2\\R_3 & C_3\\R_4 & C_4\end{array}\right)\circ M\left(\begin{array}{cc}C_1 & R_1\\C_2& R_2\\C_3 & R_3\\C_4 & R_4\end{array}\right)=\left(\begin{array}{cc}\gamma  & \delta\\ \alpha-\gamma & \beta-\delta \\\alpha-\gamma & \beta-\delta \\1-2\alpha+\gamma & 1-2\beta+\delta\end{array}\right).
\end{equation}
We will write out an explicit form of the equations after establishing the following claim.
\begin{claim}\label{claim:equalityrc}
It holds that $R_2=R_3$ and $C_2=C_3$.
\end{claim}
\begin{proof}
Observe that the matrix $M$ remains invariant upon interchanging its 2nd,3rd rows and 2nd,3rd columns, i.e.,
\begin{multline*}
(R_1,R_2,R_3,R_4,C_1,C_2,C_3,C_4) \mbox{ satisfy } \eqref{et}\Longleftrightarrow \\ (R_1,R_3,R_2,R_4,C_1,C_3,C_2,C_4) \mbox{ satisfy } \eqref{et}.
\end{multline*}
Since $g_2(\mathbf{Y})$ is strictly concave, this yields the claim.
\end{proof}
The rest of the proof will work as follows: as in the first moment, we first do some manipulations with the equations so that we arrive to a nice form. This form coincides with an inequality version of the tree equations \eqref{eq:densitiesone} for the ferromagnetic Ising model without external field, which can be analyzed easily. 

Using Claim~\ref{claim:equalityrc}, the equalities \eqref{et} give
\begin{equation}\label{eq:frty}
\begin{aligned}
&R_1=\frac{\gamma}{B^2_1 C_1+2B_1 C_2+C_4}, & &C_1=\frac{\delta}{B^2_1R_1+2B_1R_2+R_4},\\
&R_2=\frac{\alpha-\gamma}{B_1 C_1+(B_1B_2+1)C_2+B_2C_4}, & &C_2= \frac{\beta-\delta}{B_1R_1+(B_1B_2+1)R_2+B_2R_4},\\
&R_4=\frac{1-2\alpha+\gamma}{C_1+2B_2 C_2+B^2_2 C_4}, & & C_4=\frac{ 1-2\beta+\delta}{R_1+2B_2R_2+B^2_2R_4}. 
\end{aligned}
\end{equation}
Also, using again Claim~\ref{claim:equalityrc}, the equalities in \eqref{es1} and \eqref{es2} give 
\begin{gather}
\frac{\partial \phi_2}{\partial \gamma} =
\ln\left( \left(\frac{\gamma(1-2\alpha+\gamma)}{(\alpha-\gamma)^2}\right)^{\Delta-1}\left(\frac{R_2^2}{R_1 R_4}\right)^\Delta \right)=0,\label{eq:dergamma}\\
\frac{\partial \phi_2}{\partial \delta} =
\ln\left( \left(\frac{\delta(1-2\beta+\delta)}{(\beta-\delta)^2}\right)^{\Delta-1}\left(\frac{C_2^2}{C_1 C_4}\right)^\Delta \right)=0.\label{eq:derdelta}
\end{gather}
We now set $r_1=R_1/R_2, r_4=R_4/R_2, c_1=C_1/C_2, c_4=C_4/C_2$. After dividing the appropriate pairs of equations in \eqref{eq:frty}, we obtain
\begin{align}
r_1=\frac{\gamma}{\alpha-\gamma}\cdot\frac{B_1c_1+(B_1B_2+1)+B_2c_4}{B^2_1c_1+2B_1+c_4},&\  r_4 =\frac{1-2\alpha+\gamma}{\alpha-\gamma}\cdot\frac{B_1c_1+(B_1B_2+1)+B_2c_4}{c_1+2B_2+B^2_2c_4},\label{r1r4}\\
c_1 = \frac{\delta}{\beta-\delta}\cdot\frac{B_1r_1+(B_1B_2+1)+B_2r_4}{B^2_1r_1+2B_1+r_4},& \  
c_4 = \frac{1-2\beta+\delta}{\beta-\delta}\cdot\frac{B_1r_1+(B_1B_2+1)+B_2r_4}{r_1+2B_2+B^2_2r_4}.\label{c1c4}
\end{align}
Equations \eqref{eq:dergamma} and \eqref{eq:derdelta} become
\begin{align}
\left(r_1r_4\right)^\Delta=\left(\frac{\gamma(1-2\alpha+\gamma)}{(\alpha-\gamma)^2}\right)^{\Delta-1},\quad
\left(c_1c_4\right)^\Delta=\left(\frac{\delta(1-2\beta+\delta)}{(\beta-\delta)^2}\right)^{\Delta-1}.\label{secab}
\end{align}
Using the identity
\begin{multline*}
(B^2_1c_1+2B_1+c_4)(c_1+2B_2+B^2_2c_4)=\\
(B_1c_1+(B_1B_2+1)+B_2c_4)^2+(1-B_1B_2)^2(c_1c_4-1),
\end{multline*}
and multiplying the equations in \eqref{r1r4}, we obtain
\begin{align}
r_1r_4&=\frac{\gamma(1-2\alpha+\gamma)}{(\alpha-\gamma)^2}\cdot\frac{(B_1c_1+(B_1B_2+1)+B_2c_4)^2}{(B_1c_1+(B_1B_2+1)+B_2c_4)^2+(1-B_1B_2)^2(c_1c_4-1)}.\label{eq:prodr}
\end{align}
It will be convenient at this point to work with $d=\Delta-1$. We plug the first equation in \eqref{secab} into \eqref{eq:prodr}, yielding
\begin{equation}\label{recurone}
\begin{aligned}
(r_1r_4)^{1/d}-1&=\frac{(1-B_1B_2)^2(c_1c_4-1)}{(B_1c_1+(B_1B_2+1)+B_2c_4)^2},\\
 (c_1c_4)^{1/d}-1&=\frac{(1-B_1B_2)^2(r_1r_4-1)}{(B_1r_1+(B_1B_2+1)+B_2r_4)^2}.
\end{aligned}
\end{equation}
The second equality in \eqref{recurone} follows by a completely symmetric argument for $c_1,c_4$. Observe that \eqref{recurone} implies that one of the following three cases can hold.
\[\mbox{\sc{Case I}}: r_1 r_4=c_1c_4=1, \   \mbox{\sc{Case II}}: r_1 r_4>1,c_1c_4>1, \ \mbox{\sc{Case III}}: r_1 r_4<1,c_1c_4<1. \] 
\mbox{\sc{Case I}} reduces to $\gamma=\alpha^2, \ \delta=\beta^2$ by \eqref{secab}. Thus we may focus on \mbox{\sc{Cases II}} and \mbox{\sc{III}}. We further restrict our attention to \mbox{\sc{Case II}}, with \mbox{\sc{Case III}} being completely analogous.
From AM-GM, we have the inequalities 
\[B_1c_1+B_2c_4\geq 2\sqrt{B_1B_2c_1 c_4}, \ B_1r_1+B_2r_4\geq 2\sqrt{B_1B_2r_1 r_4},\]
so \eqref{recurone} gives
\begin{gather*}
(r_1r_4)^{1/d}-1\leq\frac{(1-B_1B_2)^2(c_1c_4-1)}{\big(B_1B_2+1+2\sqrt{B_1B_2c_1c_4}\big)^2}, \\
(c_1c_4)^{1/d}-1\leq\frac{(1-B_1B_2)^2(r_1r_4-1)}{\big(B_1B_2+1+2\sqrt{B_1B_2r_1r_4}\big)^2}.
\end{gather*}
It is straightforward to verify the identity \[\Big(B_1B_2+1+2z\sqrt{B_1B_2}\Big)^2+(1-B_1B_2)^2(z^2-1)=\Big((1+B_1B_2)z+2\sqrt{B_1 B_2}\Big)^2.\] Set $x=\sqrt{r_1 r_4}$, $y=\sqrt{c_1c_4}$. Using the identity with $z=x,y$ and taking square roots, we obtain
\begin{equation}
x^{1/d}\leq\frac{(1+B_1B_2)y+2\sqrt{B_1 B_2}}{2\sqrt{B_1B_2}y+B_1B_2+1}, \quad
y^{1/d}\leq\frac{(1+B_1B_2)x+2\sqrt{B_1 B_2}}{2\sqrt{B_1B_2}x+B_1B_2+1}.\label{ineq:recura}
\end{equation}
Using the substitution $B'=(1+B_1 B_2)/2\sqrt{B_1B_2}$, this can be rewritten in the form
\begin{equation}
x^{1/d}\leq\frac{B'y+1}{y+B'}, \quad
y^{1/d}\leq\frac{B'x+1}{x+B'}.\label{ineq:recurb}
\end{equation}
The astute reader will immediately realize the analogy of
\eqref{ineq:recurb} with the tree equations \eqref{eq:densitiesone}
for the Ising model without external field.
Moreover, since $B'\geq 1$ we are in the case of the {\em
ferromagnetic} Ising model.
For the ferromagnetic Ising model on the
infinite tree $\Tree_{d+1}$, we have uniqueness when $1<B\leq
\frac{d+1}{d-1}$.
Hence intuitively the lemma holds when $B'\leq \frac{d+1}{d-1}$. However, in
\eqref{ineq:recurb} we have an inequality version of the tree
recursions \eqref{eq:densitiesone} and therefore a bit more work is required.

In our setting, we obtain that \eqref{ineq:recurb} cannot hold when $x>1,y>1$. To see this, multiply the inequalities in \eqref{ineq:recurb} to obtain
\begin{equation}\label{eq:stablea}
x^{1/d}y^{1/d}\leq \frac{B'x+1}{x+B'}\cdot \frac{B'y+1}{y+B'}.
\end{equation}
Note that \eqref{eq:stablea}  is reversed in \mbox{\sc{Case III}}. The following lemma implies that \eqref{eq:stablea} cannot hold in \mbox{\sc{Case II}}, and similarly for the reverse inequality \eqref{eq:stablea} in  \mbox{\sc{Case III}}. 
\begin{lemma}\label{lem:maininequality}
Let $d\geq 2$. When $B'\leq 1/B_c(\Tree_{d+1})=\frac{d+1}{d-1}$, for $z>1$ it holds that $z^{1/d}> \displaystyle \frac{B'z+1}{z+B'}$. The inequality is reversed for $z<1$.
\end{lemma}
We give the proof of Lemma~\ref{lem:maininequality} after observing that $B'\leq \frac{d+1}{d-1}$ reduces to $\sqrt{B_1 B_2}\geq \frac{\sqrt{d}-1}{\sqrt{d}+1}$. This concludes the proof of Lemma~\ref{lem:seccritical}.
\end{proof}
\begin{proof}[Proof of Lemma~\ref{lem:maininequality}.]
The case $z<1$ reduces to the case $z>1$ using the inversion $z\leftarrow 1/z$. Thus we focus on proving the Lemma for $z>1$.  
Let $f(w):=w^{d+1}-B' w^d+B' w-1$. The lemma reduces to proving  $f(w)>0$ when $w>1$ (under the substitution  $w=z^{1/d}$). We have
\begin{align*}
f'(w)&=(d+1)w^d-dB'w^{d-1}+B'.\\
f''(w)&=dw^{d-2}\big((d+1)w-(d-1)B'\big).
\end{align*}
The assumptions imply that $(d+1)-(d-1)B'\geq 0$. Hence, when $w>1$, we have $f''(w)>f''(1)$. Note that $f''(1)\geq 0$ only if $(d+1)-(d-1)B'\geq 0$. Thus, $f''(w)$ is strictly positive when $w>1$ and hence $f'(w)$ is strictly increasing. It follows that for $w>1$, we have $f'(w)> f'(1)$ and
\[f'(1)=(d+1)-(d-1)B'\geq0.\]
Thus, $f$ is strictly increasing for $w>1$, yielding $f(w)>f(1)=0$.
\end{proof}

\subsection{Some Remarks}\label{sec:remarks}
Here we give the idea behind the arguments establishing Lemmas~\ref{lem:secondmaxising} and~\ref{lem:secondmaxhardcore}. As observed in Section~\ref{sec:logsecmax}, these lemmas boil down to proving that, for $(\alpha,\beta)=(p^+,p^-)$, the only critical point of $\phi_2(\gamma,\delta)$ is at $(\gamma,\delta)=(\alpha^2,\beta^2)$ under the respective hypotheses.

Our point of departure is once again the equations in \eqref{recurone}, albeit we proceed with a more thorough analysis. Looking at the argument for the case $\frac{\sqrt{d}-1}{\sqrt{d}+1}\leq \sqrt{B_1B_2}$, the point of the proof which is subject to tighter analysis is the use of the AM-GM inequalities. These are weak if the ratios $r_1/r_4,c_1/c_4$ are much bigger than $B_2/B_1$. In Section~\ref{sec:secondmaxisingproof}, we turn this weakness into our favor. Namely, for the antiferromagnetic Ising model without external field, in the region $0<B<\frac{\sqrt{d}-1}{\sqrt{d}+1}$, the densities $p^\pm$ are heavily biased towards 1 and 0 respectively, and this reflects at the values of $r_1,r_4,c_1,c_4$. This observation can be turned into a simple proof of Lemma~\ref{lem:secondmaxising}. For the proof of Lemma~\ref{lem:secondmaxhardcore}, it is easier to get an explicit handle on the values $p^\pm$ (and thus on the bias of the $r_i,c_j$), which can be used to analyze the equations \eqref{recurone} (see Section~\ref{sec:hardhard}). 

Further, let us make some small minor observations. The reader may have noticed that the proof of Lemma~\ref{lem:seccritical} applies in a more general setup than the one stated in Condition~\ref{cond:maxima}. Namely, in the region $\sqrt{B_1B_2}\geq \frac{\sqrt{d}-1}{\sqrt{d}+1}$, Lemma~\ref{lem:seccritical} holds  for all values of $\alpha,\beta$ such that $1>\alpha,\beta>0$ and not just for $(\alpha,\beta)=(p^\pm,p^\mp)$ (note, this range of $B_1,B_2$ covers the whole uniqueness region even with external field). It is thus conceivable that, when $\sqrt{B_1B_2}\geq \frac{\sqrt{d}-1}{\sqrt{d}+1}$, the bound on the partition function $Z^{\alpha,\beta}_G$ in Lemma~\ref{lem:smallgraph} holds for $1>\alpha,\beta>0$  (we do not attempt to show a more general version of Lemma~\ref{lem:smallgraph} in
this paper since we are only interested in the values of $(\alpha,\beta)$ which maximize $\E_\G[Z^{\alpha,\beta}_G]$).

\section{NP-Hardness Results}
\label{sec:NP-outline}
In this section, we give the proofs for our NP-hardness results.  We will prove the following theorem, which allows us to prove Theorems~\ref{thm:hardcore},~\ref{thm:Ising},~\ref{thm:2spin}.
\begin{theorem}\label{thm:mainNP}
Let $\Delta\geq 3$. Assume that $(B_1,B_2,\lambda)$ lie in the non-uniqueness region of the infinite tree $\TD$. Moreover, assume that Condition~\ref{cond:maxima} holds. Then, unless NP$=$ RP, there is no FPRAS for the partition function of the 2-spin model with parameters $B_1,B_2,\lambda$ in graphs with maximum degree $\Delta$.
\end{theorem}
With Theorem~\ref{thm:mainNP} at hand, it is straightforward to prove  Theorems~\ref{thm:hardcore},~\ref{thm:Ising},~\ref{thm:2spin}.

\begin{proof}[Proofs of  Theorems~\ref{thm:hardcore},~\ref{thm:Ising},~\ref{thm:2spin}.]
Theorem~\ref{thm:2spin} follows immediately from Lemma~\ref{lem:secondmax2spin} and Theorem~\ref{thm:mainNP}.
Theorem~\ref{thm:hardcore} was known to hold for $\Delta=3$ and $\Delta\geq 6$ (\cite{Sly10}, \cite{Galanis}). Lemma~\ref{lem:secondmaxhardcore} in combination with Theorem~\ref{thm:mainNP} prove the cases $\Delta=4,5$ as well. 

We now prove Theorem~\ref{thm:Ising}. In the case of Ising model with no external field, note that Lemma~\ref{lem:secondmax2spin} and Theorem~\ref{thm:mainNP} give hardness for $\Delta\geq 3$ when 
\begin{equation}
\frac{\sqrt{d}-1}{\sqrt{d}+1}\leq B<\frac{d-1}{d+1} \mbox{ where } d=\Delta-1.\label{eq:intervals}
\end{equation}
The theorem will be proved once we show that this can be extended to the region $0<B<\frac{\sqrt{d}-1}{\sqrt{d}+1}$. Observe that hardness for $B,\Delta$ gives hardness for $B,\Delta+1$. It is easy to check that for $d\geq 2$, it holds that
\[\frac{\sqrt{d+1}-1}{\sqrt{d+1}+1}<\frac{d-1}{d+1}<\frac{d}{d+2},\]
so that the intervals corresponding to $\Delta$ and $\Delta+1$ are overlapping for $\Delta\geq 3$. Thus it suffices to check NP hardness for $d=2$ and $0<B<\frac{\sqrt{d}-1}{\sqrt{d}+1}$. This follows from Lemma~\ref{lem:secondmaxising} and Theorem~\ref{thm:mainNP}, thus completing the proof.
\end{proof}

We thus focus on proving Theorem~\ref{thm:mainNP}. Sly's reduction \cite{Sly10}, while presented for the particular case of the hard-core model, can be used to show hardness for  general antiferromagnetic 2-spin models. The only part in Sly's reduction which needs modifications is in establishing the properties of the gadget he uses. Let us first describe the gadget used in \cite{Sly10}. 

The construction of the gadget has two parameters $0<\theta,\psi<1/8$. Let $k:=(\Delta-1)^{\lfloor \theta \log_{\Delta-1}n\rfloor}$ and $\ell:=2\lfloor \frac{\psi}{2}\log_{\Delta-1}n\rfloor$ and let $m':=k(\Delta-1)^{\ell}$. Note that $m'=o(n^{1/4})$. The gadget is constructed in two steps: first, a random bipartite graph $\overline{G}$ 
with $n+m'$ vertices on each side is constructed. The two sides of the graph will be labelled with $+,-$. For $s\in\{+,-\}$, let the vertices on the $s$-side be $W_s\cup U_s$, where
$|W_s|=n$ and $|U_s|=m'$. The edges of the graph are the
union of $\Delta-1$ uniformly random perfect matchings between $W_+\cup U_+$ and
$W_-\cup U_-$, together with a uniformly random perfect
matching between $W_+$ and $W_-$. Thus, in the resulting random graph $\overline{G}$ , the vertices
in $W:=W_+\cup W_-$ have degree
$\Delta$ and the vertices in $U:=U_+\cup U_-$ have degree $\Delta-1$. We denote the graph distribution defined by this construction as $\overline{\mathcal{G}}$.

	The second part of the construction appends complete trees of depth $\ell$ to $\overline{G}$  in the following manner.  For $s\in\{+,-\}$, partition the vertices of $U_s$ into $k$ groups of $(\Delta-1)^\ell$ vertices.  For each group we create
a new $(\Delta-1)$-ary tree of even depth $\ell$ (where $\ell$ is an even integer as specified earlier). The leaves of this tree are the group of vertices in $U_s$, the other vertices
of the tree are new.  After the addition of these $2k$ trees to $\overline{G}$ the vertices of $U_s$ have degree $\Delta$; in fact, all vertices in this new graph have degree $\Delta$ except for the roots of the  trees which have degree $\Delta-1$.  Let $R_+$ (respectively, $R_-$) denote the roots of those trees whose leaves are a subset of $U_+$ ($U_-$).  Let $R:=R_+\cup R_-$. Informally, the addition of the trees makes it easier to prove stronger concentration properties for the spins of the vertices with degree $\Delta-1$ (which are now the roots of the trees).

Analogously to \cite{Sly10}, for a configuration $\sigma$ on $H$, the phase $Y(\sigma)$ of $\sigma$ is defined to be $+$ if the number of vertices assigned $-1$ in $W_+$ is greater than the number of vertices assigned $-1$ in $W_-$, otherwise $Y(\sigma)=-$. In other words, the phase $Y(\sigma)$ of a configuration $\sigma$ simply points to the set  between  $W_+,W_-$ with the greatest number of  vertices assigned $-1$ in $\sigma$. 

The properties of Sly's gadget are stated in the following lemma (proved in Section~\ref{sec:sketchproof}). We use the following notation: for a configuration $\sigma$ on $H$, $\sigma_R$ denotes the restriction of $\sigma$ to vertices in $R$. Moreover, for a spin $s\in\{-1,+1\}$, $\sigma^{-1}(s)$ denotes the vertices assigned the spin $s$.

\begin{lemma}[Analogue of {\cite[Theorem 2.1]{Sly10}}]\label{lem:gadget2}
Let $\Delta\geq 3$. Assume that $(B_1,B_2,\lambda)$ lie in the non-uniqueness region of the infinite tree $\TD$. Moreover, assume that Condition~\ref{cond:maxima} holds. There exist constants $\theta(\Delta,B_1,B_2,\lambda)$, $\psi(\Delta,B_1,B_2,\lambda)>0$ such that for all sufficiently large $n$  the graph $H$ satisfies with probability $1-o_n(1)$ the following:
\begin{enumerate}
\item The two phases occur with roughly equal probability, i.e., for $i\in\{+,-\}$, we have \label{it:phaseprob2}
\[\mu_H\big(Y(\sigma)=i\big)\geq \frac{1}{n}.\]
\item Conditioned on the phase $i\in\{+,-\}$, the spins of vertices in $R$ are approximately independent, that is,   \label{it:approxind2}
\[\max_{\eta\in \{-1,+1\}^{R}}\Big|\frac{\mu_H\big(\sigma_R=\eta\, |\, Y(\sigma)=i\big)}{Q_R^{i}(\eta)}-1\Big|\leq n^{-2\theta},\]
\end{enumerate}
where $Q^{i}_R$ is the following product distribution on configurations $\eta: R_+\cup R_-\rightarrow \{\pm 1\}$:
\[Q^{i}_R(\eta)=
(q^i)^{ |\eta^{-1}(-1)\cap R_+| }
(1-q^i)^{ |\eta^{-1}(+1)\cap R_+| }
(q^{-i})^{ |\eta^{-1}(-1)\cap R_-| }
(1-q^{-i})^{ |\eta^{-1}(+1)\cap R_-| }.
\]
Recall that for $i\in\{+,-\}$, $q^i$ is the probability that the root of the infinite $(\Delta-1)$-ary tree $\TDary$ is assigned the spin $-1$ in the extremal measure $\hat{\mu}^{j}$ (see Section~\ref{sec:tree-recursions}).
\end{lemma}


We now describe briefly how these properties  of $H$ are established in \cite{Sly10}. This is done in two steps: 
\begin{enumerate}
\item First prove that in the graph $\overline{G}$, conditioned on the phase $i\in \{+,-\}$, the spins of the vertices in $U$ are asymptotically (jointly) distributed as a product distribution $Q^i_U$ on $\{-1,+1\}^{U}$. The latter distribution is obtained by replacing $R$ with $U$ in the definition of $Q^i_R$ in Lemma~\ref{lem:gadget2}. Roughly, this can be obtained as follows from what we have seen so far: in Section~\ref{sec:proof-approach} we showed that for a random $\Delta$-regular bipartite graph, with high probability over the choice of the graph, in the Gibbs distribution the total weight of configurations from the set $\Sigma^{p^+,p^-}$ dominate exponentially over the rest. An analogous phenomenon occurs for a graph $\overline{G}\sim \overline{\G}$, which allows to  quantify the aforementioned behavior, see Section~\ref{sec:modifiedaaa} for more details.
\item Translate the results to $R$. This is accomplished in a crafty way in \cite{Sly10}: the product measures $Q^{i}_U$ on $\{-1,+1\}^{U}$  are analogous to the even and odd boundaries in the infinite $(\Delta-1)$-ary tree, providing for an elegant translation of the results (see the proof of the second part of Lemma~\ref{lem:gadget2} for more details). 
Moreover, the extremality of the measures implies that the roots of the trees, conditioned on the phase, are strongly concentrated around their  expected value. The latter part follows after studying non-reconstruction in the extremal measures and proving that the correlation of the spins of the vertices in $R$ to the spins of the vertices in $U$ decays doubly exponentially in $\ell$, see Section~\ref{sec:reco}. Crucially, this allows to bound the distance between the (true) distribution of the spins of the vertices in $R$ (conditioned on the phase $i$) and the product distribution $Q^i_R$ by an inverse polynomial function of $n$. 
\end{enumerate}
In Section~\ref{sec:modifiedaaa}, we give the lemmas in \cite{Sly10} which require adjustment in the 2-spin model. Their proofs are analogous to those in \cite{Sly10}, by comparing the quantities in the graph distribution $\overline{\mathcal{G}}$ with those in the distribution $\G$. In Section~\ref{sec:reco}, we describe the non-reconstruction results that Sly uses. Finally, in Section~\ref{sec:sketchproof}, we give the proof of Lemma~\ref{lem:gadget2} and conclude Theorem~\ref{thm:mainNP}, checking that Sly's arguments  go through in our setting as well.

\subsection{The Partition Function for Graphs Sampled from $\overline{\mathcal{G}}$}\label{sec:modifiedaaa}
In this section, we give an outline of the lemmas which are needed to extend Sly's reduction for general 2-spin models. 

Let $\eta$ be an assignment of spins to $U$, i.e., $\eta\in\{-1,+1\}^U$. For a graph $\overline{G}\sim\overline{\G}$, let $Z^{\alpha,\beta}_{\overline{G}}(\eta)$ be the partition function over
configurations which agree with $\eta$ on $U$, i.e., $Z^{\alpha,\beta}_{\overline{G}}(\eta)$ is the total weight of configurations in $\sigma\in\{-1,+1\}^{V(\overline{G})}$ subject to $|N_-(\sigma)\cap W_+|=\alpha n$, $|N_-(\sigma)\cap W_-|=\beta n$ and $\sigma_U=\eta$.

Denote by $\eta_1^+$ the number of vertices in $U_+$ which are assigned spin $+1$ and let $\eta_1^- = m'-\eta_i^+$ (the number of vertices in $U_+$ that are assigned spin $-1$). Define similarly $\eta_2^+,\eta_2^-$ for $U_-$. We have the following analogue of \cite[Lemma 3.1]{Sly10}. 
\begin{lemma}\label{lem:modfirst}
For any $\alpha,\beta$ and all $\eta\in\{-1,+1\}^U$ we have
\begin{multline*}
\frac{\E_{\overline{\G}}[Z^{\alpha,\beta}_{\overline{G}}(\eta)]}{\E_\G[Z^{\alpha,\beta}_{G}]}=\\ (1+o(1))
\left(\frac{\alpha^{\eta_1^-}(1-\alpha)^{\eta_1^+}\beta^{\eta_2^-}(1-\beta)^{\eta_2^+}}{(\alpha-x^*)^{\eta_1^-}
(1-\alpha-\beta+x^*)^{\eta_1^+-\eta_2^-}
(\beta-x^*)^{\eta_2^-}}B_2^{\eta_1^+ - \eta_2^-}
\right)^{\Delta-1}\lambda^{\eta_1^-+\eta_2^-},
\end{multline*}
where $x^*$ is the (unique) non-negative solution of the equation $B_1B_2(\alpha-x)(\beta-x)=x(1-\alpha-\beta+x)$. In particular, when $(\alpha,\beta)=(p^+,p^-)$, it holds that
\begin{equation}\label{eq:productmeasure}
\frac{\E_{\overline{\G}}[Z^{\alpha,\beta}_{\overline{G}}(\eta)]}{\E_\G[Z^{\alpha,\beta}_{G}]}= (1+o(1))\frac{C^*}{(1-q^+)^{m'}(1-q^-)^{m'}}(q^+)^{\eta_1^-}(1-q^+)^{\eta_1^+}(q^-)^{\eta_2^-}(1-q^-)^{\eta_2^+},
\end{equation}
where 
\begin{equation}\label{eq:cstar}
C^* = \left(\frac{(B_2+Q^+)(B_2+Q^-)}{B_2+Q^++Q^-+B_1Q^+Q^-}\right)^{(\Delta-1)m'}.
\end{equation}
\end{lemma}

\noindent The proof of Lemma~\ref{lem:modfirst} is given in Section~\ref{sec:modfirstproof}. Seeking high probability lower bounds on $Z^{\alpha,\beta}_{\overline{G}}(\eta)$, we study the second moment of $Z^{\alpha,\beta}_{\overline{G}}(\eta)$.
\begin{lemma}[analogue of Lemma 3.5 in~\cite{Sly10}]\label{lem:modsecond}
When Condition~\ref{cond:maxima} is true, for $(\alpha,\beta)=(p^+,p^-)$ and all $\eta\in\{-1,+1\}^U$ we have
\begin{equation}
\frac{\E_{\overline{\G}}\big[\big(Z^{\alpha,\beta}_{\overline{G}}(\eta)\big)^2\big]}{\E_{\G}\big[\big(Z^{\alpha,\beta}_{G}\big)^2\big]}= (1+o(1))
(C^*)^2 \left(\lambda\left(\frac{1+B_1Q^-}{B_2+Q^-}\right)^{\Delta-1}\right)^{2\eta_1^-}
\left(\lambda\left(\frac{1+B_1Q^+}{B_2+Q^+}\right)^{\Delta-1}\right)^{2\eta_2^-}\label{eq:secondmodratio}
\end{equation}
where $Q^-,Q^+$ are given by~\eqref{eq:qpqm} and
$C^*$ is given by~\eqref{eq:cstar}.
\end{lemma}
\noindent The proof of Lemma~\ref{lem:modsecond} is given in Section~\ref{lem:modsecondproof}. Once again, it turns out that the second moment method fails to give high probability bounds, as is captured by the following lemma.
\begin{lemma}~\label{lem:rfsgadget}
When Condition~\ref{cond:maxima} is true, for $(\alpha,\beta)=(p^+,p^-)$ it holds that
\begin{equation*}
\lim_{n\rightarrow\infty}\frac{\E_{\overline{\G}}\big[\big(Z^{\alpha,\beta}_{\overline{G}}(\eta)\big)^2\big]}{\left(\E_{\overline{\G}}[Z^{\alpha,\beta}_{\overline{G}}(\eta)]\right)^2}=\big(1-\omega^2\big)^{-(\Delta-1)/2}\big(1-(\Delta-1)^2\omega^2\big)^{-1/2},
\end{equation*}
where $\omega$ is given by \eqref{eq:definitionofomega}.
\end{lemma}

\noindent The proof of Lemma~\ref{lem:rfsgadget} is given in Section~\ref{sec:rfsgadgetproof}. We then apply the small subgraph conditioning method to obtain the following analogue of \cite[Lemma 3.9]{Sly10}.
\begin{lemma}~\label{lem:smallgraphgadget}
For $\Delta\geq 3$, suppose that for $B_1,B_2, \lambda$ in the non-uniqueness region, Condition~\ref{cond:maxima} is true. Then, for $(\alpha,\beta)=(p^+,p^-)$,  asymptotically almost surely over the choice of the graph $\overline{G}\sim\overline{\G}$, for  $\eta\in \{-1,+1\}^U$, it holds that 
\[ Z^{\alpha,\beta}_{\overline{G}}(\eta)>\frac{1}{\sqrt{n}}\E_{\overline{\G}}[Z^{\alpha,\beta}_{\overline{G}}(\eta)].\]
\end{lemma}
\noindent The proof of Lemma~\ref{lem:smallgraphgadget} is given in Section~\ref{sec:smallgraphregular}.

\subsection{The Reconstruction Problem in the Extremal Measures}\label{sec:reco}
The final ingredient we need is proving that a strong form of non-reconstruction holds in the extremal measures. Let us first introduce some notation. 

Denote by $\hat{T}_{\rho,\ell,\Delta}$ the subtree of $\TDary$ rooted at $\rho$ up to depth $\ell$. Let  $S_{\rho,\ell}$ be the leaves of $\hat{T}_{\rho,\ell,\Delta}$. Denote by $X_{\rho,\ell,+}$ the marginal probability that the root $\rho$ is
occupied in a configuration $\sigma$ generated by the following process: 
\begin{enumerate}
\item Sample a configuration $\hat{\sigma}$ from the projection of the measure $\hat{\mu}^+$ on $\hat{T}_{\rho,\ell,\Delta}$, 
\item Conditioning on the configuration $\hat{\sigma}_S$ on $S_{\rho,\ell}$, 
sample a configuration $\sigma$ from the Gibbs distribution on $\hat{T}_{\rho,\ell,\Delta}$. 
\end{enumerate}
Note that the configuration of the vertices in $S_{\rho,\ell}$ is a random vector, so that $X_{\rho,\ell,+}$ is a random variable. Define similarly $X_{\rho,\ell,-}$. We need the following analogue of \cite[Lemma 4.2]{Sly10}.
\begin{lemma}\label{lem:Sly42}
Assume that $\Delta\geq 3$ and $B_1,B_2,\lambda$ lie in the non-uniqueness region of $\TD$. There exist constants $\zeta_1(\Delta,B_1,B_2,\lambda), \zeta_2(\Delta,B_1,B_2,\lambda) > 0$, such that for all sufficiently large $\ell$, 
$$\mathbb{P}\big[\big\vert X_{\rho,\ell,\pm} - q^\pm\big\vert  \geq \exp(-\zeta_1\ell)\big]\leq \exp(-\exp(\zeta_2\ell)).$$
\end{lemma}
\begin{proof} 
Martinelli, Sinclair, and Weitz \cite[Proof of Theorem 5.1, Equation (21)]{Martinelli} prove the required strong concentration for the variables $X_{\rho,\ell,\pm}$ for the Ising model with arbitrary external field (in a different form than the one stated here but completely equivalent). Their results can be translated to the general 2-spin model, after carrying out the reduction described in \cite{SST}. Finally, in the case of hard constraints Lemma~\ref{lem:Sly42} follows by the respective result in \cite[Lemma 4.2]{Sly10} (slightly extended in \cite[Lemma 30]{Galanis}). Of course, one could also apply the methods in \cite{Martinelli} or \cite{Sly10} to directly get the result in general.
\end{proof}

\subsection{Finishing the Proofs of the Hardness Results}\label{sec:sketchproof}
We first give the proof of Lemma~\ref{lem:gadget2}; this is  a straightforward extension of the proof of \cite[Theorem 2.1]{Sly10}, so we give a desciption of the main elements and point to the relevant lemmas which verify that Sly's proof goes through in our setting as well. 

\def\Go{\overline{G}}
\def\Gc{\mathcal{G}}
\def\Gco{\overline{\mathcal{G}}}
\def\PrGco{\mathrm{Pr}_{\Gco}}
\begin{proof}[Proof of Lemma~\ref{lem:gadget2}.]
We first introduce some notation that is used by Sly \cite{Sly10}. Define $\Sigma^{\pm}$ to be the configurations $\sigma\in \{-1,+1\}^{V(\Go)}$ such that $|N_-(\sigma)\cap W_{\pm}|\geq |N_-(\sigma)\cap W_\mp|$, that is, a configuration in $\Sigma^{+}$ assigns the spin $-1$ to more vertices in $W_{+}$ than in $W_{-}$ (and conversely for $\Sigma^-$). Similarly, for $\eta\in\{-1,+1\}^U$, define $\Sigma^{\pm}(\eta)$ to be the configurations $\sigma$ in $\sigma\in \Sigma^{\pm}$ such that $\sigma_U=\eta$. Further, for $G\sim \Go$, define
\[Z^{\pm}_{\Go}(\eta)=\sum_{\sigma\in \Sigma^{\pm}(\eta)}w_G(\sigma), \quad Z^{\pm}_{\Go}:=\sum_{\sigma\in \Sigma^{\pm}}w_G(\sigma).\]
Note that $Z^{\pm}_{\Go}=\sum_{\eta\in\{-1,+1\}^U}Z^{\pm}_{\Go}(\eta)$. Define similarly $Z^{\pm}_G$ for a random $\Delta$-regular graph $G$, i.e., for $G\sim \Gc$, let $Z^+_G=\sum_{\alpha\geq \beta}Z^{\alpha,\beta}_G$ and $Z^-_G=\sum_{\alpha< \beta}Z^{\alpha,\beta}_G$.

Using that $(\alpha,\beta)=(p^\pm,p^\mp)$ maximizes the logarithm of $\E_{\Gc}[Z^{\alpha,\beta}_G]$ (see Lemma~\ref{lem:firstmomentmax}), the first item of Lemma~\ref{lem:modfirst} together with Markov's inequality, shows that the largest contribution to   $\E_{\Gco}[Z^{\pm}_{\Go}(\eta)]$ is given by $\E_{\Gco}[Z^{p^{\pm},p^{\mp}}_{\Go}(\eta)]$. This gives
\begin{equation}\label{eq:aazzss}
\E_{\Gco}[Z^{\pm}_{\Go}(\eta)]=(1+o(1))(C')^{m'}Q^{\pm}_U(\eta)\E_{\Gc}[Z^{\pm}_G],
\end{equation}
where the last equality uses the conclusion \eqref{eq:productmeasure} of Lemma~\ref{lem:modfirst} (the constant $C'$ can be inferred from \eqref{eq:productmeasure}, though it will not be important for what follows). The derivation of \eqref{eq:aazzss} is identical to the proof of \cite[Lemma 3.3, Equation (3.9)]{Sly10}. Summing \eqref{eq:aazzss} over $\eta\in \{-1,+1\}^{U}$ yields
\begin{equation}\label{eq:aazzssb}
\E_{\Gco}[Z^{\pm}_{\Go}]=(1+o(1))(C')^{m'}\E_{\Gc}[Z^{\pm}_G].
\end{equation}
Our results from the small subgraph conditioning method (see Section~\ref{sec:oversmall}), combined with the above observations, yield that $Z^{\pm}_{\Go}(\eta)$ is close to its expectation, in particular 
\begin{equation}\label{eq:wertyb}
\lim_{n\rightarrow\infty}\max_{\eta\in\{-1,+1\}^U} \PrGco\left(Z^{\pm}_{\Go}(\eta)<\frac{1}{\sqrt{n}}\E_{\Gco}[Z^{\pm}_{\Go}(\eta)]\right)=0.
\end{equation}
see the proof of \cite[Theorem 3.10]{Sly10} for details. 

At this point, one needs to transfer these results to the graph $H$. Define $Z_H^{\pm}$ to be the contribution to the partition function of $H$ from configurations with phase $\pm$ (note that the restriction of such configurations on $\Go$ is a configuration in $\Sigma^{\pm}$). To connect the partition function of $H$ with the partition function of $\Go$, let $T$ be the graph induced by the edges $E(H)\backslash E(\Go)$. In particular, $T$ is a union of disjoint trees and the union of the leaves of the trees is precisely the set $U$. For $\eta\in\{-1,+1\}^U$, define $Z_T(\eta)$ to be the contribution to the partition function of $T$ from configurations $\sigma\in \{-1,+1\}^T$ such that $\sigma_U=\eta$. These definitions allow to write 
\[Z^{\pm}_H=\sum_{\eta\in\{-1,+1\}^U}Z^{\pm}_{\Go}(\eta)Z_T(\eta).\]
It follows by \eqref{eq:wertyb} and Markov's inequality that 
\begin{equation*}
\begin{aligned}
\lim_{n\rightarrow\infty}\PrGco\left(\sum_{\eta\in\{-1,+1\}^U}Z^{\pm}_{\Go}(\eta)Z_T(\eta)<\frac{1}{2\sqrt{n}}\sum_{\eta\in\{-1,+1\}^U}Z_T(\eta)\E_{\Gco}[Z^{\pm}_{\Go}(\eta)]\right)=0,\\
\lim_{n\rightarrow\infty}\PrGco\left(\sum_{\eta\in\{-1,+1\}^U}Z^{\pm}_{\Go}(\eta)Z_T(\eta)>\frac{\sqrt{n}}{3}\sum_{\eta\in\{-1,+1\}^U}Z_T(\eta)\E_{\Gco}[Z^{\pm}_{\Go}(\eta)]\right)=0,
\end{aligned}
\end{equation*}
see also \cite[(4.22)\ \&\ (4.23)]{Sly10}. It follows that $\mu_H(Y(\sigma)=\pm)=Z_H^{\pm}/Z_H\geq 1/n$ a.a.s.\ over the choice of the graph $\Gc$. This proves the first item in Lemma~\ref{lem:gadget2}.

The proof of the second item in Lemma~\ref{lem:gadget2} utilizes to a greater extent the connection to the tree and relies crucially on the non-reconstruction results we overviewed in Section~\ref{sec:reco}. First observe that 
\begin{equation}\label{eq:Hproductmeasure}
\mu_H(\sigma_U=\eta\, |\, Y(\sigma)=\pm)=\frac{Z^{\pm}_{\Go}(\eta)Z_T(\eta)}{Z_H^{\pm}}=\frac{Z^{\pm}_{\Go}(\eta)Z_T(\eta)}{\sum_{\eta'\in\{-1,+1\}^U}Z^{\pm}_{\Go}(\eta')Z_T(\eta')}.
\end{equation}
Now, provided that $Z^{\pm}_{\Go}(\eta)<\frac{1}{\sqrt{n}}\E_{\Gco}[Z^{\pm}_{\Go}(\eta)]$ (see \eqref{eq:wertyb}), the right hand side of \eqref{eq:Hproductmeasure} is (with large probability) within a multiplicative factor $n$ from 
\[\nu^{\pm}(\eta):=\frac{Z_T(\eta)\E_{\Gco}[Z^{\pm}_{\Go}(\eta)]}{\sum_{\eta'\in\{-1,+1\}^U}Z_T(\eta')\E_{\Gco}[Z^{\pm}_{\Go}(\eta')]}=\frac{Z_{T}(\eta)Q^{\pm}_U(\eta)}{\sum_{\eta'\in\{-1,+1\}^U}Z_T(\eta')Q^{\pm}_U(\eta')},\]
where in the first equality we used \eqref{eq:aazzssb}. The crucial idea, captured in the proof of \cite[Lemma 4.3]{Sly10}, is that the measures $\nu^{\pm}$ on configurations $\{-1,+1\}^U$ are the same as the distributions induced by the random processes we described in Section~\ref{sec:reco}, see the proof of \cite[Lemma 4.3]{Sly10}  for more details. This allows to use Lemma~\ref{lem:Sly42} together with \eqref{eq:wertyb} to bound the probability mass of ``bad" configurations $\eta$ on $U$, i.e., configurations on $U$ which exert large influence on the roots $R$ of the trees, see \cite[Proof of Theorem 2.1]{Sly10}. This yields Item 2 in Lemma~\ref{lem:gadget2}, concluding the argument.
\end{proof}

With Lemma~\ref{lem:gadget2} at hand, Theorem~\ref{thm:mainNP} can be proved completely analogously to \cite[Proof of Theorems 1 \& 2]{Sly10}.
\begin{proof}[Proof of Theorem~\ref{thm:mainNP}.]
Use  Lemma~\ref{lem:gadget2} and the reduction in \cite[Section 2.2]{Sly10}. 
\end{proof}

\section{Calculating Moments}
\label{sec:calculating-moments}
\subsection{Random Bipartite $\Delta$-Regular Graphs}\label{sec:ratiofirstsecondproof}
We use the information we have obtained for the first and second moments to prove Lemma~\ref{lem:ratiofirstsecond}. Its proof is a technical calculation. 

Recall that $Q^\pm=\frac{q^\pm}{1-q^\pm}$. To aid the presentation, it will be convenient to define the following expressions, which will appear in the analysis of the asymptotics of the first and second moments.
\begin{equation}\label{eq:helpfulAs}
\begin{aligned}
E_1&:=B_1 Q^{-} Q^{+}+Q^{-}+Q^{+}+B_2,\\
E_2&:=B_1 Q^{-} Q^{+}+B_1 B_2(Q^{+}+Q^{-})+B_2,\\
E_3&:=(1-B_1 B_2)^2 Q^{-}Q^{+}\\
&\qquad \quad+\big(1+B_1(Q^{+}+Q^{-}) + B_1^2Q^{-}Q^{+}\big)\big(B^2_2+B_2(Q^{-}+Q^{+})+Q^{-}Q^{+}\big).
\end{aligned}
\end{equation}
We are also going to need the explicit values of $\mathbf{X}^*=\arg\max_{\mathbf{X}}g_1(\mathbf{X})$ and $\mathbf{Y}^*=\arg\max_{\mathbf{X}}g_2(\mathbf{Y})$ when $\alpha=p^+,\beta=p^-,\gamma=\alpha^2,\delta=\beta^2$. These are given by:
\begin{gather}
\left(\begin{array}{cc}x^{*}_{11}&x^{*}_{12}\\x^{*}_{21}&x^{*}_{22}\end{array}\right)=\frac{1}{E_1}\left(\begin{array}{cc}B_1 Q^{-}Q^{+}& Q^{+}\\Q^{-}& B_2\end{array}\right),\label{optimalxfirst}\\[0.2cm]
\left(\begin{array}{cccc}
y^*_{11}&y^*_{12}&y^*_{13}&y^*_{14}\\
y^*_{21}&y^*_{22}&y^*_{23}&y^*_{24}\\
y^*_{31}&y^*_{32}&y^*_{33}&y^*_{34}\\
y^*_{41}&y^*_{42}&y^*_{43}&y^*_{44}
\end{array}\right)=
\left(\begin{array}{cc}x^{*}_{11}&x^{*}_{12}\\x^{*}_{21}&x^{*}_{22}\end{array}\right) \otimes \left(\begin{array}{cc}x^{*}_{11}&x^{*}_{12}\\x^{*}_{21}&x^{*}_{22}\end{array}\right),\label{eq:optimalxsecond}
\end{gather}
where the operator $\otimes$ denotes tensor product. The easiest way to argue about the validity of these values is to use the strict concavity of the functions $g_1(\mathbf{X}), g_2(\mathbf{Y})$ and just check that the required equalities for the critical points are satisfied.

\begin{lemma}\label{lem:asympfirst}
For $(\alpha,\beta)=(p^+,p^-)$, we have that
\begin{equation*}
\lim_{n\rightarrow\infty}\frac{\E_{\G}[Z^{\alpha,\beta}_G]}{\frac{1}{n}e^{n\Phi_1(\alpha,\beta,\mathbf{X}^{*})}}=\frac{1}{2\pi} \big(\alpha\beta(1-\alpha)(1-\beta)\big)^{(\Delta-1)/2}\left(\frac{Q^{+}Q^{-}E_2}{E^3_1}\right)^{-\Delta/2}.
\end{equation*}
\end{lemma}
\begin{lemma}\label{lem:asympsec}
For $(\alpha,\beta)=(p^+,p^-)$, we have that
\begin{multline*}
\lim_{n\rightarrow\infty}\frac{\E_{\G}[(Z^{\alpha,\beta}_G)^2]}{\frac{1}{n^2}e^{n\Phi_2(\gamma^*,\delta^*,\mathbf{Y}^{*})}}=\frac{1}{4\pi^2}\big(\alpha\beta(1-\alpha)(1-\beta)\big)^{2(\Delta-1)}\left(\frac{(Q^{+}Q^{-})^4E^3_2E_3}{E^{13}_1}\right)^{-\Delta/2}\\
\times\left(\frac{(Q^+ Q^-)^2 E_2 E_3}{E_1^{7}}\right)^{1/2}\left(1-(\Delta-1)^2\omega^2\right)^{-1/2},
\end{multline*}
where $\gamma^*=\alpha^2,\delta^*=\beta^2$, $\mathbf{Y}^{*}$ is given by \eqref{eq:optimalxsecond} and $\omega$ by \eqref{eq:definitionofomega}.
\end{lemma}

The proofs of Lemmas~\ref{lem:asympfirst} and~\ref{lem:asympsec} are given in Sections~\ref{sec:asympfirst} and~\ref{sec:asympsec} respectively. We need to state one more lemma before we can give the proof of Lemma~\ref{lem:ratiofirstsecond}.
\begin{lemma}\label{lem:firstsecvalues}
It holds that $2\Phi_1(\alpha,\beta,\mathbf{X^*})=\Phi_2(\alpha^2,\beta^2,\mathbf{Y}^*)$.
\end{lemma}
\begin{proof}
It clearly suffices to prove
\begin{align*}
f_2(\alpha^2,\beta^2)=2f_1(\alpha,\beta)\mbox{ and } g_2(\mathbf{Y}^*)=2g_1(\mathbf{X^*}).
\end{align*}
Both of these can be verified by straightforward calculations. A neater way to do this, using more directly the uncorrelated property that Condition \ref{cond:maxima} represents, is to use entropy arguments. We illustrate this in the case of  $f_1$ and $f_2$, with the same technique extending to $g_1,g_2$ as well. 

Define $A,B$ to be independent Bernoulli variables with success probabilities $\alpha,\beta$ respectively. It is straightforward to see that 
 $f_1(\alpha,\beta)=-\big(H(A)+H(B)\big)$ where $H(X)$ is the entropy of the random variable $X$. Let $A'=(A_1,A_2)$, where $A_1,A_2$ are independent copies of the random variable $A$. Define similarly $B'$. Using the definition of entropy, it is easy to check that $f_2(\alpha^2,\beta^2)=-\big(H(A')+H(B')\big)$. But clearly $H(A')=H(A_1)+H(A_2)=2H(A)$ and similarly for $B,B'$, thus verifying that $f_2(\alpha^2,\beta^2)=2f_1(\alpha,\beta)$.   

Using \eqref{eq:optimalxsecond}, one can easily extend the previous argument to show $g_2(\mathbf{Y}^*)=2g_1(\mathbf{X^*})$. We should emphasize that the proof of Lemma~\ref{lem:firstsecvalues} does not rely on the fact that $Q^+,Q^-$ are a function of the critical probabilities on the infinite tree, i.e., we do not need Lemma~\ref{lem:densitiesone}, but rather the representation given in \eqref{eq:pppm} and the uncorrelated property that Condition \ref{cond:maxima} represents.
\end{proof}

\begin{proof}[Proof of Lemma~\ref{lem:ratiofirstsecond}.]
Let $L:=\lim_{n\rightarrow\infty}\frac{\E_{\G}\left[(Z^{\alpha,\beta}_G)^2\right]}{\left(\E_{\G}[Z^{\alpha,\beta}_G]\right)^2}$. Employ Lemmas~\ref{lem:asympfirst},~\ref{lem:asympsec}, and~\ref{lem:firstsecvalues} to obtain
\begin{equation*}
L=\left(1-(\Delta-1)^2\omega^2\right)^{-1/2}\big(\alpha\beta(1-\alpha)(1-\beta)\big)^{\Delta-1}\left(\frac{(Q^+ Q^-)^2 E_2 E_3}{E_1^{7}}\right)^{1/2}R^{-\Delta/2},
\end{equation*}
where
\[R=\left(\frac{(Q^{+}Q^{-})^4E^3_2E_3}{E^{13}_1}\right)\cdot \left(\frac{Q^{+}Q^{-}E_2}{E^3_1}\right)^{-2}=\frac{(Q^{+}Q^{-})^2E_2E_3}{E^7_1}.\]
Thus
\begin{align*}
L&=\big(1-(\Delta-1)^2\omega^2\big)^{-1/2}\big(\alpha\beta(1-\alpha)(1-\beta))^{\Delta-1}\left(\frac{(Q^{+}Q^{-})^2E_2E_3}{E^7_1}\right)^{-(\Delta-1)/2}\\
&=\big(1-(\Delta-1)^2\omega^2\big)^{-1/2}\left(\frac{1}{\alpha^2\beta^2(1-\alpha)^2(1-\beta)^2}\cdot\frac{(Q^{+}Q^{-})^2E_2E_3}{E^7_1}\right)^{-(\Delta-1)/2}.
\end{align*}
By \eqref{eq:pppm}, we have $\alpha=\frac{Q^+(1 +B_1 Q^-)}{E_1}$ and $\beta=\frac{Q^-(1 +B_1 Q^+)}{E_1}$, so that
\begin{align*}
L&=\big(1-(\Delta-1)^2\omega^4\big)^{-1/2}\left(\frac{E_1 E_2 E_3}{(B_2 + Q^-)^2(B_2+Q^+)^2(1 +B_1 Q^-)^2(1 +B_1 Q^+)^2}\right)^{-\frac{\Delta-1}{2}}.
\end{align*}
The lemma follows after observing the easy to prove identity
\[\frac{E_1 E_2 E_3}{(B_2 + Q^-)^2(B_2+Q^+)^2(1 +B_1 Q^-)^2(1 +B_1 Q^+)^2}=1-\omega^2.\]

\end{proof}

\subsubsection{The asymptotics of the first moment}\label{sec:asympfirst}
In this section we prove Lemma~\ref{lem:asympfirst}, i.e. compute the asymptotics of $\E_{\G}[Z^{\alpha,\beta}_G]$ for $(\alpha,\beta)=(p^+,p^-)$, see the relevant Sections~\ref{sec:firstmomentform} and~\ref{sec:generallog}.

\begin{proof}[Proof of Lemma~\ref{lem:asympfirst}.]
We will use the full-dimensional representation (see Section~\ref{sec:maxentropy}) of $g_1(\mathbf{X})$: while we are still going to refer to the values $x^{*}_{ij}$, in what follows the only free variable is going to be $x_{11}$. By the approximation
\[\binom{bn}{an}=(1+o(1))\frac{1}{\sqrt{2\pi n}}\sqrt{\frac{b}{a(b-a)}}e^{n\big(b\ln b-a\ln a-(a-b)\ln(a-b)\big)},\]
we obtain
\begin{align*}
\frac{\E_{\G}[Z^{\alpha,\beta}_G]}{\frac{1}{n}e^{n\Phi_1(\alpha,\beta,\mathbf{X}^{*})}}&=(1+o(1))\frac{1}{2\pi}\cdot \Big(\alpha\beta(1-\alpha)(1-\beta)\Big)^{\frac{\Delta-1}{2}}\cdot\\
&\qquad\left(\sum_{x_{11}}\Big(\prod_{i,j}x_{ij}\Big)^{-1/2}\frac{1}{\sqrt{2\pi n}}\cdot e^{n\big(g_1(x_{11})-g_1(x_{11}^*)\big)}\right)^{\Delta}.
\end{align*}
Note that in the above sum terms with $|x_{11}-x_{11}^*|\leq O(1)$ dominate exponentially over the rest and in a small enough ball around $x^*_{11}$ the function decays quadratically, as established by Lemma~\ref{lem:gmax}. Thus, standard integration arguments (for precise details see \cite[Chapter 9]{JLR}) 
yield
\begin{align*}
\lim_{n\rightarrow\infty}\frac{\E_{\G}[Z^{\alpha,\beta}_G]}{\frac{1}{n}e^{n\Phi_1(\alpha,\beta,\mathbf{X}^{*})}}&=\frac{1}{2\pi} \Big(\alpha\beta(1-\alpha)(1-\beta)\Big)^{(\Delta-1)/2}\Big(\prod_{i,j} x^{*}_{ij}\Big)^{-\Delta/2}\\
&\qquad\left(\frac{1}{\sqrt{2\pi}}\int^\infty_{-\infty}\exp\left(\frac{\partial^2 g_1}{\partial x_{11}^2}(x_{11}^*)\cdot\frac{(x_{11}-x_{11}^*)^2}{2}\right)d x_{11}\right)^\Delta.
\end{align*}
Let $H:=\frac{\partial^2 g_1}{\partial x_{11}^2}(x_{11}^*)$. Since $g_1$ is strictly concave and it achieves its maximum at $x_{11}^*$ we have $H<0$. It follows that
\[\frac{1}{\sqrt{2\pi}}\int^\infty_{-\infty}\exp\left(\frac{\partial^2 g_1}{\partial x_{11}^2}(x_{11}^*)\cdot\frac{(x_{11}-x_{11}^*)^2}{2}\right)d x_{11}=(-H)^{-1/2},\]
and consequently
\begin{align*}
\lim_{n\rightarrow\infty}\frac{\E_{\G}[Z^{\alpha,\beta}_G]}{\frac{1}{n}e^{n\Phi_1(\alpha,\beta,\mathbf{X}^{*})}}&= \frac{1}{2\pi} \big(\alpha\beta(1-\alpha)(1-\beta)\big)^{(\Delta-1)/2}\Big((-H)\prod_{i,j} x^{*}_{ij}\Big)^{-\Delta/2}.
\end{align*}
We have 
\[\frac{\partial^2 g_1}{\partial x_{11}^2}=-\left(\frac{1}{x_{11}}+\frac{1}{\alpha-x_{11}}+\frac{1}{\beta-x_{11}}+\frac{1}{1-\alpha-\beta+x_{11}}\right).\]
Substituting the values of $\alpha,\beta$ from \eqref{eq:pppm} and the values of the $x^*_{ij}$'s from \eqref{optimalxfirst}, we obtain
\[-H=\frac{E_1E_2}{B_1 B_2 Q^- Q^+},\ \prod_{i,j} x^{*}_{ij}=\frac{B_1B_2(Q^{+}Q^{-})^2}{E_1^4}, \mbox{ so that }
(-H)\prod_{i,j} x^{*}_{ij}=\frac{Q^{+}Q^{-}E_2}{E^3_1},\]
and the result follows.
\end{proof}

\subsubsection{The Asymptotics of the Second Moment}\label{sec:asympsec}
In this section we prove Lemma~\ref{lem:asympsec}, i.e., compute the asymptotics of $\E_{\G}[(Z^{\alpha,\beta}_G)^2]$ for $(\alpha,\beta)=(p^+,p^-)$. We have
\begin{equation*}
\E_{\G}[(Z^{\alpha,\beta}_G)^2]=\sum_{\gamma,\delta}\E_{\G}[Y^{\gamma,\delta}_G],
\end{equation*}
where $\E_{\G}[Y^{\gamma,\delta}_G]$ was calculated in Section~\ref{sec:secondmomentform} (see also the relevant Section~\ref{sec:generallog}).

\begin{proof}[Proof of Lemma~\ref{lem:asympsec}.]
As in the proof of Lemma~\ref{lem:asympfirst}, in the following we are going to refer to the full-dimensional representations (see Section~\ref{sec:maxentropy}) of $g_2,\Phi_2$, i.e. our free variables are going to be $\{y_{ij}\}$ with $1\leq i,j\leq3$, so that the rest of the $y_{ij}$ are just shorthand for their respective values (in terms of the free variables). We have that
\begin{align*}
\frac{\E_{\G}[(Z^{\alpha,\beta}_G)^2]}{\frac{1}{n^2}e^{n\Phi_2(\gamma^*,\delta^*,\mathbf{Y}^{*})}}&= (1+o(1))\frac{1}{4\pi^2}\\
&\quad\times\sum_{\gamma,\delta}\left(\frac{1}{\sqrt{2\pi n}}\right)^2\Big(\gamma\delta(1-2\alpha+\gamma)(1-2\beta+\delta)(\alpha-\gamma)^2(\beta-\delta)^2\Big)^{\frac{\Delta-1}{2}}\\
&\quad\ \times\left(\sum_\mathbf{Y}\left(\frac{1}{\sqrt{2\pi n}}\right)^9\left(\prod^4_{i=1}\prod^4_{j=1}\frac{1}{\sqrt{y_{ij}}}\right)e^{n(\Phi_2(\gamma,\delta,\mathbf{Y})-\Phi_2(\gamma^*,\delta^*,\mathbf{Y}^*))/\Delta}\right)^{\Delta}.
\end{align*}
The function $\Phi_2(\gamma,\delta,\mathbf{Y})$ has a unique maximum at $(\gamma^*,\delta^*,\mathbf{Y}^*)$. Thus, in the above sum, terms with $\norm{(\gamma,\delta,\mathbf{Y})-(\gamma^*,\delta^*,\mathbf{Y}^*)}\leq O(1)$ dominate exponentially over the rest. Moreover, in the region $\norm{(\gamma,\delta,\mathbf{Y})-(\gamma^*,\delta^*,\mathbf{Y}^*)}\leq O(1)$ the function $\Phi_2$ decays quadratically. 

For convenience, let $L:=\lim_{n\rightarrow \infty}\frac{\E_{\G}[(Z^{\alpha,\beta}_G)^2]}{\frac{1}{n^2}e^{n\Phi_2(\gamma^*,\delta^*,\mathbf{Y}^{*})}}$. As in the proof of Lemma~\ref{lem:asympfirst}, we obtain
\begin{align}
\begin{split}
L&=\left(\frac{1}{\sqrt{2\pi}}\right)^{15}\Big(\gamma^*\delta^*(1-2\alpha+\gamma^*)(1-2\beta+\delta^*)(\alpha-\gamma^*)^2(\beta-\delta^*)^2\Big)^{\frac{\Delta-1}{2}}\\
&\qquad \qquad\left(\prod^4_{i,j=1}y^*_{ij}\right)^{-\frac{\Delta}{2}}\int^{\infty}_{-\infty}\int^{\infty}_{-\infty}\left(\frac{1}{\left(\sqrt{2\pi}\right)^9}\int^{\infty}_{-\infty}e^{\frac{1}{2}(\gamma,\delta,\mathbf{Y})\cdot\mathbf{H}\cdot(\gamma,\delta,\mathbf{Y})^{T}}d\mathbf{Y}\right)^\Delta d\gamma d\delta\notag
\end{split}\notag\\
\begin{split}\label{eq:secmain}
&=\frac{1}{4\pi^2}\Big(\alpha \beta(1-\alpha)(1-\beta)\Big)^{2(\Delta-1)}\left(\prod^4_{i,j=1}y^*_{ij}\right)^{-\frac{\Delta}{2}}\\
&\qquad \qquad \left(\frac{1}{\sqrt{2\pi}}\right)^{2}\int^{\infty}_{-\infty}\int^{\infty}_{-\infty}\left(\frac{1}{\left(\sqrt{2\pi}\right)^9}\int^{\infty}_{-\infty}e^{\frac{1}{2}(\gamma,\delta,\mathbf{Y})\cdot\mathbf{H}\cdot(\gamma,\delta,\mathbf{Y})^{T}}d\mathbf{Y}\right)^\Delta d\gamma d\delta,
\end{split}
\end{align}
where $\mathbf{H}$ denotes the Hessian matrix of $\Phi_2$ evaluated at $(\gamma^*,\delta^*,\mathbf{Y}^*)$ scaled by $1/\Delta$ and the operator~$\cdot$ stands for matrix multiplication. 

We thus focus on computing the integral in \eqref{eq:secmain}. This will be achieved by two successive Gaussian integrations. First, we split the exponent of the quantity inside the integral as follows.
\begin{equation*}
\frac{1}{2}(\gamma,\delta,\mathbf{Y})\cdot\mathbf{H}\cdot(\gamma,\delta,\mathbf{Y})^{T}=\frac{1}{2}(\gamma,\delta)\cdot\mathbf{H}_{\gamma,\delta}\cdot(\gamma,\delta)^{T}-\frac{1}{2}\mathbf{Y}\cdot(-\mathbf{H}_{\mathbf{Y}})\cdot\mathbf{Y}^T+ \sum^3_{i=1}\sum^3_{j=1}T_{ij} y_{ij},
\end{equation*}
where $\mathbf{H}_{\gamma,\delta}$ denotes the principal minor of $\mathbf{H}$ corresponding to $\gamma,\delta$, $\mathbf{H}_{\mathbf{Y}}$ denotes the principal minor of $\mathbf{H}$ corresponding to $\mathbf{Y}$ and
\[T_{ij}=\frac{1}{\Delta}\left(\gamma \cdot \frac{\partial^2 \Phi_2}{\partial \gamma \partial y_{ij}}+\delta \cdot \frac{\partial^2 \Phi_2}{\partial \delta \partial y_{ij}}\right).\]

It follows that 
\begin{multline*}
\left(\int^{\infty}_{-\infty}e^{\frac{1}{2}(\gamma,\delta,\mathbf{Y})\cdot\mathbf{H}\cdot(\gamma,\delta,\mathbf{Y})^{T}}d\mathbf{Y}\right)^\Delta=\\
e^{\frac{\Delta}{2}(\gamma,\delta)\cdot\mathbf{H}_{\gamma,\delta}\cdot(\gamma,\delta)^{T}}\left(\int^{\infty}_{-\infty}e^{-\frac{1}{2}\mathbf{Y}\cdot(-\mathbf{H}_{\mathbf{Y}})\cdot\mathbf{Y}^T+ \sum^3_{i=1}\sum^3_{j=1}T_{ij} y_{ij}}d\mathbf{Y}\right)^\Delta.
\end{multline*}
Let $\mathbf{T}$ denote the row vector with entries $T_{ij}$. Note that $\mathbf{H}_\mathbf{Y}$ is the Hessian of $g_2(\mathbf{Y})$ evaluated at $\mathbf{Y}^*$. Since $g_2(\mathbf{Y})$ is concave and $g_2$ achieves its maximum at $\mathbf{Y}^*$, we have that $\mathbf{H}_\mathbf{Y}$ is negative definite. Hence,
\begin{multline*}
\left(\frac{1}{\left(\sqrt{2\pi}\right)^9}\int^{\infty}_{-\infty}e^{\frac{1}{2}(\gamma,\delta,\mathbf{Y})\cdot\mathbf{H}\cdot(\gamma,\delta,\mathbf{Y})^{T}}d\mathbf{Y}\right)^\Delta =\\
\frac{1}{\left(\mathrm{Det}(-\mathbf{H}_{\mathbf{Y}})\right)^{\Delta/2}}e^{\frac{\Delta}{2}\left(\mathbf{T}\cdot (-\mathbf{H}_\mathbf{Y})^{-1}\cdot\mathbf{T}^T+(\gamma,\delta)\cdot\mathbf{H}_{\gamma,\delta}\cdot(\gamma,\delta)^{T}\right)}.
\end{multline*}
We are thus left with the task of computing the integral 
\begin{equation}
\left(\frac{1}{\sqrt{2\pi}}\right)^2\int^{\infty}_{-\infty}\int^{\infty}_{-\infty}e^{\frac{\Delta}{2}\left(\mathbf{T}\cdot (-\mathbf{H}_\mathbf{Y})^{-1}\cdot\mathbf{T}^T+(\gamma,\delta)\cdot\mathbf{H}_{\gamma,\delta}\cdot(\gamma,\delta)^{T}\right)}.\label{eq:firstgaussian}
\end{equation}
It can easily be seen that the expression in the exponent is a quadratic polynomial in $\gamma,\delta$, so that we may again perform Gaussian integration. We have 
\begin{equation}
\frac{\Delta}{2}\left(\mathbf{T}\cdot (-\mathbf{H}_\mathbf{Y})^{-1}\cdot\mathbf{T}^T+(\gamma,\delta)\cdot\mathbf{H}_{\gamma,\delta}\cdot(\gamma,\delta)^{T}\right)= D\gamma^2+E\gamma\delta+F\delta^2,\label{eq:defDEF}
\end{equation}
for some complicated $D,E,F$, which can be determined explicitly by \eqref{eq:defDEF}, once we perform the substitutions. Hence,
\begin{equation}
\left(\frac{1}{\sqrt{2\pi}}\right)^2\int^{\infty}_{-\infty}\int^{\infty}_{-\infty}e^{\frac{\Delta}{2}\left(\mathbf{T}\cdot (-\mathbf{H}_\mathbf{Y})^{-1}\cdot\mathbf{T}^T+(\gamma,\delta)\cdot\mathbf{H}_{\gamma,\delta}\cdot(\gamma,\delta)^{T}\right)}d\gamma d\delta=\frac{1}{(4DF-E^2)^{1/2}},\label{eq:secondgaussian}
\end{equation}
provided that $4DF-E^2>0$, so that the integration is meaningful. We will check this condition later.

Combining equations \eqref{eq:secmain}, \eqref{eq:firstgaussian}, \eqref{eq:secondgaussian}, we obtain
\begin{align}
L&=\frac{1}{4\pi^2}\Big(\alpha \beta(1-\alpha)(1-\beta)\Big)^{2(\Delta-1)}\bigg(\mathrm{Det}(-\mathbf{H}_{\mathbf{Y}})\prod^4_{i,j=1}y^*_{ij}\bigg)^{-\Delta/2}(4DF-E^2)^{-1/2}\notag\\
&=\frac{1}{4\pi^2}\Big(\alpha \beta(1-\alpha)(1-\beta)\Big)^{2(\Delta-1)}\cdot R,\label{eq:secmomentalmost}
\end{align}
where
\[R:=\bigg(\mathrm{Det}(-\mathbf{H}_{\mathbf{Y}})\prod^4_{i,j=1}y^*_{ij}\bigg)^{-\Delta/2}(4DF-E^2)^{-1/2}.\]
At this point, we use Maple to perform the remaining calculations. The solution in the form of \eqref{eq:optimalxsecond} is particularly handy, since it rationalizes the relevant expressions. We have:
\begin{equation*}
\begin{gathered}
\mathrm{Det}(-\mathbf{H}_{\mathbf{Y}})=\frac{E_1^{19} E_2^3 E_3}{(B_1B_2)^8(Q^+ Q^-)^{12}},\quad \prod^4_{i,j=1}y^*_{ij}=\frac{(B_1B_2)^8(Q^+ Q^-)^{16}}{E_1^{32}},\\
4DF-E^2=\frac{E_1^{7}}{(Q^+Q^-)^2 E_2 E_3}\cdot\left(1-(\Delta-1)^2\omega^2\right).
\end{gathered}
\end{equation*}
Note that by Lemma~\ref{lem:technicalinequality}, we have $4DF-E^2>0$, so that the integration in \eqref{eq:secondgaussian} is indeed meaningful. After plugging the above values into $R$ and substituting back in \eqref{eq:secmomentalmost}, the lemma follows.
\end{proof}

\subsection{Gadget}
\label{sec:gadget}
In this subsection we prove Lemmas \ref{lem:modfirst} and \ref{lem:modsecond}. We will need the following bound.
\begin{lemma}[general form of Lemma 3.2 in~\cite{Sly10}]\label{lllo}
Let $b_1,\dots,b_\ell\geq 0$ and $y_1,\dots,y_\ell$ be integers. Let $a=b_1+\dots+b_\ell$ and $x=y_1+\dots+y_\ell$. Assume
$y_i^2\leq b_i$ for each $i\in\{1,\dots,\ell\}$. Then
$$
\frac{\binom{a+x}{ b_1+y_1, b_2+y_2,\dots,b_\ell+y_\ell}}{\binom{a}{b_1,b_2,\dots,b_\ell}} = \left(1+O\left(\sum_{i=1}^\ell y_i^2/b_i \right)\right)\frac{a^x}{b_1^{y_1}\dots
b_\ell^{y_\ell}}.
$$
\end{lemma}

\begin{proof}
We have
\begin{align*}
\frac{\binom{a+x}{ b_1+y_1, b_2+y_2,\dots,b_\ell+y_\ell}}{\binom{a}{b_1,b_2,\dots,b_\ell}}& =
\frac{(a+x)!}{a!}\prod_{i=1}^\ell\frac{b_i!}{(b_i+y_i)!}
\\ &=
\frac{a^x}{b_1^{y_1}\dots
b_\ell^{y_\ell}}\left(\prod_{j=1}^x(1+j/a)\right)/
\prod_{i=1}^\ell \left(\prod_{j=1}^{y_i} (1+j/b_i)\right) \\
& = \frac{a^x}{b_1^{y_1}\dots
b_\ell^{y_\ell}}\left(1+O\left(\frac{x^2}{a}+\frac{y_1^2}{b_1}+\dots+\frac{y_\ell^2}{b_\ell}\right)\right).
\end{align*}
The lemma now follows from the fact that $x^2/a\leq
y_1^2/b_1+\dots+y_\ell^2/b_\ell^2$ (which follows from the Cauchy-Schwarz
inequality).
\end{proof}

\subsubsection{First Moment}\label{sec:modfirstproof}

\begin{proof}[Proof of Lemma~\ref{lem:modfirst}.]
By \eqref{eq:firstmoment1}, recall that
   \begin{align*}
\E_\G\big[
Z^{\alpha,\beta}_{G}\big] =\lambda^{(\alpha+\beta) n} \binom{n}{\alpha n}\binom{n}{\beta n} \left(
\sum_{x\leq\min\{\alpha,\beta\}}
\kappa^{\alpha,\beta,x}_{G} \right)^{\Delta},
\end{align*}
where
\begin{align*}
\kappa^{\alpha,\beta,x}_{G} = & \frac{\binom{\alpha
n}{x n}\binom{(1-\alpha)n}{(\beta-x)n}}{\binom{n}{\beta n}} B_1^{xn}
B_2^{(1-\alpha-\beta+x)n}.
\end{align*}
Similarly, we have 
\begin{multline*}
\E_{\overline{\G}}\big[ Z^{\alpha,\beta}_{\overline{G}}(\eta)\big]=\lambda^{(\alpha+\beta)n+\eta^-_1+\eta^-_2} \binom{n}{\alpha n}\binom{n}{\beta n}\\
\times \left(
\sum_{x\leq\min\{\alpha+\eta_1^-/n,\beta+\eta_2^-/n\}}
\kappa^{\alpha,\beta,x}_{\overline{G}}(\eta) \right)^{\Delta-1} \left(
\sum_{x\leq\min\{\alpha,\beta\}}
\kappa^{\alpha,\beta,x}_{G} \right),
\end{multline*}
where
\begin{align*}
\kappa^{\alpha,\beta,x}_{\overline{G}}(\eta) = & \frac{\binom{\alpha
n+\eta_1^-}{x n}\binom{(1-\alpha)n+\eta_1^+}{(\beta-x)n+\eta_2^-}}{\binom{n+m'}{\beta
n+\eta^-_2}} B_1^{xn}
B_2^{(1-\alpha-\beta+x)n+\eta_1^+-\eta_2^-}.
\end{align*}
From Lemma~\ref{lllo} and the fact that
$\eta_1^+,\eta_1^-,\eta_2^+,\eta_2^-\leq n^{1/4}$, we have the
following estimate on the
ratio of $\kappa^{\alpha,\beta,x}_{\overline{G}}(\eta)$ and
$\kappa^{\alpha,\beta,x}_{G}$:
\begin{equation}\label{wwzzz}
\frac{\kappa^{\alpha,\beta,x}_{\overline{G}}(\eta)}{\kappa^{\alpha,\beta,x}_{G}}
= (1+o(1))
\frac{\alpha^{\eta_1^-}(1-\alpha)^{\eta_1^+}\beta^{\eta_2^-}(1-\beta)^{\eta_2^+}}{(\alpha-x)^{\eta_1^-}
(1-\alpha-\beta+x)^{\eta_1^+-\eta_2^-}
(\beta-x)^{\eta_2^-}}B_2^{\eta_1^+ -
\eta_2^-}.
\end{equation}

It is easy to see that for any $\alpha,\beta$, the logarithm of the function $\kappa^{\alpha,\beta,x}_{G}$ scaled by $n$ is exactly the full-dimensional representation of the function $g_1(\mathbf{X})$ (see Section~\ref{sec:generallog}). Since $g_1(\mathbf{X})$ is strictly concave with respect to $\mathbf{X}$, it has a unique maximum. With these observations, it can easily be calculated that the value of $x=x_{11}$ for
which $g_1(\mathbf{X})$ achieves its maximum is given by the unique positive solution of \[B_1B_2(\alpha-x^*)(\beta-x^*)=x^*(1-\alpha-\beta+x^*).\]
By Lemma~\ref{lem:gmax}, the logarithm
of $\kappa^{\alpha,\beta,x}_{G}$ decays quadratically
in the distance from $x^*$. Let
\[{\cal A} = \{x : |x-x^*|\leq n^{-1/4}\}.\]
The relative contribution of terms outside ${\cal A}$ to $\E_\G\big[Z^{\alpha,\beta}_{G}\big]$ is $\exp(-\Omega(n^{1/2}))$ and hence
can be omitted. Similarly, using~\eqref{wwzzz}, the relative
contribution of terms outside ${\cal A}$ to $\E_{\overline{\G}}\big[ Z^{\alpha,\beta}_{\overline{G}}(\eta)\big]$ is $\exp(-\Omega(n^{1/2}))$
and hence can also be omitted.

For $x\in{\cal A}$ we have
\begin{align*}
\frac{\kappa^{\alpha,\beta,x}_{\overline{G}}(\eta)}{\kappa^{\alpha,\beta,x}_{G}}
= & (1+o(1))
\frac{\alpha^{\eta_1^-}(1-\alpha)^{\eta_1^+}\beta^{\eta_2^-}(1-\beta)^{\eta_2^+}}{(\alpha-x^*)^{\eta_1^-}
(1-\alpha-\beta+x^*)^{\eta_1^+-\eta_2^-}
(\beta-x^*)^{\eta_2^-}}B_2^{\eta_1^+ - \eta_2^-}.
\end{align*}
Thus, 
\begin{multline*}
\frac{\E_{\overline{\G}}\big[ Z^{\alpha,\beta}_{\overline{G}}(\eta)\big]}{\E_\G\big[
Z^{\alpha,\beta}_{G}\big]}=\\
(1+o(1))\left(
\frac{\alpha^{\eta_1^-}(1-\alpha)^{\eta_1^+}\beta^{\eta_2^-}(1-\beta)^{\eta_2^+}}{(\alpha-x^*)^{\eta_1^-}
(1-\alpha-\beta+x^*)^{\eta_1^+-\eta_2^-}
(\beta-x^*)^{\eta_2^-}}B_2^{\eta_1^+ - \eta_2^-}
\right)^{\Delta-1}\lambda^{\eta_1^-+\eta_2^-},
\end{multline*}
which proves the first part of the lemma.

When $(\alpha,\beta)=(p^+,p^-)$, we express $\alpha,\beta$ and $x^*=x^*_{11}$ in terms of $Q^+,Q^-$
using equations~\eqref{eq:pppm}
and~\eqref{optimalxfirst}. We obtain
\begin{align*}
\lefteqn{
\frac{\alpha^{\eta_1^-}(1-\alpha)^{\eta_1^+}\beta^{\eta_2^-}(1-\beta)^{\eta_2^+}}{(\alpha-x^*)^{\eta_1^-}
(1-\alpha-\beta+x^*)^{\eta_1^+-\eta_2^-}
(\beta-x^*)^{\eta_2^-}}B_2^{\eta_1^+ - \eta_2^+}
} \hspace{0.6in}
 \\
&=\left(\frac{B_2(1-\alpha)(1-\beta)}{1-\alpha-\beta+x^*}\right)^{m'} 
\left(\frac{\alpha(1-\alpha-\beta+x^*)}{B_2
(\alpha-x^*)(1-\alpha)}\right)^{\eta_1^-}
\left(\frac{\beta(1-\alpha-\beta+x^*)}{B_2
(\beta-x^*)(1-\beta)}\right)^{\eta_2^-}\\
& = \left(\frac{(B_2+Q^+)(B_2+Q^-)}{B_2+Q^++Q^-+B_1Q^+Q^-}\right)^{m'}
\left(\frac{1+B_1Q^-}{B_2+Q^-}\right)^{\eta_1^-}
\left(\frac{1+B_1Q^+}{B_2+Q^+}\right)^{\eta_2^-},
\end{align*}
and hence
\begin{equation}\label{eq:qwertya}
\frac{\E_{\overline{\G}}\big[ Z^{\alpha,\beta}_{\overline{G}}(\eta)\big]}{\E_\G\big[
Z^{\alpha,\beta}_{G}\big]}= (1+o(1))
C^* \left(\lambda\left(\frac{1+B_1Q^-}{B_2+Q^-}\right)^{\Delta-1}\right)^{\eta_1^-}
\left(\lambda\left(\frac{1+B_1Q^+}{B_2+Q^+}\right)^{\Delta-1}\right)^{\eta_2^-},
\end{equation}
where $C^*$ is given by  \eqref{eq:cstar}. Using the fact that $Q^+,Q^-$ satisfy \eqref{eq:densitiesone}, we have
\[Q^+=\lambda\left(\frac{1+B_1Q^-}{B_2+Q^-}\right)^{\Delta-1}\quad\mbox{and}\quad Q^-=\lambda\left(
\frac{1+B_1Q^+}{B_2+Q^+}\right)^{\Delta-1}.\]
Hence, we obtain
\begin{eqnarray*}
\frac{\E_{\overline{\G}}\big[ Z^{\alpha,\beta}_{\overline{G}}(\eta)\big]}{\E_\G\big[
Z^{\alpha,\beta}_{G}\big]} 
&= & (1+o(1))C^*
(1+Q^+)^{m'}(1+Q^-)^{m'}
\\
& & \times\left(\frac{Q^+}{1+Q^+}\right)^{\eta_1^-}
\left(\frac{1}{1+Q^+}\right)^{\eta_1^+}
\left(\frac{Q^-}{1+Q^-}\right)^{\eta_2^-}
\left(\frac{1}{1+Q^-}\right)^{\eta_2^+}.
\end{eqnarray*}
The second part of the lemma follows after observing that $q^{\pm}=\frac{Q^\pm}{1+Q^\pm},\ 1-q^\pm=\frac{1}{1+Q^\pm}$.
\end{proof}

\subsubsection{Second Moment}\label{lem:modsecondproof}
\begin{proof}[Proof of Lemma~\ref{lem:modsecond}.]
By \eqref{eq:secondmoment}, we have
\begin{equation}\label{OOUU1}
\begin{aligned}
\E_{\G}\big[\big(Z^{\alpha,\beta}_{G}\big)^2\big]&=\lambda^{2(\alpha+\beta) n}\binom{n}{\alpha n}\binom{n}{\beta
n}\\
&\quad\times\sum_{\gamma,\delta}\binom{\alpha n}{\gamma
n}\binom{(1-\alpha)n}{(\alpha-\gamma)n}\binom{\beta n}{\delta
n}\binom{(1-\beta)n}{(\beta-\gamma)n}\bigg(\sum_{\mathbf{Y}}\kappa^{\gamma,\delta,\mathbf{Y}}_G\bigg)^\Delta
\end{aligned}
\end{equation}
where
\begin{align*}
\kappa^{\gamma,\delta,\mathbf{Y}}_G = &
\frac{\prod_{i=1}^4\binom{L_i
n}{y_{i1}n,y_{i2}n,y_{i3}n,y_{i4}n}\prod^4_{j=1}
\binom{R_jn}{y_{1j}n,y_{2j}n,y_{3j}n,y_{4j}n}}{\binom{n}{y_{11}n, y_{12}n, \dots ,y_{44}n}}\\
& \times B_1^{(2y_{11} + y_{12} + y_{13}+ y_{21} + y_{22} +y_{31}+y_{33})n}
B_2^{(y_{22} + y_{24} + y_{33} +
y_{34}+y_{42}+y_{43}+2y_{44})n},
\end{align*}
\begin{equation*}
\begin{array}{lll}
L_1=\gamma,& L_2=L_3=\alpha-\gamma,& L_4=1-2\alpha+\gamma,\\
R_1=\delta,& R_2=R_3=\beta-\delta,& R_4=1-2\beta+\delta,
\end{array}
\end{equation*}
and $\mathbf{Y}$ stands for the nonnegative variables $y_{ij}$,
$1\leq i,j\leq 4$ such that \[\mbox{$\sum_j$}y_{ij}=L_i\mbox{ and }\mbox{$\sum_i$}y_{ij}=R_j.\]

Similarly, we have that
\begin{equation}\label{OOUU2}
\begin{aligned}
\E_{\overline{\G}}\big[\big(Z^{\alpha,\beta}_{\overline{G}}(\eta)\big)^2\big]&=\lambda^{2(\alpha+\beta) n + 2(\eta_1^- + \eta_2^-)}\binom{n}{\alpha
n}\binom{n}{\beta n}
\\ &\quad \times\sum_{\gamma,\delta}\binom{\alpha n}{\gamma
n}\binom{(1-\alpha)n}{(\alpha-\gamma)n}\binom{\beta n}{\delta
n}\binom{(1-\beta)n}{(\beta-\gamma)n}\\
&\qquad \qquad\times\bigg(\sum_{\hat{\mathbf{Y}}}\kappa^{\gamma,\delta,\hat{\mathbf{Y}}}_{\overline{G}}(\eta)\bigg)^{\Delta-1}\bigg(\sum_{\mathbf{Y}}\kappa^{\gamma,\delta,\mathbf{Y}}_G\bigg),
\end{aligned}
\end{equation}
where
\begin{align*}
\kappa^{\gamma,\delta,\hat{\mathbf{Y}}}_{\overline{G}}(\eta)=
& \frac{\prod_{i=1}^4\binom{\hat{L}_i
n}{\hat{y}_{i1}n,\hat{y}_{i2}n,\hat{y}_{i3}n,\hat{y}_{i4}n}
\prod^4_{j=1} \binom{\hat{R}_jn}{\hat{y}_{1j}n,\hat{y}_{2j}n,\hat{y}_{3j}n,\hat{y}_{4j}n}}{\binom{n+m'}{\hat{y}_{11}n, \hat{y}_{12}n, \dots ,\hat{y}_{44}n}}\\
& \times \B_1^{(2\hat{y}_{11} + \hat{y}_{12} + \hat{y}_{13}+ \hat{y}_{21} +
\hat{y}_{22} +\hat{y}_{31}+\hat{y}_{33})n}
\B_2^{(\hat{y}_{22} + \hat{y}_{24} + \hat{y}_{33} +
\hat{y}_{34}+\hat{y}_{42}+\hat{y}_{43}+2\hat{y}_{44})n},
\end{align*}
\begin{equation*}
\begin{array}{lll}
\hat{L}_1=\gamma+\eta_1^-,& \hat{L}_2=\hat{L}_3=\alpha-\gamma,&
\hat{L}_4=1-2\alpha+\gamma+\eta_1^+,\\
\hat{R}_1=\delta+\eta_2^-,& \hat{R}_2=\hat{R}_3=\beta-\delta,&
\hat{R}_4=1-2\beta+\delta+\eta_2^+,
\end{array}
\end{equation*}
and $\hat{\mathbf{Y}}$ stands for the nonnegative variables
$\hat{y}_{ij}$, $1\leq i,j\leq 4$ such that \[\mbox{$\sum_j$}\hat{y}_{ij}=\hat{L}_i\mbox{ and }\mbox{$\sum_j$}\hat{y}_{ij}=\hat{R}_j.\]

For a given $\mathbf{Y}$, consider $\hat{\mathbf{Y}}$ such that
$y_{ij}=\hat{y}_{ij}$ for all $i,j\in\{1,2,3,4\}$, except
$\hat{y}_{14}=y_{14}+\eta_1^-$, $\hat{y}_{41}=y_{41}+\eta_2^-$, and
$\hat{y}_{44}=y_{44}+m'-\eta_1^- - \eta_2^-$. This is always possible, except at the boundary, but as we will see shortly, these cases can be safely ignored. By
Lemma~\ref{lllo}, we have
\begin{equation}\label{eyy}
\frac{\kappa^{\gamma,\delta,\hat{\mathbf{Y}}}_{\overline{G}}(\eta)}{\kappa^{\gamma,\delta,\mathbf{Y}}_G}=
(1+o(1)) \frac{\gamma^{\eta_1^-}(1-2\alpha+\gamma)^{\eta_1^+}\delta^{\eta_2^-}(1-2\beta+\delta)^{\eta_2^+}}{y_{14}^{\eta_1^-}
y_{41}^{\eta_2^-} y_{44}^{m'-\eta_1^--\eta_2^-}}
B_2^{2m'-2\eta_1^--2\eta_2^-}.
\end{equation}
Let
$$
{\cal A} = \{(\gamma,\delta,\mathbf{Y}) :
\|(\gamma,\delta,\mathbf{Y})-(\gamma^*,\delta^*,\mathbf{Y}^*)\|_2\leq
n^{-1/4}\},
$$
where $\gamma^*=\alpha^2$, $\delta^*=\beta^2$ and $\mathbf{Y}^*$ is given by \eqref{eq:optimalxsecond}. By Condition~\ref{cond:maxima} and Lemma~\ref{lem:gmax}, the relative contribution of terms outside ${\cal A}$ to $\E_{\G}\big[\big(Z^{\alpha,\beta}_{G}\big)^2\big]$ is
$\exp(-\Omega(n^{1/2}))$ and hence can be omitted. Similarly,
using~\eqref{eyy}, the relative contribution of terms outside
${\cal A}$ to $\E_{\overline{\G}}\big[\big(Z^{\alpha,\beta}_{\overline{G}}(\eta)\big)^2\big]$ is
$\exp(-\Omega(n^{1/2}))$ and hence
can also be omitted. For $(\gamma,\delta,{\mathbf Y})\in{\cal A}$ we have
\begin{equation*}
\frac{\kappa^{\gamma,\delta,\hat{\mathbf{Y}}}_{\overline{G}}(\eta)}{\kappa^{\gamma,\delta,\mathbf{Y}}_G}=
(1+o(1))
\frac{(\gamma^*)^{\eta_1^-}(1-2\alpha+\gamma^*)^{\eta_1^+}(\delta^*)^{\eta_2^-}(1-2\beta+\delta^*)^{\eta_2^+}}{(y_{14}^*)^{\eta_1^-}
(y_{41}^*)^{\eta_2^-} (y_{44}^*)^{m'-\eta_1^--\eta_2^-}}
B_2^{2m'-2\eta_1^--2\eta_2^-}.
\end{equation*}
Using equations~\eqref{eq:pppm}
and~\eqref{eq:optimalxsecond}, we express $\alpha,\beta$ and $\mathbf{Y}$ in terms of $Q^+,Q^-$. For $(\gamma,\delta,{\mathbf
Y})\in{\cal A}$, we obtain that
\begin{align*}
\frac{\kappa^{\gamma,\delta,\hat{\mathbf{Y}}}_{\overline{G}}(\eta)}{\kappa^{\gamma,\delta,\mathbf{Y}}_G}
& = (1+o(1)) \left(\frac{B_2^2
(1-2\alpha+\gamma^*)(1-2\beta+\delta^*)}{y^*_{44}}\right)^{m'}\\
& \quad \times \left(\frac{\gamma^* y^*_{44}}{B_2^2 y^*_{14}
(1-2\alpha+\gamma^*)}\right)^{\eta_1^-}
\left(\frac{\delta^* y^*_{44}}{B_2^2 y^*_{41}
(1-2\beta+\delta^*)}\right)^{\eta_2^-}\\
& = (1+o(1)) \left(\frac{(B_2+Q^+)(B_2+Q^-)}{B_2+Q^++Q^-+B_1Q^+Q^-}\right)^{2m'}\\
&\quad \times\left(\frac{1+B_1Q^-}{B_2+Q^-}\right)^{2\eta_1^-}
\left(\frac{1+B_1Q^+}{B_2+Q^+}\right)^{2\eta_2^-}.
\end{align*}
Plugging the last equation into~\eqref{OOUU1} and~\eqref{OOUU2} we
obtain the lemma.
\end{proof}

\subsubsection{Proof of Lemma~\ref{lem:rfsgadget}}\label{sec:rfsgadgetproof}
\begin{proof}[Proof of Lemma~\ref{lem:rfsgadget}.]
By equations \eqref{eq:secondmodratio} and \eqref{eq:qwertya}, we have
\[\lim_{n\rightarrow\infty}\frac{\E_{\overline{\G}}\big[\big(Z^{\alpha,\beta}_{\overline{G}}(\eta)\big)^2\big]}{\big(\E_{\overline{\G}}\big[Z^{\alpha,\beta}_{\overline{G}}(\eta)\big]\big)^2}=\lim_{n\rightarrow\infty}\frac{\E_{\G}\big[\big(Z^{\alpha,\beta}_{G}\big)^2\big]}{\big(\E_{\G}\big[Z^{\alpha,\beta}_{G}\big]\big)^2}.\]
The lemma follows, after using Lemma~\ref{lem:ratiofirstsecond}.
\end{proof}

\section{Asymptotically Almost Surely results}\label{sec:smallgraph}

\subsection{Overview}\label{sec:oversmall}
In this section, we apply the small subgraph conditioning method to obtain Lemmas~\ref{lem:smallgraph} and~\ref{lem:smallgraphgadget}, see Section~\ref{sec:proof-approach} for a brief discussion. Closely related results have appeared in \cite{MWW,Sly10} for the hard-core model, but here we need to account for the extra parameters $B_1,B_2$. The proofs are similar, however the technical calculations are a bit more complicated due to the extra parameters.

The following theorem is taken from \cite{MWW}. The notation $[X]_{m}$ refers to the $m$-th order falling factorial of the variable $X$.

\begin{theorem}\label{thm:smallgraphmethod}
Let $\lambda_i>0$ and $\delta_i>-1$ be real numbers for $i=1,2,\hdots$. Let $r(n)\rightarrow 0$ and suppose that for each $n$ there are random variables $X_i=X_i(n)$, $i=1,2,\hdots$ and $Y=Y(n)$, all defined on the same probability space $\G=\G_n$ such that $X_i$ is nonnegative integer valued, $Y$ is nonnegative and $\E[ Y]>0$ (for $n$ sufficiently large). Suppose furthermore that
\begin{enumerate}
\item For each $k\geq 1$, the variables $X_1,\hdots,X_k$ are asymptotically independent Poisson random variables with $\E X_i\rightarrow \lambda_i$,
\item for every finite sequence $m_1,\hdots,m_k$ of nonnegative integers,
\begin{equation}
\frac{\E[Y[X_1]_{m_1}\cdots [X_k]_{m_k}]}{\E[Y]}\rightarrow \prod^k_{i=1}\big(\lambda_i(1+\delta_i)\big)^{m_i},\label{eq:prodsmall}
\end{equation}
\item $\sum_i \lambda_i \delta^2_i<\infty$.
\item $\E[Y^2]/(\E[Y])^2\leq \exp\big(\sum_i \lambda_i \delta^2_i)+o(1)$ as $n\rightarrow \infty$.
\end{enumerate}
Then $Y>r(n)\E[Y]$ asymptotically almost surely.
\end{theorem}

\begin{proof}
See, for example, \cite[Theorem 7.1]{MWW}.
\end{proof}

Recall that $Z^{\alpha,\beta}_G$  is the measure of configurations in $\Sigma^{\alpha,\beta}_G$ for a random $\Delta$-regular bipartite graph $G$. Let $X_i$ be the number of cycles of length $i$ in $G$. Clearly $X_i=0$ iff $i$ is odd. We have the following lemmas.
\begin{lemma}[Lemma 7.3 in \cite{MWW}]\label{lem:cycle}
Condition 1 of Theorem~\ref{thm:smallgraphmethod} holds for even $i$ with
\[\lambda_i=\frac{r(\Delta,i)}{i},\]
where $r(\Delta,i)$ is the number of ways one can properly edge color a cycle of length $i$ with $\Delta$ colors.
\end{lemma}
\begin{proof}
The proof of Lemma~\ref{lem:cycle} is given in \cite[Lemma 7.3]{MWW}.
\end{proof}

\begin{lemma}\label{lem:ratiosmallgraph}
For $(\alpha,\beta)=(p^\pm,p^\mp)$, we have that
\[\frac{\E_\G[Z^{\alpha,\beta}_G X_i]}{\E_\G[Z^{\alpha,\beta}_G]}\rightarrow \lambda_i(1+\delta_i)\quad\mbox{ and }\quad \frac{\E_{\overline{\G}}[Z^{\alpha,\beta}_{\overline{G}}(\eta) X_i]}{\E_{\overline{\G}}[Z^{\alpha,\beta}_{\overline{G}}(\eta)]}\rightarrow \lambda_i(1+\delta_i),\]
where
\[\delta_i=\frac{(1-B_1B_2)^{i/2}(Q^+Q^-)^{i/2}}{\Big((1+B_1 Q^+)(1+B_1 Q^-)(B_2+Q^+)(B_2+Q^-)\Big)^{i/2}}.\]
\end{lemma}
\noindent The proof of Lemma~\ref{lem:ratiosmallgraph} is given in Section~\ref{sec:smallgraphregular}.

\begin{lemma}\label{lem:sumasymptotics}
It holds that \[\exp\left(\sum_{\mbox{even }i\geq 2}\lambda_i\delta^2_i\right)=\big(1-\omega^2\big)^{-(\Delta-1)/2}\big(1-(\Delta-1)^2\omega^2\big)^{-1/2},\] where $\omega$ is given by \eqref{eq:definitionofomega}.
\end{lemma}
\noindent The proof of Lemma~\ref{lem:sumasymptotics} is given in Section~\ref{sec:smallgraphregular}.

\begin{lemma}\label{lem:finitesequence}
For $(\alpha,\beta)=(p^\pm,p^\mp)$ and for every finite sequence $m_1,\hdots,m_k$ of nonnegative integers, it holds that 
\[\frac{\E_\G\big[Z^{\alpha,\beta}_G(\eta) [X_2]_{m_1}\cdots [X_{2k}]_{m_k}\big]}{\E_\G[Z^{\alpha,\beta}_G]}\rightarrow \prod^{k}_{i=1}\big(\lambda_i(1+\delta_i)\big)^{m_i},\]
and
\[\frac{\E_{\overline{\G}}\big[Z^{\alpha,\beta}_{\overline{G}}(\eta)[X_2]_{m_1}\cdots [X_{2k}]_{m_k}\big]}{\E_{\overline{\G}}[Z^{\alpha,\beta}_{\overline{G}}(\eta)]}\rightarrow \prod^{k}_{i=1}\big(\lambda_i(1+\delta_i)\big)^{m_i}.\]
\end{lemma}
\begin{proof}
The proof of Lemma~\ref{lem:finitesequence} is identical to \cite[Proof of Lemma 7.5]{MWW}.\end{proof}

We are now in position to prove Lemmas~\ref{lem:smallgraph} and \ref{lem:smallgraphgadget}.
\begin{proof}[Proofs of Lemmas~\ref{lem:smallgraph} and~\ref{lem:smallgraphgadget}.]
We verify the conditions of Theorem~\ref{thm:smallgraphmethod} in each case: Condition 1 holds (Lemma~\ref{lem:cycle}), Condition 2 holds (Lemma~\ref{lem:finitesequence}), Conditions 3 and 4 hold (Lemma~\ref{lem:sumasymptotics} and Lemmas~\ref{lem:ratiofirstsecond} and~\ref{lem:rfsgadget}). Lemmas~\ref{lem:smallgraph} and \ref{lem:smallgraphgadget} follow by picking $r(n)=1/n$ and $r(n)=1/\sqrt{n}$ respectively.
\end{proof}

\subsection{Random Bipartite $\Delta$-regular Graphs}\label{sec:smallgraphregular}
In this section, we prove Lemma~\ref{lem:ratiosmallgraph}. We give the proof for the variables $Z^{\alpha,\beta}_G$ and omit the proof for the variables $Z^{\alpha,\beta}_{\overline{G}}(\eta)$, which can be carried out the same way as in \cite[Lemma 3.8]{Sly10} and along the lines of Section~\ref{sec:modfirstproof}.

Recall that $r(\Delta,i)$ denotes the number of proper $\Delta$-edge colorings of a cycle of length $i$.
\begin{lemma}\label{lem:properedgecolorings}
$r(\Delta,i)=(\Delta-1)^i+(-1)^i(\Delta-1)$.
\end{lemma}
\begin{proof}
The (simple) proof of Lemma~\ref{lem:properedgecolorings} is given in \cite[Lemma 7.6]{MWW}.
\end{proof}

\begin{proof}[Proof of Lemma~\ref{lem:ratiosmallgraph}.]
For sets $S,T$ such that $S\subset V_1, T\subset V_2,|S|=\alpha n, |T|=\beta n$, denote by $Y_{S,T}$ the measure of the configuration $\sigma$ that $S,T$ induce, i.e. for a vertex $v\in V(G)$, $\sigma(v)=1$ iff $v\in S\cup T$.

Fix a specific pair of $S,T$. By symmetry, 
\begin{equation}
\frac{\E[Z^{\alpha,\beta}_GX_i]}{\E[Z^{\alpha,\beta}_G]}=\frac{\E[Y_{S,T} X_i]}{\E[Y_{S,T}]}. \label{eq:smallratio}
\end{equation}
We now decompose $X_i$ as follows:
\begin{itemize}
\item $\xi$ will denote a proper $\Delta$-edge colored, rooted and oriented $i$-cycle ($r(\Delta,i)$ possibilities), in which the vertices are colored with red($R$), blue($B$) and white($W$) such that no two adjacent vertices in the cycle have both $R$ or $B$ color, i.e. the only assignments which are prohibited for an edge are $(R,R)$, $(B,B)$.  The color of the edges will prescribe which of the $\Delta$ perfect matchings an edge of a (potential) cycle will belong to. $R$ denotes vertices belonging to $S$, $B$ denotes vertices belonging to $T$ and $W$ the remaining vertices.
\item Given $\xi$, $\zeta$ denotes a position that an $i$-cycle can be in (i.e. the exact vertices it traverses, in order) such that prescription of the vertex colors of $\xi$ is satisfied. 
\item $\mathbf{1}_{\xi,\zeta}$ is the indicator function whether a cycle specified by $\xi,\zeta$ is present in the graph $G$.
\end{itemize}
Note that each possible cycle corresponds to exactly $2i$ different configurations $\xi$ (the number of ways to root and orient the cycle). For each of those $\xi$, the respective sets of configurations $\zeta$ are the same. Hence, we may write
\[X_i=\frac{1}{2i}\sum_\xi\sum_\zeta \mathbf{1}_{\xi,\zeta}.\]
Let $P_1:= \Pr[\mathbf{1}_{\xi,\zeta}=1]$. It follows that
\begin{align*}
\E[Y_{S,T}X_i]&=\frac{1}{2i}\sum_\xi\sum_\zeta P_1\cdot\E[Y_{S,T}|\mathbf{1}_{\xi,\zeta}=1].
\end{align*}
In light of \eqref{eq:smallratio}, we need to study the ratio $\E[Y_{S,T}|\mathbf{1}_{\xi,\zeta}=1]/\E[Y_{S,T}]$. At this point, to simplify notation, we may assume that $\xi,\zeta$ are fixed. 

We have shown in Section~\ref{sec:firstmomentform} that
\[ \E[Y_{S,T}]=\lambda^{(\alpha+\beta)n}\left(\sum_{\delta} \frac{\binom{\alpha n}{\delta n}\binom{(1-\alpha)n}{(\beta-\delta)n}}{\binom{n}{\beta n}}B_1^{\delta n}B_2^{(1-\alpha-\beta+\delta)n}\right)^\Delta,\]
where $\delta$ denotes the number of edges between $S,T$ in one matching.

To calculate $\E[Y_{S,T}|\mathbf{1}_{\xi,\zeta}=1]$, we need some notation. For colors $c_1,c_2\in\{R,B,W\}$, we say that an edge is of type $\{c_1,c_2\}$ if its endpoints have colors $c_1,c_2$. Let $j_1,j_2$ denote the number of red and blue vertices in the coloring prescribed by $\xi$. For $k = 1,\hdots,\Delta$, let $a_1(k)$ denote the number of edges of color $k$ of type $\{R,B\}$, $a_2(k)$ denote the number of edges of color $k$ of type $\{W,W\}$, $d_1(k)$ denote the number of edges of color $k$ of type $\{R,W\}$, $d_2(k)$ denote the number of edges of color $k$ of type $\{B,W\}$. Finally, for $j=1,2$ let $a_j=\sum_k a_j(k)$ and $d_j=\sum_k d_j(k)$. By considering the sum of the degrees of $R$ vertices, the sum of the degrees of $B$ vertices and the total number of edges of the cycle, we obtain the following equalities.
\begin{equation}
a_1+d_1=2j_1,\ a_1+d_2=2j_2,\ a_1+a_2+d_1+d_2=i.\label{eq:graphsimple}
\end{equation}
A straightforward modification of the arguments in Section~\ref{sec:firstmomentform} yields
\begin{multline*}
\E[Y_{S,T}|\mathbf{1}_{\xi,\zeta}=1]=\\\lambda^{(\alpha+\beta)n}\prod^\Delta_{k=1}\left(\sum_{\delta_k} \frac{\binom{\alpha n-a_1(k)-d_1(k)}{\delta_k n}\binom{(1-\alpha)n-a_2(k)-d_2(k)}{(\beta-\delta_k)n-a_1(k)-d_2(k)}}{\binom{n-a_1(k)-a_2(k)-d_1(k)-d_2(k)}{\beta n-a_1(k)-d_2(k)}}B_1^{\delta_k n+a_1(k)}B_2^{(1-\alpha-\beta+\delta_k) n+a_1(k)}\right).
\end{multline*}

By Lemma~\ref{lllo}, we have
\begin{align*}
\frac{\binom{\alpha n-a_1(k)-d_1(k)}{\delta_k n}}{\binom{\alpha n}{\delta_k n}}&\sim \left(\frac{\alpha-\delta_k}{\alpha}\right)^{a_1(k)+d_1(k)},\\
\frac{\binom{(1-\alpha)n-a_2(k)-d_2(k)}{(\beta-\delta_k)n-a_1(k)-d_2(k)}}{\binom{(1-\alpha)n}{(\beta-\delta_k)n}}&\sim \left(\frac{1-\alpha-\beta+\delta_k}{1-\alpha}\right)^{a_2(k)+d_2(k)}\left(\frac{\beta-\delta_k}{1-\alpha-\beta+\delta_k}\right)^{a_1(k)+d_2(k)},\\
\frac{\binom{n-a_1(k)-a_2(k)-d_1(k)-d_2(k)}{\beta n-a_1(k)-d_2(k)}}{\binom{n}{\beta n}}&\sim\beta^{a_1(k)+d_2(k)}\left(1-\beta\right)^{a_2(k)+d_1(k)}.
\end{align*}
By the same line of arguments as in the proofs of Lemmas~\ref{lem:modfirst} and~\ref{lem:modsecond}, we obtain that
\begin{align*}
\frac{\E[Y_{S,T}|\mathbf{1}_{\xi,\zeta}=1]}{\E[Y_{S,T}]}&\sim\frac{(\delta^{*})^{a_2-a_1}(\alpha-\delta^{*})^{a_1+d_1}(\beta-\delta^{*})^{a_1+d_2}(B_1B_2)^{a_1}}{\alpha^{a_1+d_1}\beta^{a_1+d_2}(1-\alpha)^{a_2+d_2}(1-\beta)^{a_2+d_1}}.
\end{align*}
Clearly $P_1\sim n^{-i}$ and for given $\xi$, the number of possible $\zeta$ is asymptotic to \[\alpha^{j_1}\beta^{j_2}(1-\alpha)^{i/2-j_1}(1-\beta)^{i/2-j_2}n^i.\]
Thus, for the given $\xi$, we have
\begin{align*}
\frac{\sum_\zeta P_1\E[Y_{S,T}|\mathbf{1}_{\xi,\zeta}=1]}{\E[Y_{S,T}]}\sim R,
\end{align*}
where 
\[R:=\frac{(1-\alpha-\beta+\delta^{*})^{a_2-a_1}(\alpha-\delta^{*})^{a_1+d_1}(\beta-\delta^{*})^{a_1+d_2}(B_1B_2)^{a_1}}{\alpha^{a_1+d_1-j_1}\beta^{a_1+d_2-j_2}(1-\alpha)^{a_2+d_2+j_1-i/2}(1-\beta)^{a_2+d_1+j_2-i/2}}.\]
Utilize \eqref{eq:graphsimple} to express everything in terms of $i,j_1,j_2,a_2$ to get
\[R=\frac{(1-\alpha-\beta+\delta^*)^{i-2  j_1-2j_2}(\alpha-\delta^*)^{2  j_1} (\beta-\delta^*)^{2  j_2}(B_1B_2)^{a_2 - i + 2 j_1 + 2 j_2}}{\alpha^{ j_1}\beta^{ j_2}(1-\alpha)^{\frac{i}{2}- j_1}  (1-\beta)^{\frac{i}{2}- j_2} }.\]
Substituting
\[\alpha= \frac{Q^{+}(1+B_1 Q^{-})}{B_2+Q^{-}+Q^{+}+B_1 Q^{-} Q^{+}},\ \beta=\frac{Q^{-}(1+B_1Q^{+})}{B_2+Q^{-}+Q^{+}+B_1 Q^{-} Q^{+}}\]
and $\delta^*=\frac{B_1 Q^{-}Q^{+}}{B_2+Q^{-}+Q^{+}+B_1 Q^{-} Q^{+}}$, we obtain 
\[R=\frac{B_1^{a_2 - i + 2 j_1 + 2 j_2}B_2^{a_2}(Q^+)^{j_1}(Q^-)^{j_2}}{(1+B_1 Q^{-})^{j_1}(1+B_1 Q^{+})^{j_2}(B_2+Q^-)^{i/2-j_1}(B_2+Q^+)^{i/2-j_2}}.\]
We thus obtain
\begin{equation*}
\frac{\E[Y_{S,T}|\mathbf{1}_{\xi,\zeta}=1]}{\E[Y_{S,T}]}\sim x^i y^{j_1} z^{j_2} w^{a_2},
\end{equation*}
where
\begin{gather*}
x=\frac{1}{B_1 \sqrt{(B_2+Q^-)(B_2+Q^+)}},\quad y=\frac{B_1^2 Q^+(B_2+Q^-)}{1+B_1 Q^{-}},\\
z=\frac{B^2_1 Q^-(B_2+Q^+)}{1+B_1 Q^{+}},\quad w=B_1B_2.
\end{gather*}
Thus, by \eqref{eq:smallratio}, we obtain
\begin{equation}
\frac{\E[YX_i]}{\E[Y]}=\frac{r(\Delta,i)}{2i}\cdot x^i\sum_{j_1,j_2,a_2}a_{j_1,j_2,a_2}y^{j_1} z^{j_2} w^{a_2},\label{eq:poissonsum}
\end{equation}
where $a_{i,j_1,j_2,a_2}$ is the number of possible $\xi$ with $j_1$ red vertices, $j_2$ blue vertices and $a_2$ edges of type $\{W,W\}$. To analyze such sums, an elegant technique (given in \cite{Janson}) is to define an appropriate transition matrix. The powers of the matrix count the (multiplicative) weight of walks in the implicit graph. In our setting, the transition matrix is given by
\begin{equation*}
\mathbf{A}=\left(\begin{array}{ccccccc}
0& & y& & 0& & y\\
z& & 0& & z& & 0\\
0& & 1& & 0& & w\\
1& & 0& & w& & 0
\end{array}\right).
\end{equation*}
Note that a rooted oriented cycle with specification $\xi,\zeta$ corresponds to a unique closed walk in the implicit graph defined by $\mathbf{A}$. We have that $\mathrm{Tr}(\mathbf{A}^i)=\sum^4_{j=1}e^i_j$, where $e_j$ are the eigenvalues of $\mathbf{A}$ ($j=1,\hdots,4$). We thus obtain
\begin{equation*}
\frac{\E[YX_i]}{\E[Y]}=\frac{r(\Delta,i)}{2i}\sum^4_{j=1} (x e_j)^i.
\end{equation*}
A computation gives
\begin{equation*}
\begin{aligned}
e_1=-e_2&=B_1 \left(1-B_1 B_2\right) \sqrt{\frac{Q^+ Q^-}{(1+B_1 Q^+) (1+B_1 Q^-)}},\\
e_3=-e_4&=B_1 \sqrt{(B_2+Q^+) (B_2+Q^-)}.
\end{aligned}
\end{equation*}
It follows that for even $i$, we have
\begin{equation*}
\frac{\E[YX_i]}{\E[Y]}=\frac{r(\Delta,i)}{i}\left(1+\left(\frac{(1-B_1B_2)\sqrt{Q^+Q^-}}{\sqrt{(1+B_1 Q^+)(1+B_1 Q^-)(B_2+Q^+)(B_2+Q^-)}}\right)^i\right),
\end{equation*}
thus proving the lemma.
\end{proof}

\begin{proof}[Proof of Lemma~\ref{lem:sumasymptotics}.]
Using Lemma~\ref{lem:properedgecolorings}, we have
\begin{equation*}
\sum_{\mbox{\small{even} }i\geq 2}\lambda_i\delta^2_i=\sum_{\mbox{\small{even} }i\geq 2}\frac{r(\Delta,i)}{i}\cdot \omega^i=\sum_{\mbox{\small{even} }i\geq 2}\frac{(\Delta-1)^i+(\Delta-1)}{i}\cdot\omega^i.
\end{equation*}
Observe that $\sum_{j\geq 1}\frac{x^{2j}}{2j}=-\frac{1}{2}\ln(1-x^2)$ for all $|x|<1$. Clearly this holds for $\omega$ and by Lemma~\ref{lem:technicalinequality}, it also holds for $(\Delta-1)\omega$. It follows that
\[\sum_{\mbox{\small{even} }i\geq 2}\lambda_i\delta^2_i=-\frac{1}{2}\Big(\ln\big(1-(\Delta-1)^2\omega^2\big)+(\Delta-1)\ln\big(1-\omega^2\big)\Big),\]
thus proving the lemma.
\end{proof}

\section{Remaining Proofs}\label{sec:remainingproofs}

\subsection{Proof of Lemma~\ref{lem:technicalinequality}}\label{sec:technicalinequality}
In this section, we give the proof of Lemma~\ref{lem:technicalinequality}. We prove only the inequality for $\omega$, since the inequality for $\omega^*$ is the condition for non-uniqueness, see the discussion in Section \ref{sec:tree-recursions}. We will in fact establish a slightly stronger inequality, as is captured by the following lemma. To aid the presentation, let us work with $x=Q^+,y=Q^-, d=\Delta-1$.
\begin{lemma}\label{lem:aidforproof}
For any $d\geq 2$, it holds that $B_1 xy+B_1B_2(x+y)+B_2> (d-1)(1-B_1B_2)\sqrt{xy}.$
\end{lemma}
\begin{proof}[Proof of Lemma~\ref{lem:technicalinequality}.]
Let $W:=B_1 xy+B_1B_2(x+y)+B_2$. Observe
\[W=(B_1x+1)(y+B_2)-(1-B_1B_2)y=(B_1y+1)(x+B_2)-(1-B_1B_2)x.\]
Using the two expressions of $W$ and the AM-GM inequality, we obtain
\begin{align}
W^2&=(B_1x+1)(y+B_2)(B_1y+1)(x+B_2)+(1-B_1B_2)^2xy \notag\\
&\qquad -(1-B_1B_2)\big(y(B_1x+1)(B_2+y)+x(B_1y+1)(B_2+x)\big)\notag\\
&\leq\big(\sqrt{(B_1x+1)(B_1y+1)(B_2+x)(B_2+y)}-(1-B_1B_2)\sqrt{xy}\big)^2.\label{eq:mediuminequality}
\end{align}
Trivially 
\[(B_1x+1)(B_1y+1)(B_2+x)(B_2+y)>xy>(1-B_1B_2)^2 xy,\]
so \eqref{eq:mediuminequality} and Lemma~\ref{lem:aidforproof} give
\[\sqrt{(B_1x+1)(y+B_2)(B_1y+1)(x+B_2)}-(1-B_1B_2)\sqrt{xy}>(d-1)(1-B_1B_2).\]
which after trivial manipulations reduces to the desired inequality.
\end{proof}

\begin{proof}[Proof of Lemma~\ref{lem:aidforproof}.]
The $x,y$ satisfy 
\[x=\lambda \left(\frac{B_1 y+1}{y+B_2}\right)^d\mbox{ and } y=\lambda \left(\frac{B_1 x+1}{x+B_2}\right)^d.\]
It follows that
\begin{equation}
x\left(\frac{B_1 x+1}{x+B_2}\right)^d=y\left(\frac{B_1 y+1}{y+B_2}\right)^d \Rightarrow x\Big((B_1 x+1)(y+B_2)\Big)^d=y\Big((B_1 y+1)(x+B_2)\Big)^d.\label{eq:startingpoint}
\end{equation}
Let $W:=B_1 xy+B_1B_2x+B_1B_2y+B_2$. Observe that
\begin{align*}
(B_1 x+1)(y+B_2)&=B_1 xy+B_1B_2x+y+B_2=W+(1-B_1B_2)y,\\
(B_1 y+1)(x+B_2)&=B_1 xy+B_1B_2y+x+B_2=W+(1-B_1B_2)x.
\end{align*}
Hence, \eqref{eq:startingpoint} gives
\[x\left(W+(1-B_1B_2)y\right)^d=y\left(W+(1-B_1B_2)x\right)^d.\]
Expanding using Newton's formula gives
\[x\sum^d_{k=0}\binom{d}{k}W^{d-k}\big((1-B_1B_2)y\big)^{k}=y\sum^d_{k=0}\binom{d}{k}W^{d-k}\big((1-B_1B_2)x\big)^{k},\]
which is equivalent to
\[W^d(x-y)=xy\sum^d_{k=2}\binom{d}{k}W^{d-k}
(1-B_1B_2)^k\left(x^{k-1}-y^{k-1}\right).\]
Since $x\neq y$, this can be rewritten as
\begin{equation}
W^d=xy\sum^d_{k=2}\binom{d}{k}W^{d-k}
(1-B_1B_2)^k\left(\frac{x^{k-1}-y^{k-1}}{x-y}\right).\label{eq:startingpointa}
\end{equation}
\begin{claim}\label{cl:simplebound}
For $k\geq 2$ and $x,y>0$ with $x\neq y$, it holds that
$\displaystyle\frac{x^{k-1}-y^{k-1}}{x-y}> (k-1)(xy)^{(k-2)/2}$.
\end{claim}
The simple proof of Claim~\ref{cl:simplebound} is given at the end. Using Claim~\ref{cl:simplebound}, \eqref{eq:startingpointa} gives
\begin{equation*}
W^d> \sum^d_{k=2}\binom{d}{k}(k-1)W^{d-k}
(1-B_1B_2)^k (xy)^{k/2},
\end{equation*}
or equivalently
\begin{equation}
\underbrace{W^d+ \sum^d_{k=2}\binom{d}{k}W^{d-k}(1-B_1B_2)^k (xy)^{k/2}}_{C}> \underbrace{\sum^d_{k=2}\binom{d}{k}kW^{d-k}
(1-B_1B_2)^k (xy)^{k/2}}_{D}.\label{eq:endingpoint}
\end{equation}
Using again Newton's formula and the identity $\binom{d}{k}=\frac{d}{k}\binom{d-1}{k-1}$, we have
\begin{align*}
C&=\big(W+(1-B_1B_2)\sqrt{xy}\big)^d-dW^{d-1}(1-B_1B_2)\sqrt{xy},\\
D&=d(1-B_1B_2)\sqrt{xy}\sum^d_{k=2}\binom{d-1}{k-1}W^{d-k}
(1-B_1B_2)^{k-1} (\sqrt{xy})^{k-1}\\
&=d(1-B_1B_2)\sqrt{xy}\Big(\big(W+(1-B_1B_2)\sqrt{xy}\big)^{d-1}-W^{d-1}\Big).
\end{align*}
Thus, \eqref{eq:endingpoint} gives
\[W> (d-1)(1-B_1B_2)\sqrt{xy},\]
which is exactly the inequality we wanted. Finally, we give the proof of Claim~\ref{cl:simplebound}.
\begin{proof}[Proof of Claim~\ref{cl:simplebound}.]
Since $k\geq 2$, observe that
\begin{align*}
\frac{x^{k-1}-y^{k-1}}{x-y}&=x^{k-2}+x^{k-3}y+\hdots+y^{k-2}\\
&\geq (k-1)\Big((xy)^{(k-1)(k-2)/2}\Big)^{1/(k-1)}=(k-1)(xy)^{(k-2)/2}.
\end{align*}
The inequality is an application of the AM-GM inequality to $x^{k-2},\hdots,y^{k-2}$. Equality holds iff $x~=~y$.
\end{proof}

This completes the proof of Lemma \ref{lem:aidforproof}.
\end{proof}

\subsection{Proofs for Section~\ref{sec:maxentropy}}\label{sec:maxentropyproof}
\begin{proof}[Proof of Lemma~\ref{lem:gmax}.]
We first give the proof when $M_{ij}>0$ for all $i,j$. Then it is easy to see that the function $g(\mathbf{Z})$ attains its maximum in the interior of the region \eqref{e1}, since
\[\frac{\partial g}{\partial Z_{ij}}=\ln M_{ij}-\ln Z_{ij}-1\rightarrow +\infty \mbox{ as } Z_{ij}\rightarrow 0^+.\]
Since $g(\mathbf{Z})$ is differentiable in the interior of the region \eqref{e1} (by $\alpha_i\neq 0\neq \beta_j$), we may formulate the maximization of $g(\mathbf{Z})$ using Lagrange multipliers. Let
\[\Lambda(\mathbf{Z},\boldsymbol{\lambda}, \boldsymbol{\lambda'})=g(\mathbf{Z})+\lambda_i(\sum_j Z_{ij}-\alpha_i)+\lambda_j'(\sum_i Z_{ij}-\beta_j).\]
The optimal $\mathbf{Z}$ must satisfy
\[\frac{\partial \Lambda}{\partial Z_{ij}}=0 \Longleftrightarrow \ln M_{ij}-\ln Z_{ij}-1+\lambda_i+\lambda_j'\Longleftrightarrow Z_{ij}=\exp\left(\lambda_i+\lambda'_j-1\right)M_{ij}.\]
Setting $R_i=\exp(\lambda_i), C_j=\exp(\lambda'_j-1)$ one obtains that the stationary points of $g(\mathbf{Z})$ have the form stated in Lemma~\ref{lem:gmax}. The equalities \eqref{e3} come from the conditions $\frac{\partial \Lambda}{\partial \lambda_i}=\frac{\partial \Lambda}{\partial \lambda_j'}=0$.

We next treat the case where some entries of $M_{ij}$ are zero.  Recall the convention (Section~\ref{sec:generallog}) that $\ln 0\equiv -\infty$ and $0\ln 0\equiv 0$. To see why these conventions are relevant, note that whenever $M_{ij}=0$, $g(\mathbf{Z})$ can have a maximum at $\mathbf{Z}^*$ only if $Z^*_{ij}=0$. Going back to the interpretations of $g(\mathbf{Z})$ in the first and second moments, the reader will immediately notice that these are in complete accordance with the combinatorial meaning of the variables $Z_{ij}$. Let $P=\{(i,j):\ M_{ij}=0\}$. We may restrict our attention for the maximum of $g(\mathbf{Z})$ in the region
\begin{equation}\label{rege1}
\begin{array}{ll}
\sum_{j} Z_{ij} = \alpha_i(u),& \mbox{for } i\in 1,\hdots, m,\\
\sum_{i} Z_{ij} = \beta_j(v),& \mbox{for } j\in 1,\hdots,n,\\
Z_{ij}\geq 0,& \mbox{for } (i,j)\notin P,\\
Z_{ij}=0& \mbox{for }(i,j)\in P.
\end{array}
\end{equation}
The proof given above (under the assumptions $M_{ij}>0$ for all $i,j$) now goes through for the region \eqref{rege1}, as it can readily be checked. Since $M_{ij}=0$ for $(i,j)\in P$, the form of $\mathbf{Z}^*$ stated in Lemma~\ref{lem:gmax} is consistent with the region \eqref{rege1}.

Finally, the desired quadratic decay is an immediate consequence of the stationarity of $\mathbf{Z}^*$, the strict concavity of $g$ and Taylor's expansion.
\end{proof}

\begin{proof}[Proof of Lemma~\ref{lem:maxentropy}.]
The lemma is essentially an application of the envelope theorem, but we give the short proof for the sake of completeness. From~\eqref{e3} we have
\begin{eqnarray}
\left(\frac{\partial}{\partial u} R_i\right) \sum_{j=1}^m M_{ij} C_j + R_i \sum_{j=1}^m M_{ij} \frac{\partial}{\partial u} C_j & = & \frac{\partial}{\partial u} \alpha_i,\label{e4A} \\
\left(\frac{\partial}{\partial u} C_j\right) \sum_{i=1}^n M_{ij} R_i + C_j \sum_{i=1}^n M_{ij} \frac{\partial}{\partial u} R_i & = & 0 \label{e4B}.
\end{eqnarray}

Then
\begin{multline*}
\frac{\partial g^{*}}{\partial u} =  -\sum_{i=1}^n \sum_{j=1}^m M_{ij} \left(
 \left(\frac{\partial}{\partial u} R_i\right)(C_j\ln (R_i) + C_j\ln (C_j)) \right.\\ 
\left.+\left(\frac{\partial}{\partial u} C_j\right)(R_i\ln (R_i) + R_i\ln (C_j))+ \left(\frac{\partial}{\partial u} R_i\right) C_j + R_i \left(\frac{\partial}{\partial u} C_j\right) \right),
\end{multline*}
which, using~\eqref{e4A} and~\eqref{e4B}, simplifies to
\begin{align*}
\frac{\partial g^{*}}{\partial u} =  - & \sum_{i=1}^n \ln (R_i)\left(  \left(\frac{\partial}{\partial u} R_i\right)\sum_{j=1}^m M_{ij} C_j + R_i\left(\sum_{j=1}^m M_{ij}\frac{\partial}{\partial u} C_j \right)\right)\\
& - \sum_{j=1}^m \ln (C_j)\left(  \left(\frac{\partial}{\partial u} C_j\right)\sum_{i=1}^n M_{ij} R_i + C_j\left(\sum_{i=1}^n M_{ij}\frac{\partial}{\partial u} R_i \right)\right)\\
& - \sum_{j=1}^m \left(\left(\frac{\partial}{\partial u} C_j\right) \sum_{i=1}^n M_{ij} R_i + C_j \sum_{i=1}^n M_{ij} \frac{\partial}{\partial u} R_i\right)\\
=-& \sum_{i=1}^n \ln(R_i) \frac{\partial}{\partial u}\alpha_i.
\end{align*}
\end{proof}

\subsection{Proofs for Section~\ref{sec:logfirstmax}}\label{sec:logfirstmaxproof}

\begin{proof}[Proof of Lemma~\ref{lem:firboundary}.]
To avoid unnecessary complications, we give the proof for $B_1 B_2>0$. The case with $B_1=0$ (essentially hard-core model) is covered in \cite[Claim 2.2]{DFJ}. By \eqref{eaeaea2} we have
\begin{equation}\label{eq:tfvbg3}
\frac{\partial \phi_1 }{\partial\alpha}=\ln\left(\lambda \frac{1-\alpha}{\alpha}\left(\frac{\alpha}{1-\alpha}\frac{R_2}{R_1}\right)^\Delta \right), \quad 
\frac{\partial \phi_1}{\partial\beta}=\ln\left(\lambda \frac{1-\beta}{\beta}\left(\frac{\beta}{1-\beta}\frac{C_2}{C_1}\right)^\Delta \right),
\end{equation}
where the $R_i,C_j$ satisfy \eqref{eb}. By \eqref{eb}, we have
\[\frac{\alpha}{1-\alpha}=\frac{R_1}{R_2}\cdot\frac{B_1C_1+C_2}{C_1+B_2C_2},\quad \frac{\beta}{1-\beta}=\frac{C_1}{C_2}\cdot\frac{B_1R_1+R_2}{R_1+B_2R_2}.\]
Hence, the derivatives $\displaystyle\frac{\partial \phi_1}{\partial \alpha},\displaystyle \frac{\partial \phi_1}{\partial \beta}$ in \eqref{eq:tfvbg3}  become 
\begin{equation}\label{ea1copy}
\frac{\partial \phi_1 }{\partial\alpha}=\ln\left(\lambda \frac{1-\alpha}{\alpha}\left(\frac{B_1C_1+C_2}{C_1+B_2C_2}\right)^\Delta \right),\quad \frac{\partial \phi_1}{\partial\beta}=\ln\left(\lambda \frac{1-\beta}{\beta}\left(\frac{B_1R_1+R_2}{R_1+B_2R_2}\right)^\Delta \right).
\end{equation}
We next establish the  following claim.
\begin{claim}\label{cl:helpfulfir}
When $B_1>0, B_2>0$, the function
\[f(x,y):=\frac{B_1x+y}{x+B_2y},\]
defined for $x,y\geq 0$ such that at least one of $x,y$ is non-zero is upper and lower bounded by strictly positive constants depending only on $B_1$, $B_2$.
\end{claim}
\begin{proof}
Observe that for $t>0$ it holds that $f(x,y)=f(tx,ty)$, so that we may assume $x+y=1$. Thus, $B_1x+y$ and $x+B_2 y$ are convex combinations of $\{1,B_1\}$ and $\{1,B_2\}$ respectively. Since we assumed that $B_1,B_2>0$, $f(x,y)$ is bounded by strictly positive constants.
\end{proof}
Observe that Claim~\ref{cl:helpfulfir} is applicable to our case since not all of $R_1,R_2$  or $C_1,C_2$ can be equal to 0 (by equations \eqref{eb}).

 We consider the boundaries with $\alpha(1-\alpha)\rightarrow 0^+$ and prove that the derivatives with respect to $\alpha$ go the right way. The  boundary cases with $\beta(1-\beta)\rightarrow 0^+$ are completely analogous. We consider cases.
\begin{enumerate}
\item $\alpha\rightarrow 0^+$. Then $1-\alpha\rightarrow1$. By Claim~\ref{cl:helpfulfir} and \eqref{ea1copy}, we obtain $\frac{\partial \phi_1}{\partial \alpha}\rightarrow +\infty$.
\item $1-\alpha\rightarrow 0^+$. Then $\alpha\rightarrow 1$. By Claim~\ref{cl:helpfulfir} and \eqref{ea1copy}, we obtain $\frac{\partial \phi_1}{\partial \alpha}\rightarrow -\infty$.
\end{enumerate}

This concludes the proof of Lemma~\ref{lem:firboundary}.
\end{proof}

\begin{proof}[Proof of Lemma~\ref{lem:firhessian}.]
Denote by $\mathbf{H_1}$ the Hessian of $\Phi_1$ evaluated at $(p^+,p^-,X^*)$. We keep only $\alpha,\beta,x_{11}$ as free variables in $\Phi_1$ (the rest of the $x_{ij}$ are substituted by \eqref{eq:xregion}). We check that the principal minors along the diagonal are positive, the first principal minor being the lower right corner of the matrix. This is equivalent to applying Sylvester's criterion in a rearranged version of $-\mathbf{H_1}$. We use Maple to perform the necessary computations and plug into the resulting formulas the values of $p^\pm,\mathbf{X}^*$ as given in \eqref{eq:pppm} and \eqref{optimalxfirst} respectively.

Denote by $P_i$ the determinant of the $i$-th leading principal minor. We obtain that $P_1=\frac{\Delta E_1 E_2}{B_1 B_2 Q^- Q^+}$, 
\begin{align*}
P_2=\frac{\Delta E_1^2 (B_2+Q^-) (1+B_1 Q^-)\big(1+(\Delta-1)\omega\big)}{B_1 B_2 (Q^-)^2 Q^+},\quad 
P_3=\frac{\Delta E^4_1\big(1-(\Delta-1)^2\omega\big)}{B_1B_2\big(Q^+Q^-\big)^2},
\end{align*}
where $E_1,E_2$ and $\omega$ are given by \eqref{eq:helpfulAs} and \eqref{eq:definitionofomega}. Trivially $P_1, P_2>0$. The fact that $P_3>0$  is just Lemma~\ref{lem:technicalinequality}.

To prove that $(p^*,p^*,X^o)$ is a saddle point, observe that the respective expressions for $P_1,P_2,P_3$ can be obtained by substituting $Q^-\leftarrow Q^*, Q^+\leftarrow Q^*$, where $Q^*=\frac{q^*}{1-q^*}$. We thus have that $P_3<0$ by Lemma~\ref{lem:technicalinequality}.
\end{proof}

\subsection{Proofs for Section~\ref{sec:logsecmax}}\label{sec:logsecmaxproof}
\begin{proof}[Proof of Lemma~\ref{lem:secboundary}.]
To avoid unnecessary complications, we give the proof for $B_1 B_2>0$. The case with $B_1=0$ (essentially hard-core model) is covered in \cite[Lemma 12]{Galanis}. 

By \eqref{eq:dergamma}, we have
\begin{equation*}
\frac{\partial \phi_2}{\partial \gamma}  =\ln\left(\frac{(\alpha-\gamma)^2}{\gamma(1-2\alpha+\gamma)}\left(\frac{\gamma(1-2\alpha+\gamma)}{(\alpha-\gamma)^2}\frac{R_2^2}{R_1 R_4}\right)^\Delta \right),
\end{equation*}
where the $R_i,C_j$ satisfy \eqref{eq:frty}. By \eqref{eq:frty}, we have 
\[\frac{\gamma(1-2\alpha+\gamma)}{(\alpha-\gamma)^2}=\frac{R_1 R_4}{R^2_2}\cdot \frac{(B^2_1 C_1+2B_1 C_2+C_4)(C_1+2B_2 C_2+B^2_2 C_4)}{(B_1 C_1+(B_1B_2+1)C_2+B_2C_4)^2},\]
Hence, the derivative $\displaystyle\frac{\partial \phi_2}{\partial \gamma}$ becomes 
\begin{equation}\label{eq:dergammacopy}
\frac{\partial \phi_2}{\partial \gamma}  =
\ln\left(\frac{(\alpha-\gamma)^2}{\gamma(1-2\alpha+\gamma)}\left(  \frac{(B^2_1 C_1+2B_1 C_2+C_4)(C_1+2B_2 C_2+B^2_2 C_4)}{(B_1 C_1+(B_1B_2+1)C_2+B_2C_4)^2}  \right)^\Delta\right),
\end{equation}
Analogously, starting from \eqref{eq:derdelta}, we can rewrite $\displaystyle \frac{\partial \phi_2}{\partial \delta}$  as
\begin{equation}\label{eq:derdeltacopy}
\frac{\partial \phi_2}{\partial \delta}  =
\ln\left(\frac{(\beta-\delta)^2}{\delta(1-2\beta+\delta)}\left(  \frac{(B^2_1 R_1+2B_1 R_2+R_4)(R_1+2B_2 R_2+B^2_2 R_4)}{(B_1 R_1+(B_1B_2+1)R_2+B_2R_4)^2}  \right)^\Delta\right).
\end{equation}

Before proceeding we need the following claim.
\begin{claim}\label{cl:veryhelpful}
When $B_1>0, B_2>0$, the function 
\[f(x,y,z):=\frac{(B^2_1 x+2B_1 y+z)(x+2B_2 y+B^2_2 z)}{(B_1 x+(B_1B_2+1)y+B_2z)^2}, \]
defined for $x,y,z\geq 0$ such that at least one of $x,y,z$ is non-zero is upper and lower bounded by strictly positive constants depending only on $B_1,B_2$.
\end{claim}
\begin{proof}
Observe that for $t>0$ it holds that $f(x,y,z)=f(tx,ty,tz)$, so we may assume that $x+y+z=1$. Thus, $B^2_1 x+2B_1 y+z$, $x+2B_2 y+B^2_2 z$ and $B_1x+(B_1B_2+1)y+B_2z$ are convex combinations of $\{B^2_1,2B_1,1\},\{1,2B_2,B^2_2\}$ and $\{B_1,B_1B_2+1,B_2\}$ respectively. Since we assumed that $B_1,B_2>0$, $f(x,y,z)$ is bounded by strictly positive constants.
\end{proof}
Observe that Claim~\ref{cl:veryhelpful} is applicable to our case since not all of $R_1,R_2,R_4$ or $C_1,C_2, C_4$ can be equal to 0 (by equations \eqref{eq:frty}).

We next consider  the boundaries with $\gamma(\alpha-\gamma)(1-2\alpha+\gamma)\rightarrow 0^+$ and prove that the derivative $\frac{\partial \phi_2}{\partial \gamma}$  goes to infinity the right way, using the expression in \eqref{eq:dergammacopy}. The boundary cases with $\beta(\beta-\delta)(1-2\beta+\delta)\rightarrow 0^+$ can be treated analogously using \eqref{eq:derdeltacopy}. 

We consider cases.
\begin{enumerate}
\item $\gamma\rightarrow 0^+$. This case may only occur if $1-2\alpha\geq 0$. Then $\alpha-\gamma\rightarrow \alpha>0$ and $1-2\alpha+\gamma\rightarrow 1-2\alpha\geq 0$. By Claim~\ref{cl:veryhelpful} and \eqref{eq:dergammacopy}, we obtain $\frac{\partial \phi_2}{\partial \gamma}\rightarrow +\infty$.
\item $1-2\alpha+\gamma\rightarrow 0^+$. This case may only occur if $2\alpha-1\geq 0$. Then $\gamma\rightarrow 2\alpha-1\geq 0$ and $\alpha-\gamma\rightarrow 1-\alpha>0$. By Claim~\ref{cl:veryhelpful} and \eqref{eq:dergammacopy}, we obtain $\frac{\partial \phi_2}{\partial \gamma}\rightarrow +\infty$.
\item $\alpha-\gamma\rightarrow 0^+$. Then $\gamma\rightarrow\alpha>0$ and $1-2\alpha+\gamma\rightarrow 1-\alpha>0$. By Claim~\ref{cl:veryhelpful} and \eqref{eq:dergammacopy}, we obtain $\frac{\partial \phi_2}{\partial \gamma}\rightarrow -\infty$.
\end{enumerate}
This concludes the proof of Lemma~\ref{lem:secboundary}.
\end{proof}

\begin{proof}[Proof of Lemma~\ref{lem:sechessian}.]
For $(\alpha,\beta)=(p^+,p^-)$, denote by $\mathbf{H_2}$ the Hessian of $\Phi_2$ evaluated at $(\alpha^2,\beta^2,\mathbf{Y}^*)$. We keep only $\gamma,\delta,y_{ij}$ with $i,j=1,2,3$ as free variables in $\Phi_2$ (the rest of the $y_{ij}$ are substituted by \eqref{eq:yregion}). 

We use Sylvester's criterion to check that $-\mathbf{H_2}$ is positive definite. We check that the principal minors are positive, the first principal minor being the lower right corner of the matrix. This is equivalent to applying Sylvester's criterion in a rearranged version of $-\mathbf{H_2}$. Note that the first nine principal minors are guaranteed to be positive since they correspond to the Hessian of the convex function $-g_2(\mathbf{Y})$ evaluated at its minimum point. Thus we need to check only the remaining two principal minors.  We use Maple to perform the necessary computations and we plug into the resulting formulas the values of $p^\pm$ from \eqref{eq:pppm} and the $\mathbf{Y}^*$ as given in \eqref{eq:optimalxsecond}.

Denote by $P_i$ the determinant of the $i$-th leading principal minor. We have:
\begin{align*}
P_{10}&=\frac{\Delta^9 E^{22}_1E^2_2(B_2+Q^-)^2(1+B_1 Q^-)^2\big(1+(\Delta-1)\omega^2\big)}{(B_1 B2)^8\big(Q^-\big)^{14}\big(Q^{+}\big)^{12}},\\ P_{11}&=\mathrm{det}(-\mathbf{H_2})=\frac{\Delta^9 E_1^{26}E^2_2\big(1-(\Delta-1)^2\omega^2\big)}{(B_1 B2)^8\big(Q^-\big)^{14}\big(Q^{+}\big)^{14}},
\end{align*}
where $E_1,E_2$ and $\omega$ are given by \eqref{eq:helpfulAs} and \eqref{eq:definitionofomega}. Trivially $P_{10}>0$. And we also have $P_{11}>0$ by Lemma~\ref{lem:technicalinequality}.
\end{proof}

\subsection{Proof of Lemma~\ref{lem:secondmaxising}}\label{sec:secondmaxisingproof}
We now prove Condition~\ref{cond:maxima} for the Ising model without external field and for $d=2$, that is $B_1=B_2=B, \lambda=1$ and $\Delta=3$. See Section~\ref{sec:remarks} for the proof overview.

Our point of departure is \eqref{recurone}. As in the proof of Lemma~\ref{lem:secondmax2spin}, we may assume that $r_1r_4,c_1c_4$ are either both greater than 1 or both are less than 1. We are going to argue that \eqref{recurone} cannot hold in these cases. For $B_1=B_2=B$, multiplying the equations in \eqref{recurone} gives
\begin{multline*}\label{eq:stable}
(Bc_1+(B^2+1)+Bc_4)^2(Br_1+(B^2+1)+Br_4)^2=\\
\frac{(1-B^2)^2(c_1c_4-1)}{(c_1c_4)^{1/2}-1}\frac{(1-B^2)^2(r_1r_4-1)}{(r_1r_4)^{1/2}-1}.
\end{multline*}
We will prove that for $0<B<\frac{\sqrt{2}-1}{\sqrt{2}+1}$, when $r_1 r_4\neq 1$ (and similarly when $c_1 c_4\neq 1$)
\begin{equation}
\big(Br_1+(B^2+1)+Br_4\big)^2>\frac{(1-B^2)^2(r_1r_4-1)}{\sqrt{r_1r_4}-1}.\label{ineq:onea}
\end{equation}
In the region $0<B<\frac{\sqrt{2}-1}{\sqrt{2}+1}$ the values of $\alpha$ and $\beta$ are heavily biased towards 1 and 0 respectively, and this reflects at the values of $r_1,r_4,c_1,c_4$. The required bias of the values of $\alpha,\beta$ is captured by the following lemma.
\begin{lemma}\label{lem:biasdtwo}
Assume $0<B<\displaystyle\frac{\sqrt{2}-1}{\sqrt{2}+1}$. For $(\alpha,\beta)=(p^+,p^-)$ it holds that $\displaystyle\frac{\alpha}{1-\alpha}=\displaystyle\frac{1-\beta}{\beta}>\displaystyle\frac{4}{9}\cdot \frac{1}{B^3}$.
\end{lemma} 
\noindent The proof of Lemma~\ref{lem:biasdtwo} is given in Section~\ref{sec:biasdtwo}. We utilize Lemma~\ref{lem:biasdtwo} to prove: 
\begin{lemma}\label{lem:biasedvaluesdtwo}
Assume $0<B<\frac{\sqrt{2}-1}{\sqrt{2}+1}$. For $(\alpha,\beta)=(p^+,p^-)$ and $r_1,r_4,c_1,c_4$ satisfying the equations in \eqref{r1r4}, \eqref{c1c4}, it holds that
\begin{itemize}
\item $r_1>\displaystyle\frac{1}{3}\cdot\frac{1}{B^{2}}$ and $c_4>\displaystyle\frac{1}{3}\cdot\displaystyle\frac{1}{B^{2}}$.
\item If in addition $r_1 r_4,c_1 c_4>1$, then $r_1>\displaystyle\frac{4}{9}\cdot\frac{1}{B^{2}}$ and $c_4>\displaystyle\frac{4}{9}\cdot\displaystyle\frac{1}{B^{2}}$.
\end{itemize}
\end{lemma}
\noindent The proof of Lemma~\ref{lem:biasedvaluesdtwo} is given in Section~\ref{sec:biasedvaluesdtwo}. We now prove inequality \eqref{ineq:onea}. Rewrite \eqref{ineq:onea} as
\[B^2(r_1+r_4)^2+2B(B^2+1)(r_1+r_4)+(B^2+1)^2>(1-B^2)^2\left(1+\sqrt{r_1r_4}\right),\] or
\begin{equation}
B^2(r_1+r_4)^2+2B(B^2+1)(r_1+r_4)+4B^2>(1-B^2)^2\sqrt{r_1r_4}.\label{ineq:equivdtwo}
\end{equation}
At this point, we split the analysis into the cases $r_1r_4>1$ and $r_1r_4<1$, the first being considerably harder.
\begin{claim}\label{claim:proofdtwoa}
Inequality \eqref{ineq:equivdtwo} holds when $r_1r_4<1$.
\end{claim}
\begin{proof}
We prove the stricter inequality (since $1>B>0$ and $r_1,r_4>0$)
\[2B r_1>\sqrt{r_1r_4}.\]
By Item 1 of Lemma~\ref{lem:biasedvaluesdtwo} and the assumption $r_1r_4<1$, we have
\[2Br_1>\frac{2}{3}\cdot\frac{1}{B}>1>\sqrt{r_1 r_4}.\]
\end{proof}
\begin{claim}\label{claim:proofdtwo}
Inequality \eqref{ineq:equivdtwo} holds when $r_1r_4>1$.
\end{claim}
\begin{proof}
We prove the stricter inequality (since $1>B>0$ and $r_1,r_4>0$)
\begin{equation}
B^2(r_1+r_4)^2>\sqrt{r_1r_4}.\label{ineq:maindtwo}
\end{equation}
We first identify the regions where \eqref{ineq:maindtwo} is hard to prove. Note that $B^2(r^2_1+2r_1r_4+r^2_4)\geq4B^2r_1r_4$, so that \eqref{ineq:maindtwo} is true if $4B^2r_1r_4>\sqrt{r_1r_4}$. Thus we may assume that
\begin{equation}
\frac{1}{16B^4}\geq r_1 r_4.\label{rprodupper}
\end{equation}
Since $r_1r_4>1$, by Item 2 of Lemma~\ref{lem:biasedvaluesdtwo}
\begin{equation}
B^2 r_1^2>\frac{4^2}{9^2B^{2}},\label{ineq:ronelow}
\end{equation}
so that \eqref{ineq:maindtwo} is true if $\frac{4^2}{9^2B^{2}}\geq\sqrt{r_1r_4}$. Thus we may assume that
\begin{equation}
r_1r_4>\frac{4^4}{9^4B^4}.\label{rprodlower}
\end{equation}
We are now ready to prove \eqref{ineq:maindtwo}. Using \eqref{rprodupper}, \eqref{ineq:ronelow} and \eqref{rprodlower} we obtain
\begin{align*}
B^2(r^2_1+2r_1r_4+r^2_4)\geq B^2(r^2_1+2r_1r_4)>\frac{1}{B^2}\left(\frac{4^2}{9^2}+2\cdot \frac{4^4}{9^4}\right)>\frac{1}{B^2}\cdot \frac{1}{4}\geq \sqrt{r_1r_4}.
\end{align*}
\end{proof}

\begin{proof}[Proof of Lemma~\ref{lem:secondmaxising}.]
The arguments in this section established the required analogue of Lemma~\ref{lem:seccritical}, i.e., that the only critical points of $\phi_2(\gamma,\delta)$ satisfy $\gamma=\alpha^2,\ \delta=\beta^2$. Now, we can conclude using Lemmas~\ref{lem:secboundary} and~\ref{lem:sechessian}, just as in the proof of Lemma~\ref{lem:secondmax2spin} in Section~\ref{sec:logsecmax}.
\end{proof}

\subsubsection{Proof of Lemma~\ref{lem:biasdtwo}}\label{sec:biasdtwo}
By the remarks in the end of Section~\ref{sec:tree-recursions}, for the Ising model with no external field it holds that $\frac{\alpha}{1-\alpha}=\frac{1-\beta}{\beta}$ and $Q^+Q^-=1$. For ease of presentation, let $x=Q^+, y=Q^-$. By \eqref{eq:densitiesone} and \eqref{eq:pppm}, we have 
\[\frac{\alpha}{1-\alpha}=\frac{x(By+1)}{y+B},\mbox{ where } x=\left(\frac{By+1}{y+B}\right)^2 \mbox{ and } x>1>y>0.\]
Since $xy=1$, we have that
\begin{equation}\label{eq:algebraicbound}
\frac{\alpha}{1-\alpha}=\frac{By+1}{y(y+B)},\mbox{ where } \frac{1}{y^3}=\left(\frac{By+1}{y(y+B)}\right)^2\mbox{ and } 0<y<1.
\end{equation}
\begin{proof}[Proof of Lemma~\ref{lem:biasdtwo}.]
For $0<B<\frac{\sqrt{2}-1}{\sqrt{2}+1}=3-2\sqrt{2}$, we want to prove that $\displaystyle\frac{\alpha}{1-\alpha}>\displaystyle\frac{4}{9}\cdot\frac{1}{B^3}$, which in light of \eqref{eq:algebraicbound}, is equivalent to $y<(3/2)\sqrt[3]{3/2} B^2$. For convenience, let $c:= (3/2)\sqrt[3]{3/2}$. The equation for $y$ in \eqref{eq:algebraicbound} may be rewritten as
\[(y+B)^2=y(By+1)^2,\]
which in turns yields that 
\[B^2 y^2+(B^2+2B-1) y+B^2=0,\]
since, by \eqref{eq:algebraicbound}, we are interested in $0<y<1$. 

Let $f(z):= B^2 z^2+(B^2+2B-1) z+B^2$. By Lemma~\ref{eq:densitiesone}, we know that $y$ is the unique root of $f(z)$ for $0<z<1$. Observe that $f(0)=B^2>0$, so the desired inequality will follow if we prove that $f(B^2c)<0$. Now, $f$ may be rewritten as
\[f(z)=B^2z\left(z+\frac{1}{z}-\frac{1-2 B - B^2}{B^2}\right).\]
For $z=B^2c$ and $0<B<3-2\sqrt{2}$, we have
\begin{gather*}
z+\frac{1}{z}=\frac{B^4 c^2+1}{B^2c}\leq \frac{(3-2\sqrt{2})^4c^2+1}{c}\cdot \frac{1}{B^2},\\
\frac{1-2 B - B^2}{B^2}\geq \frac{1-2(3-2\sqrt{2})-(3-2\sqrt{2})^2}{B^2}=\frac{16\sqrt{2}-22}{B^2}.
\end{gather*}
A straightforward calculation gives $\frac{(3-2\sqrt{2})^4c^2+1}{c}<\frac{3}{5}<16\sqrt{2}-22$, thus proving $f(B^2c)<0$. The lemma follows.
\end{proof}

\subsubsection{Proof of Lemma~\ref{lem:biasedvaluesdtwo}}\label{sec:biasedvaluesdtwo}

We need the following  lemma.
\begin{lemma}\label{lem:boundsalphabetaRC}
When $1>B_1B_2>0$ and $r_1,r_4,c_1,c_4$ satisfy \eqref{r1r4} and \eqref{c1c4}, it holds that 
\[\frac{\alpha}{1-\alpha}<\frac{1+r_1/B_2}{1+B_2r_4}, \qquad \frac{1+B_1c_1}{1+c_4/B_1}<\frac{\beta}{1-\beta}.\]
\end{lemma}
Assuming Lemma~\ref{lem:boundsalphabetaRC} for the moment, we give the proof of Lemma~\ref{lem:biasedvaluesdtwo}. 
\begin{proof}[Proof of Lemma~\ref{lem:biasedvaluesdtwo}.]
When $B_1=B_2=B$, Lemma~\ref{lem:boundsalphabetaRC} gives
\[\frac{\alpha}{1-\alpha}<\frac{1+r_1/B}{1+Br_4},\qquad \frac{1+Bc_1}{1+c_4/B}<\frac{\beta}{1-\beta}.\]
Since $r_4,c_1>0$, it follows that
\[r_1>B\left(\frac{\alpha}{1-\alpha}-1\right)>B\left(\frac{4}{9}\cdot \frac{1}{B^3}-1\right) \mbox{ and } c_4>B\left(\frac{1-\beta}{\beta}-1\right)>B\left(\frac{4}{9}\cdot\frac{1}{B^3}-1\right),\]
where we have used Lemma~\ref{lem:biasdtwo} to bound $\frac{\alpha}{1-\alpha}$ and $\frac{1-\beta}{\beta}$. This proves the first item of the lemma, after observing that $\frac{4}{9}\cdot \frac{1}{B^3}-1>\frac{1}{3}\cdot \frac{1}{B^3}$ for any $B<1/3$.

To prove the second item of the Lemma, note that when $r_1r_4,c_1c_4>1$ it holds
\[\frac{1+r_1/B}{1+Br_4}<\frac{1+r_1/B}{1+B/r_1}=\frac{r_1}{B}\mbox{ and }  \frac{1+Bc_1}{1+c_4/B}>\frac{1+B/c_4}{1+c_4/B}=\frac{B}{ c_4}.\]
It follows that
\[ r_1>B\cdot \frac{\alpha}{1-\alpha} \mbox{ and } c_4>B\cdot \frac{1-\beta}{\beta}.\]
The desired bounds now follow after using Lemma~\ref{lem:biasdtwo} to bound $\frac{\alpha}{1-\alpha}$ and $\frac{1-\beta}{\beta}$.
\end{proof}

We next give the proof of Lemma~\ref{lem:boundsalphabetaRC}.
\begin{proof}[Proof of Lemma~\ref{lem:boundsalphabetaRC}.]
Define
\begin{equation*}
\begin{array}{lll}
C_1'=\displaystyle\frac{B^2_1c_1+2B_1+c_4}{B_1c_1+(B_1B_2+1)+B_2c_4}, &  & C_4'=\displaystyle\frac{c_1+2B_2+B^2_2c_4}{B_1c_1+(B_1B_2+1)+B_2c_4},\\
& &\\
R_1'=\displaystyle\frac{B^2_1r_1+2B_1+r_4}{B_1r_1+(B_1B_2+1)+B_2r_4}, &  & R_4'=\displaystyle\frac{r_1+2B_2+B^2_2r_4}{B_1r_1+(B_1B_2+1)+B_2r_4}.
\end{array}
\end{equation*}
Using $C_1',C_4',R_1',R_4'$ we can rewrite the equations in \eqref{r1r4}, \eqref{c1c4} as
\begin{gather}
r_1C_1'=\frac{\gamma}{\alpha-\gamma},\qquad r_4C_4'=\frac{1-2\alpha+\gamma}{\alpha-\gamma},\label{eq:rs}\\
c_1R_1'=\frac{\delta}{\beta-\delta},\qquad c_4R_4'=\frac{1-2\beta+\delta}{\beta-\delta}.\label{eq:cs}
\end{gather}
Eliminating $\gamma,\delta$ from \eqref{eq:rs} and \eqref{eq:cs} respectively, we obtain
\begin{equation}
\frac{\alpha}{1-\alpha}=\frac{1+r_1C_1'}{1+r_4C_4'},\qquad \frac{\beta}{1-\beta}=\frac{1+c_1R_1'}{1+c_4R_4'}.\label{eq:mainrcRC}
\end{equation}
We have the following loose bounds on $R_1',R_4',C_1',C_4'$.
\begin{equation}\label{cl:boundRC}
B_1<R_1', \qquad R_4'<\displaystyle\frac{1}{B_1},\qquad C_1'<\displaystyle\frac{1}{B_2}, \qquad B_2<C_4'.
\end{equation}
These bounds may be proved by picking any of the inequalities, multiplying out and using $B_1B_2<1$. It is immediate now to check that the bounds in Lemma~\ref{lem:boundsalphabetaRC} follow after combining \eqref{eq:mainrcRC} and \eqref{cl:boundRC}.
\end{proof}

\subsection{Proof of Lemma~\ref{lem:secondmaxhardcore}}\label{sec:hardhard}
To prove Lemma~\ref{lem:secondmaxhardcore}, we will need to establish the following lemma.
\begin{lemma}\label{lem:hardcoreqwe}
For the hard-core model (that is, $B_{1}=0$, $B_{2}=1$) for $\Delta=3,4,5$ and 
$(\alpha,\beta)=(p^+,p^-)$ the solution of \eqref{es1}, \eqref{es2} and \eqref{et} 
satisfies $\gamma=\alpha^2$ and $\delta=\beta^2$.
\end{lemma}
Assuming for now Lemma~\ref{lem:hardcoreqwe}, let us conclude Lemma~\ref{lem:secondmaxhardcore}.
\begin{proof}[Proof of Lemma~\ref{lem:secondmaxhardcore}.]
By Lemma~\ref{lem:hardcoreqwe}, the only critical points of $\phi_2(\gamma,\delta)$ satisfy $\gamma=\alpha^2,\ \delta=\beta^2$. The result follows now by using Lemmas~\ref{lem:secboundary} and~\ref{lem:sechessian}, just as in the proof of Lemma~\ref{lem:secondmax2spin} in Section~\ref{sec:logsecmax}.
\end{proof}

Finally, we give the proof of Lemma~\ref{lem:hardcoreqwe}.
\begin{proof}[Proof of Lemma~\ref{lem:hardcoreqwe}.]
Let $d:=\Delta-1$ and let $c_1:=C_1/C_2$, $c_4:=C_4/C_2$,
$r_1:=R_1/R_2$, $r_4:=R_4/R_2$
(the same notation as we used in the proof of
Lemma~\ref{lem:seccritical}). Equations in \eqref{recurone}, for $B_1=0$ and $B_2=1$, become
\begin{equation}\label{ettt}
(c_1c_4)^{1/d}-1=\frac{r_1r_4-1}{(1+r_4)^2}\quad\mbox{and}\quad
(r_1r_4)^{1/d}-1=\frac{c_1c_4-1}{(1+c_4)^2}.
\end{equation}
Let $x=(r_1r_4)^{1/d}$, $y=(c_1c_4)^{1/d}$, $a=1/r_4$, and $b=1/c_4$.
We can rewrite equation~\eqref{ettt}
as follows
\begin{equation}\label{xxxyy}
y=\frac{1+2a+a^2 x^d}{1+2a+a^2}\quad\mbox{and}\quad x=\frac{1+2b+b^2
y^d}{1+2b+b^2}.
\end{equation}
Solving for $a$ and $b$ we obtain
\begin{equation}\label{barn}
a=\frac{y-1 + q_a}{x^d-y}\quad\mbox{and}\quad b=\frac{x-1 + q_b}{y^d-x},
\end{equation}
where
\begin{equation}\label{barn2}
q_a^2=(y-1)(x^d-1)\quad\mbox{and}\quad q_b^2 =(x-1)(y^d-1).
\end{equation}
Note that $a>0$ and $b>0$. Assume $y>1$. Then $x>1$ (otherwise the
sides of~\eqref{ettt} would
have different signs) and also $y\leq x^d$ (the right hand side of
\eqref{xxxyy} is a convex combination of $1$ and $x^d$) and hence
$|q_a|>y-1$. Thus, using $a>0$ in Equation~\eqref{barn}, we obtain
$q_a>0$ and similarly $q_b>0$. Now assume $y<1$. Then
$x<1$ and also $y\geq x^d$ (again, the right hand side of
\eqref{xxxyy} is a convex combination of $1$ and $x^d$) and hence
$|q_a|>1-y$. Thus, using $a>0$ in Equation~\eqref{barn}, we obtain
$q_a<0$ and similarly $q_b<0$.

Because of symmetry between $x$ and $y$ we only need to consider two cases
(in the case $x=y=1$ we have $\gamma=\alpha^2$ and $\delta=\beta^2$):
\begin{itemize}
\item CASE 1: $1<y\leq x\leq y^d$, $q_a:=\sqrt{(y-1)(x^d-1)}$,
$q_b:=\sqrt{(x-1)(y^d-1)}$, and
\item CASE 2: $0<y^d\leq x\leq y<1$, $q_a:=-\sqrt{(y-1)(x^d-1)}$,
$q_b:=-\sqrt{(x-1)(y^d-1)}$.
\end{itemize}

From Equation~\eqref{et} one obtains the following expressions for
$\alpha$ and $\beta$ (see equation \eqref{eq:mainrcRC})
$$
\frac{\alpha}{1-\alpha}=\frac{\frac{c_4}{1+c_4}r_1+1}{\frac{c_1+2+c_4}{1+c_4}r_4+1}\quad\mbox{and}\quad
\frac{\beta}{1-\beta}=\frac{\frac{r_4}{1+r_4}c_1+1}{\frac{r_1+2+r_4}{1+r_4}c_4+1}.
$$
Solving for $\alpha,\beta$ and using the parametrization with $x,y,a,b$ we have
\begin{equation}\label{alphabeta2}
\begin{array}{l}
\alpha=\displaystyle\frac{y(1+a)^2-1-a+ab}{x(1+b)^2+y(1+a)^2+2ab-1},\vspace{0.2cm}\\
\beta=\displaystyle\frac{x(1+b)^2-1-b+ab}{x(1+b)^2+y(1+a)^2+2ab-1}.
\end{array}
\end{equation}
Equations \eqref{eq:recurone}, \eqref{eq:pppm} for $B_1=0,\,B_2=1,\, (\alpha,\beta)=(p^+,p^-)$ easily imply
\begin{equation}\label{hardcore3}
\alpha (1-\alpha)^d = \beta (1-\beta)^d.
\end{equation}
Equation~\eqref{hardcore3} is equivalent to (assuming $\alpha\neq\beta$)
\begin{equation}\label{hardcore2}
\frac{\alpha (1-\alpha)^d - \beta
(1-\beta)^d}{\alpha-\beta}=0,
\end{equation}
which using the expressions for $\alpha$ and $\beta$
from~\eqref{alphabeta2} becomes
\begin{equation}\label{eeer}
\textstyle\frac{(y(1+a)^2-1-a+ab)(x(1+b)^2+ab+a)^d -
(x(1+b)^2-1-b+ab)(y(1+a)^2+ab+b)^d)}{x(b+1)^2-y(a+1)^2+a-b}=0.
\end{equation}
The final part of the proof will require some computational assistance
(just to manipulate polynomials). First we are going to
plug-in the expression for $a$ and $b$ from~\eqref{barn}
into~\eqref{eeer}. Then we are going to use the expressions
in~\eqref{barn2}
to reduce the powers of $q_a$ and $q_b$ occurring in the expression
obtaining an expression of the form
$$
c_{00} + c_{10} q_a + c_{01} q_b + c_{11} q_a q_b,
$$
where $c_{00},c_{01},c_{10},c_{11}$ are polynomials in $x$ and $y$.
Then we are going to show that in CASE 1 we have $c_{00}>0$,
$c_{01}>0$, $c_{10}>0$, $c_{11}>0$
and in CASE 2 we have $c_{00}>0$, $c_{01}<0$, $c_{10}<0$, $c_{11}>0$
(and hence~\eqref{eeer} cannot be zero) . This will be accomplished by
reparameterizing $x=(ty + y^d)/(t+1)$, factoring the expressions,
and observing that:
\begin{itemize}
\item factor $y-1$ occurs with even power in $c_{00}$ and $c_{11}$,
\item factor $y-1$ occurs with odd power in $c_{01}$ and $c_{10}$, and
\item all the other factors have all coefficients positive.
\end{itemize}
The details of the argument appear in Appendix~\ref{sec:appendixhardcore}. 
\end{proof}

\newpage
  
  \appendix
  
 \section{Computer Assisted Proofs for the Hardcore
Model}\label{sec:appendixhardcore}

\subsection{Case $\Delta=3$}

\begin{verbatim}
d = 2;
F = Factor[((y*(1 + a)^2 - 1 - a + a*b)*(x*(1 + b)^2 + a*b + a)^d -
         (x*(1 + b)^2 - 1 - b + a*b)*(y*(1 + a)^2 + a*b + b)^d)][[3]];
\end{verbatim}

\noindent {\bf Comment:} now $F$ contains the left-hand side of \eqref{eeer}
(we took the third factor; the other two factors are $-1$ and the
denominator of~\eqref{eeer}).

\begin{verbatim}
H = Factor[Expand[F /.
  {a -> (-1 + y + Qa)/(x^d - y), b -> (-1 + x + Qb)/(y^d - x)}]][[3]];
\end{verbatim}

\noindent {\bf Comment:} now $H$ contains the left-hand side of \eqref{eeer}
after substituting the values of $a,b$ given by~\eqref{barn} (the
result has three factors: $1/(x-y^d)^{2d-1}$, $1/(x^d-y)^{2d-1}$, and the
factor we assigned to $H$; note that the first two factors
cannot have value zero).

\begin{verbatim}
da = Exponent[H, Qa];
For [i = da, i >= 2, i--,
   H = Expand[H /. {Qa^i -> Qa^(i - 2)*(1 - x^d - y + x^d*y)}]];

db = Exponent[H, Qb];
For [i = db, i >= 2, i--,
   H = Expand[H /. {Qb^i -> Qb^(i - 2)*(1 - y^d - x + y^d*x)}]];
\end{verbatim}

\noindent {\bf Comment:} now $H$ contains the left-hand side of \eqref{eeer}
(multiplied by $(x-y^d)^{2d-1}(x^d-y)^{2d-1}$) after reducing the powers
of $q_a$ and $q_b$ using~\eqref{barn2}.

\begin{verbatim}
c00 = Factor[Coefficient[Coefficient[H, Qa, 0], Qb, 0]];
c01 = Factor[Coefficient[Coefficient[H, Qa, 0], Qb, 1]];
c10 = Factor[Coefficient[Coefficient[H, Qa, 1], Qb, 0]];
c11 = Factor[Coefficient[Coefficient[H, Qa, 1], Qb, 1]];
\end{verbatim}

\noindent {\bf Comment:} the following reparameterization will reveal the signs
of $c_{00}, c_{01}, c_{10}, c_{11}$. We show the expressions for $d=2$; for
larger $d$ we will print the first $6$ factors and check the positivity of the
last factor by Mathematica.

\begin{verbatim}
u00 = Factor[Expand[c00 /. {x -> (t*y + y^d )/(t + 1)}]]
u01 = Factor[Expand[c01 /. {x -> (t*y + y^d )/(t + 1)}]]
u10 = Factor[Expand[c10 /. {x -> (t*y + y^d )/(t + 1)}]]
u11 = Factor[Expand[c11 /. {x -> (t*y + y^d )/(t + 1)}]]
\end{verbatim}

\noindent {\bf OUTPUT:}

\begin{align*}
&-\frac{1}{(1+t)^9}(-1+y)^6 y^2 (1+t+y) \Big(4 t+28 t^2+84 t^3+140
t^4+140 t^5+84 t^6+28 t^7+4 t^8+12 t y\\
&+84 t^2 y+248 t^3 y+400 t^4 y+380 t^5 y+212 t^6 y+64 t^7 y+8 t^8 y+4
y^2+50 t y^2+228 t^2 y^2+532 t^3 y^2\\
&+714 t^4 y^2+570 t^5 y^2+264 t^6 y^2+64 t^7 y^2+6 t^8 y^2+16 y^3+141
t y^3+502t^2 y^3+954 t^3 y^3+1054 t^4 y^3\\
&+683 t^5 y^3+248 t^6 y^3+46 t^7 y^3+4 t^8 y^3+32 y^4+233 t y^4+694
t^2 y^4+1089 t^3 y^4+976 t^4 y^4+521 t^5 y^4\\
&+179t^6 y^4+43 t^7 y^4+5 t^8 y^4+40 y^5+243 t y^5+592 t^2 y^5+759 t^3
y^5+600 t^4 y^5+352 t^5 y^5+161 t^6 y^5\\
&+43 t^7 y^5+4 t^8 y^5+32 y^6+159 t y^6+318 t^2 y^6+387 t^3 y^6+385
t^4 y^6+302 t^5 y^6+139 t^6 y^6+26 t^7 y^6\\
&+t^8 y^6+16 y^7+63 t y^7+132 t^2 y^7+241 t^3 y^7+322 t^4 y^7+236 t^5
y^7+72 t^6 y^7+6 t^7 y^7+4 y^8+19 t y^8\\
&+82 t^2 y^8+201 t^3 y^8+234 t^4 y^8+110 t^5 y^8+15 t^6 y^8+12 t
y^9+69 t^2 y^9+137 t^3 y^9+100 t^4 y^9 \\
&+20 t^5 y^9+10 t y^{10}+44 t^2 y^{10}+54 t^3 y^{10}+15 t^4 y^{10}+6 t
y^{11}+16 t^2 y^{11}+6 t^3 y^{11}+2 t y^{12}+t^2 y^{12}\Big)
\end{align*}

\begin{align*}
&-\frac{1}{(1+t)^8}(-1+y)^5 y^2 (1+t+y) \Big(4 t+24 t^2+60 t^3+80
t^4+60 t^5+24 t^6+4 t^7+8 t y+50 t^2 y\\
&+126 t^3 y+164 t^4 y+116 t^5 y+42 t^6 y+6 t^7 y+4 y^2+38 t y^2+136
t^2 y^2+248 t^3 y^2+254 t^4 y^2+148 t^5 y^2\\
&+46 t^6 y^2+6 t^7 y^2+12 y^3+89 t y^3+258 t^2 y^3+383 t^3 y^3+311 t^4
y^3+138 t^5 y^3+35 t^6 y^3+6 t^7 y^3\\
&+20 y^4+118 t y^4+274 t^2 y^4+312 t^3 y^4+197 t^4 y^4+92 t^5 y^4+36
t^6 y^4+5 t^7 y^4+20 y^5+95 t y^5+165 t^2 y^5+\\
&156 t^3 y^5+127 t^4 y^5+87 t^5 y^5+29 t^6 y^5+2 t^7 y^5+12 y^6+42 t
y^6+65 t^2 y^6+94 t^3 y^6+113 t^4 y^6\\
&+70 t^5 y^6+12 t^6 y^6+4 y^7+11 t y^7+36 t^2 y^7+85 t^3 y^7+90 t^4
y^7+30 t^5 y^7+6 t y^8+35 t^2 y^8\\
&+65 t^3 y^8+40 t^4 y^8+6 t y^9+25 t^2 y^9+30 t^3 y^9+4 t y^{10}+12
t^2 y^{10}+2 t y^{11}\Big)
\end{align*}

\begin{align*}
&-\frac{1}{(1+t)^7}(-1+y)^5 y^2 (1+t+y) \Big(4 t+20 t^2+40 t^3+40
t^4+20 t^5+4 t^6+10 t y+50 t^2 y+96 t^3 y\\
&+88 t^4 y+38 t^5 y+6 t^6 y+4 y^2+38 t y^2+120 t^2 y^2+174 t^3 y^2+126
t^4 y^2+44 t^5 y^2+6 t^6 y^2+16 y^3 \\
&+107 t y^3+267 t^2 y^3+324 t^3 y^3+197 t^4 y^3+55 t^5 y^3+6 t^6
y^3+32 y^4+177 t y^4+371 t^2 y^4+366 t^3 y^4\\
&+179 t^4 y^4+47 t^5 y^4+5 t^6 y^4+40 y^5+183 t y^5+310 t^2 y^5+250
t^3 y^5+118 t^4 y^5+28 t^5 y^5+2 t^6 y^5\\
&+ 32y^6+119 t y^6+164 t^2 y^6+132 t^3 y^6+56 t^4 y^6+8 t^5 y^6+16
y^7+47 t y^7+68 t^2 y^7+52 t^3 y^7+12 t^4 y^7\\
&+4 y^8+13 t y^8+23 t^2 y^8+8 t^3 y^8+4 t y^9+2 t^2 y^9\Big)
\end{align*}

\begin{align*}
&-\frac{1}{(1+t)^6}(-1+y)^4 y^2 \Big(4 t+20 t^2+40 t^3+40 t^4+20 t^5+4
t^6+10 t y+48 t^2 y+88 t^3 y+76 t^4 y\\
& +30 t^5 y+4 t^6 y+4 y^2+38 t y^2+117 t^2 y^2+164 t^3 y^2+116 t^4
y^2+42 t^5 y^2+7 t^6 y^2+16 y^3+105 t y^3\\
&+253 t^2 y^3+294 t^3 y^3+173 t^4 y^3+49 t^5 y^3+6 t^6 y^3+32 y^4+171
t y^4+339 t^2 y^4+310 t^3 y^4+139 t^4 y^4\\
&+39 t^5 y^4+4 t^6 y^4+40 y^5+173 t y^5+268 t^2 y^5+192 t^3 y^5+92 t^4
y^5+20 t^5 y^5+32 y^6+109 t y^6+132 t^2 y^6\\
&+102 t^3 y^6+40 t^4 y^6+16 y^7+41 t y^7+54 t^2 y^7+40 t^3 y^7+4
y^8+11 t y^8+20 t^2 y^8+4 t y^9\Big)
\end{align*}

\subsection{Case $\Delta=4$}

\begin{verbatim}
d = 3;
F = Factor[((y*(1 + a)^2 - 1 - a + a*b)*(x*(1 + b)^2 + a*b + a)^d -
  (x*(1 + b)^2 - 1 - b + a*b)*(y*(1 + a)^2 + a*b + b)^d)][[3]];

H = Factor[Expand[F /.
  {a -> (-1 + y + Qa)/(x^d - y), b -> (-1 + x + Qb)/(y^d - x)}]][[3]];

da = Exponent[H, Qa];
For [i = da, i >= 2, i--,
 H = Expand[H /. {Qa^i -> Qa^(i - 2)*(1 - x^d - y + x^d*y)}]];
db = Exponent[H, Qb];
For [i = db, i >= 2, i--,
 H = Expand[H /. {Qb^i -> Qb^(i - 2)*(1 - y^d - x + y^d*x)}]];

c00 = Factor[Coefficient[Coefficient[H, Qa, 0], Qb, 0]];
c01 = Factor[Coefficient[Coefficient[H, Qa, 0], Qb, 1]];
c10 = Factor[Coefficient[Coefficient[H, Qa, 1], Qb, 0]];
c11 = Factor[Coefficient[Coefficient[H, Qa, 1], Qb, 1]];

u00 = Factor[Expand[c00 /. {x -> (t*y + y^d )/(t + 1)}]]; Length[u00];
u01 = Factor[Expand[c01 /. {x -> (t*y + y^d )/(t + 1)}]]; Length[u01];
u10 = Factor[Expand[c10 /. {x -> (t*y + y^d )/(t + 1)}]]; Length[u10];
u11 = Factor[Expand[c11 /. {x -> (t*y + y^d )/(t + 1)}]]; Length[u11];
\end{verbatim}

\noindent {\bf Comment:} the following code checks the positivity of the
coefficients of the last factor of $c_{00}, c_{01}, c_{10}, c_{11}$.

\begin{verbatim}
BAD = False;
u00[[1]]*u00[[2]]*u00[[3]]*u00[[4]]*u00[[5]]*u00[[6]]
For[i = 1, i <= Length[u00[[7]]], i++,
 If[ (u00[[7]][[i]] /. {T -> 1, y -> 1}) < 0, BAD = True]];
u01[[1]]*u01[[2]]*u01[[3]]*u01[[4]]*u01[[5]]*u01[[6]]
For[i = 1, i <= Length[u01[[7]]], i++,
 If[ (u01[[7]][[i]] /. {T -> 1, y -> 1}) < 0, BAD = True]];
u10[[1]]*u10[[2]]*u10[[3]]*u10[[4]]*u10[[5]]*u10[[6]]
For[i = 1, i <= Length[u10[[7]]], i++,
 If[ (u10[[7]][[i]] /. {T -> 1, y -> 1}) < 0, BAD = True]];
u11[[1]]*u11[[2]]*u11[[3]]*u11[[4]]*u11[[5]]*u11[[6]]
For[i = 1, i <= Length[u11[[7]]], i++,
 If[ (u11[[7]][[i]] /. {T -> 1, y -> 1}) < 0, BAD = True]];
Print[BAD];
\end{verbatim}

\noindent {\bf OUTPUT:}

\[-\frac{(-1+y)^{10} y^4 (1+y)^4 \left(1+t+y+y^2\right)^2}{(1+t)^{20}}\]

\[-\frac{(-1+y)^9 y^4 (1+y)^4 \left(1+t+y+y^2\right)}{(1+t)^{19}}\]

\[-\frac{(-1+y)^9 y^4 (1+y)^4 \left(1+t+y+y^2\right)}{(1+t)^{17}}\]

\[-\frac{(-1+y)^8 y^4 (1+y)^4 \left(1+t+y+y^2\right)}{(1+t)^{16}}\]

\noindent\(\text{False}\)

\subsection{Case $\Delta=5$}

\begin{verbatim}
d = 4;
F = Factor[((y*(1 + a)^2 - 1 - a + a*b)*(x*(1 + b)^2 + a*b + a)^d -
  (x*(1 + b)^2 - 1 - b + a*b)*(y*(1 + a)^2 + a*b + b)^d)][[3]];

H = Factor[Expand[F /.
  {a -> (-1 + y + Qa)/(x^d - y),b -> (-1 + x + Qb)/(y^d - x)}]][[3]];

da = Exponent[H, Qa];
For [i = da, i >= 2, i--,
 H = Expand[H /. {Qa^i -> Qa^(i - 2)*(1 - x^d - y + x^d*y)}]];
db = Exponent[H, Qb];
For [i = db, i >= 2, i--,
 H = Expand[H /. {Qb^i -> Qb^(i - 2)*(1 - y^d - x + y^d*x)}]];

c00 = Factor[Coefficient[Coefficient[H, Qa, 0], Qb, 0]];
c01 = Factor[Coefficient[Coefficient[H, Qa, 0], Qb, 1]];
c10 = Factor[Coefficient[Coefficient[H, Qa, 1], Qb, 0]];
c11 = Factor[Coefficient[Coefficient[H, Qa, 1], Qb, 1]];

u00 = Factor[Expand[c00 /. {x -> (t*y + y^d )/(t + 1)}]]; Length[u00];
u01 = Factor[Expand[c01 /. {x -> (t*y + y^d )/(t + 1)}]]; Length[u01];
u10 = Factor[Expand[c10 /. {x -> (t*y + y^d )/(t + 1)}]]; Length[u10];
u11 = Factor[Expand[c11 /. {x -> (t*y + y^d )/(t + 1)}]]; Length[u11];

(* proof for Delta=5, d=4 *)
BAD = False;
u00[[1]]*u00[[2]]*u00[[3]]*u00[[4]]*u00[[5]]*u00[[6]]
For[i = 1, i <= Length[u00[[7]]], i++,
 If[ (u00[[7]][[i]] /. {T -> 1, y -> 1}) < 0, BAD = True]];
u01[[1]]*u01[[2]]*u01[[3]]*u01[[4]]*u01[[5]]*u01[[6]]
For[i = 1, i <= Length[u01[[7]]], i++,
 If[ (u01[[7]][[i]] /. {T -> 1, y -> 1}) < 0, BAD = True]];
u10[[1]]*u10[[2]]*u10[[3]]*u10[[4]]*u10[[5]]*u10[[6]]
For[i = 1, i <= Length[u10[[7]]], i++,
 If[ (u10[[7]][[i]] /. {T -> 1, y -> 1}) < 0, BAD = True]];
u11[[1]]*u11[[2]]*u11[[3]]*u11[[4]]*u11[[5]]*u11[[6]]
For[i = 1, i <= Length[u11[[7]]], i++,
 If[ (u11[[7]][[i]] /. {T -> 1, y -> 1}) < 0, BAD = True]];
Print[BAD];
\end{verbatim}
\enlargethispage{\baselineskip}
\noindent {\bf OUTPUT:}

\[-\frac{(-1+y)^{14} y^6 \left(1+y+y^2\right)^6
\left(1+t+y+y^2+y^3\right)^2}{(1+t)^{35}}\]

\[-\frac{(-1+y)^{13} y^6 \left(1+y+y^2\right)^6
\left(1+t+y+y^2+y^3\right)^2}{(1+t)^{34}}\]

\[-\frac{(-1+y)^{13} y^6 \left(1+y+y^2\right)^6
\left(1+t+y+y^2+y^3\right)^2}{(1+t)^{31}}\]

\[-\frac{(-1+y)^{12} y^6 \left(1+y+y^2\right)^6
\left(1+t+y+y^2+y^3\right)}{(1+t)^{30}}\]

\noindent\(\text{False}\)

\end{document}